\documentclass[UKenglish,authorcolumns,thm-restate,a4paper]{lipics-v2019}
 \usepackage{amsmath}
\usepackage{amssymb}
\usepackage{amsthm}
\usepackage{mathtools}
\usepackage{wasysym}
\usepackage{hyperref}
\usepackage{thm-restate}
\usepackage{thmtools}
\usepackage{url}
\usepackage{booktabs}
\usepackage{subcaption}
\usepackage{cite}
\usepackage{color}
\usepackage{ marvosym }
\usepackage{tikz}
\usetikzlibrary{arrows}
\usepackage{listings}
\usepackage{multirow}
\usepackage{shuffle}
\DeclareFontFamily{U}{shuffle}{}
\DeclareFontShape{U}{shuffle}{m}{n}{ <-8>shuffle7 <8->shuffle10}{}
\usetikzlibrary{arrows,automata,shapes.gates.logic.US,shapes.gates.logic.IEC,calc}
\nolinenumbers

\usepackage[ruled,vlined]{algorithm2e}
\usepackage{subfiles}
\usepackage{hhline}
\usepackage[clock]{ifsym}
\usepackage{xspace}
\usepackage{wrapfig}
\usepackage{booktabs}
\usepackage{subcaption}
\usepackage{extarrows}
\usepackage{pdfcomment}
\usepackage{subfiles}
\usepackage{bbm}
\usepackage{tikzit}
\usepackage{fancyvrb}
\usepackage{mathpartir}

\usepackage[colorinlistoftodos,prependcaption]{todonotes}
\usepackage{xargs}
\usepackage{enumitem}
\usepackage{turnstile}
\usepackage{colortbl}
\usepackage{empheq}
\usepackage[most]{tcolorbox}
\usepackage{cancel}
\newtcbox{\mycode}[1][]{%
nobeforeafter, math upper, tcbox raise base,
enhanced, colframe=black,
colback=gray!5, boxrule=0.5pt,
#1}

\newcommand{\qbfsat}{\mathsf{QBF}}
\newcommand{\succsat}{\mathsf{SuccinctSAT}}
\newcommand{\pspace}{$\mathsf{PSPACE}$}
\newcommand{\nexp}{$\mathsf{NEXPTIME}$}
\newcommand{\expspace}{$\mathsf{EXPSpace}$}
\newcommand{\np}{$\mathsf{NP}$}
\newcommand{\ptime}{$\mathsf{PTIME}$}

\definecolor{auburn}{rgb}{0.43, 0.21, 0.1}
\definecolor{cadmiumgreen}{rgb}{0.0, 0.42, 0.24}

\newcommand{\abstcolor}{blue}
\newcommand{\leadercolor}{cadmiumgreen}

\newcommand{\bsp}{\mathsf{BSP}}

\newcommand{\bigO}[1]{\mathcal{O}(#1)}
\newcommand{\sizeOf}[1]{|#1|}
\newcommand{\cachepred}{\mathsf{CachePred}}
\newcommand{\freePred}{\mathsf{avail}}

\newcommand{\thisthreadt}{\thisthread\mathsf{T}}
\newcommand{\thatthreadt}{\thatthread\mathsf{T}}

\newcommand{\abstrun}{\rho^\abstscriptpshort}

\newcommand{\abstscript}{{\thisthreadshort\thatthreadshort}}

\newcommand{\abstscriptpshort}{{\thisthreadshort\thatthreadshort}}
\newcommand{\loadsym}{\mathsf{ld}}
\newcommand{\storesym}{\mathsf{st}}

\newcommand{\setvars}{\mathcal{X}}

\newcommand{\setvals}{\mathcal{D}}

\newcommand{\view}{\mathsf{view}}

\newcommand{\tmorph}{\mu}
\newcommand{\tmorphOf}[1]{\tmorph(#1)}

\newcommand{\quantity}{\mathcal{Q}}

\newcommand{\xvar}{\mathsf{x}}
\newcommand{\yvar}{\mathsf{y}}

\newcommand{\tstamps}[1]{\mathsf{TS}(#1)}

\newcommand{\last}{\mathsf{last}}
\newcommand{\lastOf}[1]{\last(#1)}
\newcommand{\firstOf}[1]{\mathsf{first}(#1)}

\newcommand{\abstfuncp}{\alpha_{\abstscriptpshort}}
\newcommand{\abstfuncpOf}[1]{\abstfuncp(#1)}

\newcommand{\tplus}[1]{#1^{\textbf{+}}}

\newcommand{\raisets}[1]{\mathsf{raise}(#1)}

\newcommand{\leader}{{\color{\leadercolor}\mathsf{ldr}}}

\newcommand{\essmessOf}[1]{\mathsf{edepend}(#1)}
\newcommand{\ethreadOf}[1]{\mathsf{ethread}(#1)}

\newcommand{\genproc}{\mathsf{genthread}}

\newcommand{\canread}{\mathsf{selfRead}}

\newcommand{\gentmorph}{\mathcal{M}}
\newcommand{\copythread}[1]{\mathsf{copy}(#1)}
\newcommand{\thisTS}[2]{{\color{blue} \text{ts}}^{#1}_#2}

\newcommand{\dep}{\mathsf{depend}}
\newcommand{\depgraph}{\mathit{G}}
\newcommand{\heightv}{\mathsf{height}}
\newcommand{\sinkv}{\mathsf{sink}}
\newcommand{\sourcev}{\mathsf{source}}

\newcommand{\pspacelibrary}{\com}
\newcommand{\funcAG}{\com_{\mathrm{AG}}}
\newcommand{\funcSATC}{\com_{\mathrm{SATC}}}
\newcommand{\funcFEC}[1]{\com_{\mathrm{FE[#1]}}}
\newcommand{\funcGC}{\com_{\mathrm{assert}}}

\newcommand{\bvars}{\mathsf{BVars}}

\newcommand{\va}{\texttt{var1}}
\newcommand{\vb}{\texttt{var2}}
\newcommand{\vc}{\texttt{var3}}
\newcommand{\siga}{\texttt{sig1}}
\newcommand{\sigb}{\texttt{sig2}}
\newcommand{\sigc}{\texttt{sig3}}
\newcommand{\red}[1]{{\color{red} #1}}
\newcommand{\blu}[1]{{\color{blue} #1}}

\newcommand{\pnode}{\textcolor{black}{\mathcal{A}}}
\newcommand{\ynode}{\textcolor{black}{\mathcal{B}}}

\newcommand{\issat}{\mathsf{IsSAT}}

\newcommand{\comlead}{\com_\leader}
\newcommand{\comcont}{\com_\thatthread}
\newcommand{\funcCLENC}{\com_{\mathsf{CL-ENC}}}
\newcommand{\funcSAT}{\com_{\mathsf{SAT}}}
\newcommand{\funcFC}[1]{\com_{\mathsf{Forall[#1]}}}

\newcommand{\valt}{\avalue_{\texttt{t}}}
\newcommand{\valf}{\avalue_{\texttt{f}}}

\newcommand{\nat}{\mathbb{N}}
\newcommand{\semOf}[1]{[\![#1]\!]}
\newcommand{\areg}{\mathsf{r}}
\newcommand{\vecreg}{\overline{\areg}}
\newcommand{\setreg}{\mathsf{Reg}}
\newcommand{\domain}{\mathsf{Dom}}
\newcommand{\setcom}{\mathsf{Com}}
\newcommand{\transform}{\mathsf{tf}}
\newcommand{\transformOf}[1]{\transform(#1)}
\newcommand{\setconfigs}{\mathsf{CF}}
\newcommand{\setlocalconfigsmap}{\mathsf{LCFMap}}
\newcommand{\setlocalconfigs}{\mathsf{LCF}}
\newcommand{\setlabels}{\mathsf{LAB}}
\newcommand{\settime}{\textcolor{purple}{\mathsf{Time}}}

\newcommand{\setregvals}{\mathsf{RVal}}
\newcommand{\aview}{\mathsf{vw}}
\newcommand{\atimestamp}{\mathsf{ts}}
\newcommand{\settid}{\mathsf{TID}}
\newcommand{\tidOf}[1]{\settid(#1)}
\newcommand{\setviews}{\mathsf{View}}
\newcommand{\amem}{\mathsf{m}}
\newcommand{\setmem}{\mathsf{Mem}}
\newcommand{\setevents}{\mathsf{Msgs}}
\newcommand{\eventsOf}[1]{\setevents(#1)}
\newcommand{\anadr}{\mathsf{x}}

\newcommand{\avar}{\mathsf{x}}

\newcommand{\avaluation}{\mathsf{rv}}
\newcommand{\avalue}{\mathsf{d}}
\newcommand{\anevent}{\mathsf{msg}}
\newcommand{\cas}{\mathsf{cas}}
\newcommand{\casOf}[1]{\cas(#1)}
\newcommand{\anabstevent}{\anevent^{{\abstscript}}}
\newcommand{\anexp}{\mathsf{e}}
\newcommand{\anexpOf}[1]{\anexp(#1)}

\newcommand{\com}{\mathsf{c}}
\newcommand{\aval}{\mathsf{val}}
\newcommand{\acf}{\mathsf{cf}}
\newcommand{\anlcf}{\mathsf{lcf}}
\newcommand{\anlcfmap}{\mathsf{lcfm}}
\newcommand{\choice}{\oplus}
\newcommand{\myskip}{\mathsf{skip}}

\newcommand{\assign}[2]{#1 \coloneqq #2}
\newcommand{\load}[2]{\assign{#1}{#2}}
\newcommand{\store}[2]{\assign{#1}{#2}}
\newcommand{\assume}[1]{\mathsf{assume}\; {#1}}
\newcommand{\assert}[1]{\mathsf{assert}\; {#1}}
\newcommand{\bnf}{\;\mid\;}
\newcommand{\join}{\sqcup}
\newcommand{\athread}{\mathsf{t}}
\newcommand{\init}{\mathsf{init}}
\newcommand{\arun}{\rho}
\newcommand{\noconflict}{\;\#\;}
\newcommand{\superpos}{\triangleright}

\newcommand{\rulecaslocal}{\textsc{(CAS-local)}}
\newcommand{\rulecaslocalabstthat}{\textsc{(CAS-local$^{\thatthread}$)}}

\newcommand{\ruleunlabelled}{\textsc{(Unlabelled)}}
\newcommand{\rulecasglobal}{\textsc{(CAS-global)}}
\newcommand{\ruleloadlocal}{\textsc{(LD-local)}}
\newcommand{\rulestorelocal}{\textsc{(ST-local)}}
\newcommand{\ruleloadglobal}{\textsc{(LD-global)}}
\newcommand{\rulestoreglobal}{\textsc{(ST-global)}}

\newcommand{\ruleloadlocalabstthat}{\textsc{(LD-local$^{\thatthread}$)}}

\newcommand{\rulestorelocalabstthat}{\textsc{(ST-local$^{\thatthread}$)}}

\newcommand{\abstequalp}{=_{\abstscriptpshort}}
\newcommand{\abstmem}{\amem^{\abstscript}}
\newcommand{\abstview}{\aview^{\abstscript}}
\newcommand{\abstpview}{\aview^{\thisthreadshort\thatthreadshort}}

\newcommand{\abstlcfmap}{\anlcfmap^{\abstscript}}
\newcommand{\setpreds}{\mathsf{Preds}}
\newcommand{\setdatadomain}{\mathsf{Data}}
\newcommand{\setrules}{\mathsf{Rules}}
\newcommand{\env}{\dlogmem}
\newcommand{\datalogprog}{\mathsf{Prog}}
\newcommand{\gdatalogprog}{\mathsf{Prog}}
\newcommand{\agatom}{\mathsf{g}}

\newcommand{\thisthread}{{\color{\abstcolor}\mathsf{dis}}}
\newcommand{\thatthread}{{\color{red}\mathsf{env}}}
\newcommand{\thisthreadshort}{{\color{\abstcolor}\mathsf{d}}}
\newcommand{\thatthreadshort}{{\color{red}\mathsf{e}}}

\newcommand{\projectto}[2]{#1\!\downarrow_{#2}}

\usepackage[normalem]{ulem}
\newcommand{\unloopy}{\mathsf{acyc}}

\newcommand{\uncassy}{\mathsf{nocas}}

\newcommand{\thatthreadprog}{\com_\thatthread}
\newcommand{\thisthreadprog}[1]{\com_{\thisthread}^{#1}}
\newcommand{\thisthreadprogsimple}{\com_\thisthread}
\newcommand{\thatthreadtype}{\mathsf{type}_\thatthread}
\newcommand{\thisthreadtype}{\mathsf{type}}

\newcommandx{\adwait}[2][1=]{\todo[inline,linecolor=red,backgroundcolor=purple!5,bordercolor=red,#1]{{\color{red}Adwait:} #2}}

\newcommandx{\krishna}[2][1=]{\todo[linecolor=green!50!blue,backgroundcolor=green!5,bordercolor=green!50!blue,#1]{{\color{green!50!blue}Krishna:} #2}}

\newcommandx{\roland}[2][1=]{\todo[linecolor=violet,backgroundcolor=violet!5,bordercolor=violet,#1]{{\color{violet} Roland:} #2}}

\newcommand{\pA}[1]{\pnode_{\textcolor{black}{#1}}}
\newcommand{\yA}[1]{\ynode_{\textcolor{black}{#1}}}

\usepackage{stackengine}
\stackMath
\newcommand\xxrightarrow[1]{\raisebox{-.85pt}{\ensuremath{\smash{\mathrel{%
  \setbox2=\hbox{\stackon{\scriptstyle#1}{\scriptstyle#1}}%
  \stackon[-4.5pt]{%
    \xrightarrow{\makebox[\dimexpr\wd2\relax]{}}%
  }{%
   \scriptstyle#1\,%
  }%
}}}}}

\newcommand{\set}[1]{\{#1\}}

\newcommand{\localtrans}[1]{\xrightharpoondown{#1}}

\newcommand{\messPred}{\mathsf{emp}}
\newcommand{\statePred}{\mathsf{etp}}
\newcommand{\dmessPred}{\mathsf{dmp}}
\newcommand{\dstatePred}{\mathsf{dtp}}
\newcommand{\curr}{\mathsf{lc}}

\newcommand{\varset}{\mathsf{Var}}

\newcommand{\gvar}{\texttt{g}}
\newcommand{\dvar}{\texttt{d}}

\usepackage{microtype}

\title{Safety Verification of Parameterized Systems under Release-Acquire}

\author{Adwait Godbole \footnote{This work was done when Adwait Godbole was a final year undergraduate student at  IIT Bombay}}{University of California at Berkeley} {adwait@berkeley.edu}{}{}

\author{Shankara Narayanan Krishna}{IIT Bombay, India} {krishnas@cse.iitb.ac.in}{}{}
\author{Roland Meyer}{Institute of Theoretical Computer Science, Braunschweig, Germany} {roland.meyer@tu-bs.de}{}{}
\authorrunning{Adwait Godbole, S. Krishna, Roland Meyer}
\Copyright{Adwait Godbole, S. Krishna, Roland Meyer}
\ccsdesc[500]{Concurrency, Verification}

\keywords{release acquire, parameterized systems}

\category{} 

\relatedversion{} 

\supplement{}
\hideLIPIcs  

\begin{document}

\maketitle 
\begin{abstract}
We study the safety verification problem for parameterized systems under the release-acquire (RA) semantics. 
It has been shown that the problem is intractable for systems with unlimited access 
to atomic compare-and-swap (CAS) instructions.
We show that, from a verification perspective where approximate results
help, this is overly pessimistic.
We study parameterized systems consisting of an unbounded number of
environment threads executing identical but CAS-free programs and a fixed number of distinguished threads that are unrestricted.

Our first contribution is a new semantics that considerably simplifies
RA but is still equivalent for the above systems as far as safety verification is concerned. 
We apply this (general) result to two subclasses of our model. 
We show that safety verification is only \pspace-complete for the bounded
model checking problem where the distinguished threads are loop-free.
Interestingly, we can still afford the unbounded environment.
We show that the complexity jumps to \nexp-complete for thread-modular verification where an unrestricted distinguished `ego' thread
interacts with an environment of CAS-free threads plus loop-free distinguished threads
(as in the earlier setting).
Besides the usefulness for verification, the results are strong in that
they delineate the tractability border for an established semantics. 
\end{abstract}

\section{Introduction}
Release-acquire (RA) is a popular fragment of C++11~\cite{Batty:2011:MCC:1925844.1926394} (in which reads are annotated by acquire and writes by release) that strikes a good balance between programmability and performance and has received considerable attention (see e.g., \cite{DBLP:conf/ecoop/KaiserDDLV17,DBLP:journals/pacmpl/Kokologiannakis18,DBLP:conf/esop/RaadLV18,DBLP:conf/popl/LahavGV16,Alglave:10.1145/2627752,OG15,RSL13,GPS14,GPS++18}).
The model is not limited to concurrent programs, though. 
RA has tight links~\cite{RA16} with causal consistency (CC)~\cite{Causal95}, a prominent consistency guarantee in distributed databases~\cite{conf/sosp/LloydFKA11}. 
Common to RA implementations and distributed databases is that they tend to offer functionality to multi-threaded client programs, be it means of synchronization or access to shared data. 

We are interested in verifying such implementations on top of RA. 
For verification, we can abstract the client program to invocations of the offered functionality~\cite{DBLP:journals/sigact/BloemJKKRVW16}. 
The result is a so-called instance of the implementation in which concurrent threads execute the code of interest. 
There is a subtlety.
As the RA implementation should be correct for every client, we cannot fix the instance to be verified.  
We have to prove correctness irrespective of the number of threads executing the code.
This is the classical formulation of a parameterized system as it has been studied over the last 35 years~\cite{DBLP:journals/sigact/BloemJKKRVW16}. 

We are interested in the decidability and complexity of safety verification for parameterized programs under RA.
The goal is to identify expressive classes of programs for which the problem is tractable.
There are good arguments in favor of this agenda.
From a pragmatic point of view, even if the implementation at hand does not fall into one of the classes identified, we may hope for a reasonably precise encoding.
From a conceptual point of view, tractability of verification is linked to programmability, and understanding the complexity may lead to suggestions for better consistency notions~\cite{LahavBoker20} or programming guidelines, e.g. in the form of type systems~\cite{SMRTypes20}.  
Safety verification is a good fit for linearizability~\cite{HW90}, the de-facto standard correctness conditions for concurrency libraries, and has to be settled before going to more complicated notions.

To explain the challenges of parameterized verification under RA, it will be helpful to have an understanding of how to program under RA. 
The slogan of RA is \emph{never read ``overwritten'' values}~\cite{RA16}.  
Assume we have shared variables $\gvar$ and $\dvar$, initially $0$, and a thread first stores $1$ to $\dvar$ and then $1$ to $\gvar$. 
Assume a second thread reads the $1$ from~$\gvar$.  
Under RA, that thread can no longer claim $\dvar=0$. 
Formulated axiomatically~\cite{HerdingCats}, the reads-from, modification order, program order, and from-read should be acyclic~\cite{RA16}. 
While less concise, there are operational formulations of RA that make explicit, information about the computation which will be useful for our development~\cite{RAOperationally16,RAOperationallyVafeiadis17,DBLP:conf/ecoop/KaiserDDLV17}. 
The mechanism is as follows. 
Program and modification order are encoded as natural numbers, called \emph{timestamps}. 
Each thread stores locally a \emph{view} object, a map from shared variables to timestamps. 
This map reflects the thread's progress in terms of seeing (or as above hearing from) stores to a shared variable. 
The communication is organized in a way that achieves the desired acyclicity. 
Store instructions generate \emph{messages} that decorate the variable-value pair by a view. 
This view is the one held by the thread except that the timestamp of the variable being written is raised to a strictly higher value. 
The shared memory is implemented as a pool to which the generated messages are added and in which they remain forever. 
When loading a message from the pool, the timestamp of the variable given by the message must be at least the timestamp in the thread. 
The views are then joined so that the receiver cannot load values older than what the sender has seen. 

The timestamps render the RA semantics infinite-state, which makes algorithmic verification difficult. 
Indeed, the problem of solving safety verification under RA in a complete way has been studied very recently in the non-parameterized setting and proven to be undecidable even for programs with finite control flow and finite data domains~\cite{AAAK19}. 
With this insight, \cite{AAAK19} proposes to give up completeness and show how to encode an under-approximation of the safety verification problem into sequential consistency~\cite{Lamport79}.  
Lahav and Boker~\cite{LahavBoker20} drew a different conclusion.   
They proposed strong release-acquire~(SRA) as a new consistency guarantee under which safety verification is decidable for general non-parameterized programs. 
Unfortunately, the lower bound is again non-primitive recursive. 
Also the related problem of checking CC itself for a given implementation has been studied. 
It is undecidable in general, but EXPSPACE-complete under the assumption of data independence~\cite{CausalConsistency17}. 

To sum up, despite recent efforts~\cite{CausalConsistency17,AAAK19,LahavBoker20} we are missing an expressive class of programs for which the safety verification problem under RA is tractable. 
The parameterized verification problem has not been studied.

\smallskip 
\medskip \noindent{\textbf{Problem Statement}.}  
The parameterized systems of interest have the form $\thatthread \parallel \thisthread_1 \parallel \dots \parallel \thisthread_n$. 
We have a fixed number of \textit{distinguished} threads
collectively referred to as $\thisthread$ and executing programs $\thisthreadprog{1}, \cdots \thisthreadprog{n}$,  respectively. 
Moreover, we have an \textit{environment} consisting of arbitrarily many threads executing the same program $\thatthreadprog$. 
We obtain an \emph{instance} of the system by also fixing the number of number of environment threads. 
The safety verification problem is as follows:
\smallskip
\addtolength\leftmargini{-0.15in}
\begin{quote}
\emph{Safety Verification for Parameterized Systems}:\\
Given a parameterized system $\thatthread \parallel \thisthread_1 \parallel \dots \parallel \thisthread_n$, is there an instance of the system and a computation in that instance that reaches an assertion violation?

\end{quote}
\smallskip

The complexity of the problem depends on the system class under consideration. 
We denote system classes by signatures of the form $\thatthread(\thatthreadtype) \parallel \thisthread_1(\thisthreadtype_1) \parallel \dots \parallel \thisthread_n(\thisthreadtype_n)$, where the types constrain the programs executed by the threads. 
The parameters are the structure of the control flow, which may be loop-free, denoted by $\unloopy$, and the instruction set, which may forbid the atomic compare-and-swap (CAS) command, denoted by $\uncassy$. 
We drop the type if no restriction applies. 
If a thread is not present, we do not mention it in the signature. 
With this, $\thisthread_1(\unloopy) \parallel \thisthread_2(\uncassy)\parallel \thisthread_3$ is a non-parameterized system (without $\thatthread$ threads) having 
three $\thisthread$ threads executing: a loop-free $\thisthreadprog{1}$, $\thisthreadprog{2}$ which does not have CAS instructions, and $\thisthreadprog{3}$ which is free of restrictions, respectively. 

\smallskip
\noindent{\textbf{Justifying the Parameters.}} 
In \cite{AAAK19}, the safety verification problem under RA has been shown to be undecidable for non-parameterized  ($\thatthread$-free) systems from $\thisthread_1(\uncassy)\parallel \thisthread_2(\uncassy) \parallel \thisthread_3 \parallel \thisthread_4$ and non-primitive-recursive for systems from $\thisthread_1(\uncassy)\parallel \thisthread_2(\uncassy)$. 
There are several conclusions to draw from this.  

With distinguished threads, we cannot hope to arrive at a tractable verification problem. 
We take the bounded model checking~\cite{BMC01} approach and consider loop-free code. Acyclic programs, however, are not very expressive. 
Fortunately, RA implementations tend to be parameterized, and, as we will see, this frees us from the acyclicity restriction. 
The fact that parameterization simplifies verification has been observed in various works~\cite{DBLP:conf/lics/Kahlon08,Hague11,DBLP:conf/concur/TorreMW15,EsparzaGantyMajumdarJACM16,tso-param2019} that we discuss below.  

Restricting the use of CAS requires an explanation. 
The class $\thatthread$ of unconstrained environment threads enables what we call 
\textit{leader isolation}: an $\thatthread$ thread can distinguish itself from the others by acquiring a CAS-based lock. 
Even just $t$ CAS operations allows for the isolation of $t$ distinguished threads,
which takes us back to the results of \cite{AAAK19} for $t = 2$ resp. $t=4$.
Acyclicity will not help in this case, in section \ref{app:undec} we show that safety verification for $\thatthread(\unloopy)$ is undecidable.

\smallskip

\medskip\noindent{\bf{Contributions}}.
We state our main results and present the technical details in the later parts. 

\medskip \noindent{\textbf{A Simplified Semantics}}. 
We consider parameterized systems of the form $\thatthread(\uncassy)\parallel\thisthread_1 \parallel \dots \parallel \thisthread_n$.
Our first contribution is a simplified semantics (Section \ref{Section:Simplification}) that is equivalent with the standard RA semantics as far as safety verification is concerned.  
The simplified semantics uses the notion of \textit{timestamp abstraction}, which allows us to be imprecise about the exact timestamps of the $\thatthread$ threads. 
Note that we do not make any assumptions on the form of the distinguished threads but support cyclic control flow and CAS. 
So the result in particular applies to the intractable classes from~\cite{AAAK19}, even when extended with a parameterized environment. 
Supporting CAS in the distinguished threads is important. 
Without it, there is no way to capture the optimistic synchronization strategies used in performance critical programming~\cite{Herlihy91}. 

\medskip

\noindent We continue to apply the simplified semantics to prove tight complexity bounds for the safety verification problem in two particular cases of $\thisthread$ programs. 

\medskip \noindent{\textbf{Loop-Free Setting.}}
In Section \ref{sec:consec-lfree}, we show a \pspace{}-upper bound 
for the safety verification problem of parameterized programs from $\thatthread(\uncassy)\parallel \thisthread_1(\unloopy)\parallel \cdots\parallel\thisthread_n(\unloopy)$. 
The class reflects the bounded model checking problem~\cite{BMC01}, which unrolls a given program into a loop-free under-approximation. 
Interestingly, we can sequeeze into \pspace{} the unbounded environment of cyclic threads. 
Our decision procedure is not only optimal complexity-wise, it also has the potential of being practical (we do not have experiments).
We show how to encode the safety verification problem into the query evaluation problem for linear Datalog, the format supported by Horn-clause solvers \cite{bjorner2013solving,bjorner2015horn}, a state-of-the-art backend in verification. 

\medskip \noindent{\textbf{Leader Setting.}}
We continue to show an \nexp-upper bound 
for $\thatthread(\uncassy)\parallel\thisthread_1(\unloopy)\parallel\cdots\parallel\thisthread_{n}(\unloopy)\parallel\leader$
in Section \ref{sec:nexp-c}. 
These systems add an unconstrained distinguished thread, called the leader (denoted $\leader$), to the system from Section \ref{sec:consec-lfree}. 
The class is in the spirit of thread-modular verification techniques~\cite{OwickiG76,FQ02}, where the safety of a single `ego' thread is verified when interacting with an environment.

\medskip \noindent We note that these results delineate the border of tractability: adding another $\thisthread$ thread results in a non-primitive-recursive lower bound \cite{AAAK19}, and adding CAS operations to $\thatthread$ results in undecidability (section \ref{app:undec}).

\medskip \noindent{\textbf{Lower Bounds}}.
Our last contributions are matching lower bounds for the two classes. 
Interestingly, they hold even in the absence of CAS. We show that the safety verification problem is \pspace-hard already for $\thatthread(\uncassy, \unloopy)$, while it is \nexp-hard for  $\thatthread(\uncassy, \unloopy) \parallel \leader(\uncassy)$. 

\medskip
\noindent{\bf{Related Work}}.
There is a vast body of work on algorithmic verification under consistency models.
Since our interest is in decidability and complexity, we focus on complete methods. 
We have already discussed the related work on RA and CC. \\[0.2cm]
\noindent \emph{Other Consistency models.}
Atig et al. have shown that safety verification is decidable for assembly programs running on TSO, the consistency model of x86 architectures~\cite{AtigTSO10}.
The result has been generalized to consistency models with non-speculative writes~\cite{ABBM12} and very recently to models with persistence~\cite{PersistentTSO21}.
It has also been generalized to parameterized programs executed by an unbounded number of threads~\cite{LoadBufferTSO18}. 
Behind the decision procedures are (often drastic) reformulations of the semantics combined with well-structuredness arguments~\cite{DBLP:conf/lics/AbdullaJ93}. 
A notable exception is~\cite{tso-param2019}, showing that safety verification under TSO can be solved in \pspace{} for cas-free parameterized programs, called $\thatthread(\uncassy)$ here.
On the widely-used Power architecture safety is undecidable~\cite{PowerUndec20}.  

The decidability and complexity of verification problems has been studied also for distributed databases and data structures. 
Enea et al. considered the problem of checking eventual consistency (EC)~\cite{EC09} of replicated databases and developed a surprising link to vector addition systems~\cite{EventualConsistency14} that yields decidability and complexity results for the safety and liveness aspects of EC. 
For concurrent data structures, the default correctness criterion is linearizability wrt. a specification~\cite{HW90}. 
While checking linearizability is \expspace-complete in general~\cite{Linearizability00,LinearizabilityLower15}, important data structures (for which the specification is then fixed) admit \pspace-algorithms~\cite{LinReach15}.\\[0.2cm]
\noindent \emph{Parameterized Systems with Asynchronous Communication}. 
We exploit a pleasant interplay between the asynchronous communication in RA and the parameterization of our systems in the number of threads.
Kahlon~\cite{DBLP:conf/lics/Kahlon08} was the first to observe that parameterization simplifies verification in the case of concurrent pushdowns.
Hague \cite{Hague11} showed that safety verification remains decidable when adding a distinguished leader thread. 
Esparza, Ganty, Majumdar studied the complexity of what is now called leader-contributor systems~\cite{EsparzaGantyMajumdarJACM16}.
It is surprisingly low, \np-complete for systems of finite-state components and \pspace-completeness for systems of pushdowns.  
At the heart of their technique is the so-called copycat-lemma. 
The work has been generalized \cite{DBLP:conf/concur/TorreMW15} to all classes of models that are closed under regular intersection and have a computable downward-closure.
It has also been generalized to liveness verification~\cite{DBLP:conf/cav/Durand-Gasselin15,DBLP:conf/cav/FortinMW17}. 
Finally, the study has been generalized to parameterized complexity, for safety~\cite{DBLP:journals/jar/ChiniMS20} and liveness~\cite{DBLP:conf/fsttcs/Chini0S19}. 
Our work is related in that the distinguished threads behave like a leader. 
Moreover, our simplified semantics relies on an infinite-supply property the proof of which gives a copycat variant for RA. 
Our Datalog encoding is reminiscent of the notion of Strahler number~\cite{flajolet1979number}. 

Leader-contributor systems are closely related to broadcast networks~\cite{EsparzaFinkelMayrBroadcast99,DBLP:conf/concur/SinghRS09}.
Also there, safety verification has been found to be suprisingly cheap, namely \ptime-complete~\cite{DBLP:conf/forte/DelzannoSZ12}. 
For liveness verification, there was a gap between \expspace{} and \ptime{} that was settled recently with a non-trivial polynomial-time algorithm~\cite{BroadcastLiveness19}. 
What is new in broadcast networks and neither occurs in leader-contributor systems nor in our setting is the problem of reconfiguration~\cite{ReconfigurationSagnier12,ConstrainedReconfiguration18,ReconfigurationBertrand2019}.

\section{The Release-Acquire Semantics}\label{sec:ra}
A parameterized system consists of an unknown and potentially large number of threads, all running the same program. 
Threads compute locally over a set of registers and interact with each other by writing to and reading from a shared memory. The interaction with the shared memory is under the Release Acquire (RA) semantics \cite{RA16,RAOperationally16,RAOperationallyVafeiadis17}.   

\setlength{\belowcaptionskip}{-10pt}

\tikzset{background rectangle/.style={fill=none
}}
\begin{figure*}[h]
\centering
\small
\resizebox{\textwidth}{!}{
\begin{tikzpicture}[codeblock/.style={line width=0.5pt, inner xsep=0pt, inner ysep=5pt}  , show background rectangle]
\node[codeblock] (init) at (current bounding box.north west) {
$
\def\arraystretch{1}
\begin{array}{c}
\rowcolor{orange!5}
\begin{array}{cccc}
\inferrule{(\com_1, \avaluation, \aview)\localtrans{\anevent}(\com_1', \avaluation', \aview')
  }{
  (\com_1;\com_2, \avaluation, \aview)\localtrans{\anevent} (\com_1';\com_2, \avaluation', \aview')
  } 
  &
\inferrule{i=1,2}{
  (\com_1\choice \com_2, \avaluation, \aview)\localtrans{} (\com_i, \avaluation, \aview)
  } 
  &
\inferrule{~}{
  (\myskip;\com, \avaluation, \aview)\localtrans{} (\com, \avaluation, \aview)
  } 
  &
\inferrule{~}{
  (\assert{\texttt{false}}, \avaluation, \aview)\localtrans{} \bot
  } 
\end{array}
  \\[0.5cm]
\rowcolor{orange!5}
\begin{array}{ccc}
\inferrule{~}{
  (\com^*, \avaluation, \aview)\localtrans{} (\myskip\choice \com;\com^*, \avaluation, \aview)}
&
\inferrule{
  \semOf{\anexp}(\avaluation(\vecreg)) =  \avalue\quad \avaluation'=\avaluation[\areg\mapsto \avalue]
  }{
  (\assign{\areg}{\anexpOf{\vecreg}}, \avaluation, \aview)\localtrans{} (\myskip, \avaluation', \aview)
  }
  &
\inferrule{
  \semOf{\anexp}(\avaluation(\vecreg)) \neq 0
  }{
  (\assume{\anexpOf{\vecreg}}, \avaluation, \aview)\localtrans{} (\myskip, \avaluation, \aview)
  }
\end{array}
  \\[0.5cm]
\rowcolor{blue!5}
\begin{array}{cc}
\rulestorelocal\quad\inferrule{
\avaluation(\areg)=\avalue\quad \aview<_{\xvar}\aview'
  }{
  (\store{\xvar}{\areg}
, \avaluation, \aview)\localtrans{\storesym, (\xvar, \avalue, \aview')} (\myskip, \avaluation, \aview')
  }
  \hspace{0.5cm} &
  \hspace{0.5cm}
\ruleloadlocal\quad\inferrule{
\aview(\xvar)\leq\aview'(\xvar)\quad \avaluation'=\avaluation[\areg\mapsto \avalue]}{
  (\load{\areg}{\xvar}
, \avaluation, \aview)\localtrans{\loadsym, (\xvar, \avalue, \aview')} (\myskip, \avaluation', \aview\join\aview')
  } 
\end{array}
  \\[0.5cm] 
  \rowcolor{blue!5}
  \small{\rulecaslocal
  \quad\inferrule{\avaluation(\areg_1) = \avalue_1\quad \avaluation(\areg_2)=\avalue_2\quad  \aview(\avar)\leq \aview'(\avar)=\atimestamp\\
\widetilde\aview = \aview'[\avar\mapsto\atimestamp+1]\quad\aview'' = \aview \join \widetilde \aview}{
  (\casOf{\avar, \areg_1, \areg_2}
, \avaluation, \aview)\localtrans{\loadsym, (\avar, \avalue_1, \aview')} \localtrans{\storesym, (\avar, \avalue_2, \aview'')}(\myskip, \avaluation, \aview'')}}
      \\[0.5cm] 
   \addlinespace[0.1cm] \hline \addlinespace[0.1cm]
  \rowcolor{green!5}
  \begin{array}{cc}
\ruleloadglobal~~\inferrule{\anlcfmap(\athread) = \anlcf \quad
\anlcf\localtrans{\loadsym, \anevent}\anlcf'\quad \anevent\in \amem
  }{
  (\amem, \anlcfmap) \xrightarrow{(\athread, \anevent)} (\amem, \anlcfmap[\athread\mapsto\anlcf'])
  } 
  \hspace{0.5cm}&
  \hspace{0.5cm}
  \rulestoreglobal~~\inferrule{
\anlcfmap(\athread) = \anlcf \quad \anlcf\localtrans{\storesym, \anevent}\anlcf'\quad \anevent \noconflict \amem
  }{
  (\amem, \anlcfmap)\xrightarrow{(\athread, \anevent)} (\amem\cup\set{\anevent}, \anlcfmap[\athread\mapsto\anlcf'])
  }
\end{array}
  \\
  \rowcolor{green!5}
  \begin{array}{cc}
  \rulecasglobal~\inferrule{\anlcfmap(\athread) = \anlcf \quad
\anlcf\localtrans{\loadsym, \anevent_l}\localtrans{\storesym, \anevent_s}\anlcf'\quad \anevent_l \in \amem \quad  \anevent_s \noconflict \amem
  }{
  (\amem, \anlcfmap) \xrightarrow{(\athread, \anevent)} (\amem\cup\set{\anevent_s}, \anlcfmap[\athread\mapsto\anlcf'])
  }
  \hspace{0.5cm}&
  \hspace{0.5cm}
  \ruleunlabelled~\inferrule{\anlcfmap(\athread) = \anlcf \quad
\anlcf\localtrans{}\anlcf'}{
  (\amem, \anlcfmap)\xrightarrow{\athread} (\amem, \anlcfmap[\athread\mapsto\anlcf'])
  } 
\end{array}
\end{array}$
};
\end{tikzpicture}
}
\vspace*{-5mm}
\caption{Local transition relation: silent (thread-local) transitions (pink), shared memory transitions (blue). Global transition relation (below in green)}
\label{app:TransitionRelation}
\end{figure*}

\subsection{Program Syntax}
We model the individual threads in our system as (non-deterministic) sequential programs.
Assume a standard while-language $\setcom$ defined by: 

\begin{tcolorbox}[ams align*,colback=yellow!10!white,colframe=white]
\com\;::=&\quad 
\myskip\bnf
\assume{\anexpOf{\vecreg}}\bnf 
\assert{\texttt{false}} \bnf
\assign{\areg}{\anexpOf{\vecreg}} \bnf 
\com;\com\bnf \com \choice \com\bnf 
\com^*\bnf \\
&\quad
\load{\areg}{\xvar} \bnf
\store{\xvar}{\areg} \bnf 
\casOf{\avar, \areg_1, \areg_2}
\end{tcolorbox}

The programs compute on (thread-local) registers $\areg$ from the finite set $\setreg$ using assume, assert, assignments, sequential composition, non-deterministic choice, and iteration. 
Conditionals $\mathsf{if}$ and iteratives $\mathsf{while}$ can be derived from these operators, and we use them where convenient. 
The shared memory variables $\xvar$ are accessed only by means of load, store and compare-and-swap (CAS) operations as $\load{\areg}{\xvar}$,  
$\store{\xvar}{\areg}$ and $\casOf{\xvar, \areg_1, \areg_2}$, respectively. These instructions are also referred to as \emph{events}. 
We have a finite set $\varset$ of shared variables, and work with the data domain  $\domain = \nat$.
We do not insist on a shape of expressions $\anexp$ but require an interpretation $\semOf{\anexp}:\domain^{n}\rightarrow\domain$ that respects the arity $n$ of the expression.

\subsection{Release-Acquire (RA) Semantics}
We give the semantics of parameterized systems under release-acquire consistency. 
We opted for an operational~\cite{RAOperationally16,RAOperationallyVafeiadis17} over an axiomatic~\cite{RA16} definition, and follow~\cite{AAAK19}. 
What makes the operational definition attractive is that it comes with a notion of configuration or state of the system that we use to reason about computations. 
We first define thread-local configurations, then add the shared memory, and give the global transition relation.

\noindent \textbf{Local Configurations.} The RA semantics enforces a total order on all stores to the same variable that have been performed in the computation. 
We model these total orders by $\settime =  \nat$ and refer to elements of $\settime$ as timestamps.  
Using the total orders, each thread keeps track of its progress in the computation. 
It maintains a \emph{view} from $\setviews = \varset \rightarrow \settime$, 
a function, that for a shared variable $\xvar$, returns the timestamp of the most recent event the thread has observed on $\xvar$. 
Besides, the thread keeps track of the command to be executed next (which can be represented as program counter) and the register valuation from $\setregvals  = \setreg\rightarrow \domain$.
The set of \emph{thread-local configurations} is thus
\begin{align*}
\setlocalconfigs\; =\; \setcom\times \setregvals \times \setviews. 
\end{align*}

\noindent \textbf{Unbounded Threads.} The number of threads executing in the system is not known a priori. 
As long as we restrict ourselves to safety properties, there are two ways of modeling this. 
One way is to define \emph{instance programs} for a given number of threads, and then requiring correctness of all instances, as has been done in~\cite{ParamBook}.  
The alternative is to consider an infinite number of threads right away. 
We take the latter approach and define $\settid = \nat$ to be the set of thread identifiers. 
The thread-local configuration map then assigns a local configuration to each thread:
\begin{align*}
\setlocalconfigsmap \; =\;  \settid\rightarrow \setlocalconfigs.
 \end{align*}

\noindent \textbf{Views.} The views maintained by the threads are used for synchronization. 
They determine where in the (appropriate) total order a thread can place a store and from which stores it can load a value.   
To achieve this, the shared memory consists of \emph{messages}, which are variable value pairs enriched by a view, with the form $(\xvar, \avalue, \aview)$: 
 \begin{align*}
\setevents\; =\; \varset \times \domain \times \setviews.
 \end{align*}

\noindent \textbf{Shared Memory.} A \emph{memory state} is a set of such messages, and we use $\setmem =  2^\setevents$ for the set of all memory states. 
With this, the set of all \emph{configurations} of parametrized systems under release-acquire is 
\begin{align*} 
\setconfigs = \setmem\times \setlocalconfigsmap. 
 \end{align*}

\noindent \textbf{Transitions.} To define the transition relation among configurations, we first give a \emph{thread-local transition relation} among thread-local configurations $\localtrans{}\ \subseteq \setlocalconfigs\times \setlabels\times \setlocalconfigs$ in Figure~\ref{app:TransitionRelation}. 
Thread-local transitions may be labeled or unlabeled, indicated by $\setlabels =\set{\varepsilon}\cup (\set{\loadsym, \storesym, \cas}\times \setevents)$. 
The unlabeled transitions capture the control flow within a thread and properly handle assignments and assumes. 
They are standard. 
The message-labeled transitions capture the interaction of the thread with the shared memory. We elaborate on the load, store, and CAS transitions by which a thread with local view $\aview$, interacts with the shared memory. 

\noindent{\bf{Load}}. A load transition $\load{\areg}{\xvar}$ picks a message $(\xvar, \avalue, \aview')$ from the shared memory where $\avalue$ is the value stored in the message and updates its register $\areg$ with value $\avalue$.  
The message should not be outdated, which means the timestamp of~$\xvar$ in the message, $\aview'(\xvar)$, should be at least the thread's current timestamp for $\xvar$, $\aview(\xvar)$.  
The timestamps of other variables do not influence the feasibility of the load transition. 
They are taken into account, however, when the load is performed. 
The thread's local view is updated by joining the thread's current view $\aview$ and $\aview'$ by taking the maximum timestamp per address; $(\aview\join\aview') = \lambda\anadr. \max(\aview(\anadr), \aview'(\anadr))$.

\noindent{\bf{Store}}. When a thread executes a store $\store{\xvar}{\areg}$ it adds a message 
$(\xvar, \avalue, \aview')$ to the memory, where 
 $\avalue$ is the value held by the register $\areg$.
The new thread-local view (and the message view), $\aview'$, is obtained from the current $\aview$ by increasing the time-stamp of $\xvar$.
We use  $\aview<_{x}\aview'$ to mean $\aview(\xvar)<\aview'(\xvar)$ and $\aview(\yvar)=\aview'(\yvar)$ for all variables $\yvar \neq\xvar$. 

\noindent{\bf{CAS}}. A CAS transition is a load and store instruction executed atomically. $\casOf{\xvar, \areg_1, \areg_2}$ has the intuitive meaning $\mathsf{atomic\{}\load{\areg}{\xvar};\;\assume{\areg = \areg_1};\;\store{\xvar}{\areg_2}\mathsf{\}}$.
The instruction checks whether the shared variable $\xvar$ holds the value of $\areg_1$ and, in case it does, sets it to the value of $\areg_2$. The check and the assignment happen atomically. Under RA, this means the timestamp $\atimestamp$ of the load instruction 
and the timestamp $\atimestamp'$ of the store instruction involved in the CAS should be adjacent, $\atimestamp'=\atimestamp+1$.

The transition relation among configurations 
$\xrightarrow{}\ \subseteq \setconfigs\times \settid\times (\setevents\cup\set{\varepsilon})\times \setconfigs$ is defined in Figure~\ref{app:TransitionRelation}. 
It is labeled by a thread identifier and possibly message (if the transition interacts with the shared memory). 
The relation expects a thread $\athread$ which performs the transition. 
In the case of local computations, there are no more requirements and the transition propagates to the configuration. 
In the case of loads, we require the memory to hold the message to be loaded.
In the case of stores, the message to be stored should not conflict with the memory.
In the case of CAS, we require both of the above, and that the two messages should have consecutive timestamps.
We defer the definition of non-conflicting messages for the moment until we can give it in broader perspective.

Fix a parametrized system of interest $\com$. 
The initial thread-local configuration is $\anlcf_{\init}=(\com, \avaluation_0, \aview_0)$, where the register valuation assigns $\avaluation_0(\areg)=0$ to all registers and the view has $\aview_0(\xvar)=0$ for all $\xvar \in \varset$.  
The \emph{initial configuration} of the parametrized system is $\acf_0=(\setmem_{\init}, \anlcfmap_{\init})$ with an 
initial memory $\setmem_{\init}$ consisting of messages where all shared variables store the  
value $\avalue_\init \in \domain$, along with the initial view which assigns time stamp 0 to all shared variables, 
and $\anlcfmap_{\init}(\athread)=\anlcf_{\init}$ for all threads. 
A \emph{computation} (or a run) is a finite sequence of consecutive transitions
\begin{align*}
\arun\; =\; \acf_0\xrightarrow{(\athread_1, \anevent_1)}\acf_1\xrightarrow{(\athread_2, \anevent_2)}\ldots \xrightarrow{(\athread_n, \anevent_n)}\acf_{n}.
\end{align*}

The computation is initialized if $\acf_0=\acf_{\init}$. 
We use $\tstamps{\arun}$ for the set of all non-zero timestamps that occur in all configurations across all variables. 
We use $\tidOf{\arun}$ to refer to the set of thread identifiers labeling the transitions. 
For a set $\settid'\subseteq \settid$ of thread identifiers, we use $\projectto{\arun}{\settid'}$ to project the computation to transitions from the given threads.  
With $\firstOf{\arun}=\acf_0$, $\lastOf{\arun}=\acf_n$ we access the first/last configurations in the computation. 

\begin{example}
Consider the program given in Figure \ref{Figure:dekker} which implements a simplified version of the Dekker's mutual exclusion protocol for two threads.  
There are two shared variables $\xvar$ and $\yvar$. Both $\xvar, \yvar$ are initialized to 0, and 
at instructions $\lambda_0, \lambda'_0$ the registers $\areg_1, \areg'_1$ are initialized to 1.
The first thread $\athread_1$ signals that it wants to enter the critical section by 
writing the value 1 to $\xvar$. It then checks if thread $\athread_2$ has asked to  enter the critical section 
by reading the value of $\yvar$ and storing it into the register $\areg_1$.  The thread $\athread_1$ is allowed to  enter the critical section only if the value stored in the register $\areg_1$ is 0.
The second thread $\athread_2$ behaves in a symmetric manner. 	
\vspace{-0.5cm}
\tikzset{background rectangle/.style={fill=none}}
\begin{figure}[h]
\centering
\small
\begin{subfigure}{\textwidth}
   \centering
\begin{tikzpicture}[codeblock/.style={line width=0.5pt, inner xsep=0pt, inner ysep=5pt}  , show background rectangle]
\node[codeblock] (init) at (current bounding box.north west) {
$
\def\arraystretch{1}
\begin{array}{c}
\text{Variables } \xvar \text{ and } \yvar \text{ have been initialized to 0} \\
\begin{array}{l|l}
   \hline
   \text{Thread } \athread_1 & \text{Thread } \athread_2 \\ \hline 
   \lambda_0:~ \assign{\areg_1}{1} & \lambda'_0:~ \assign{\areg'_1}{1} \\
   \lambda_1:~ \store{\xvar}{\areg_1} & \lambda'_1:~ \store{\yvar}{\areg'_1}  \\
   \lambda_2:~ \load{\areg_1}{\yvar} & \lambda'_2:~ \load{\areg'_1}{\xvar} \\
   \lambda_3:~ \texttt{if}(\areg_1 == 0):  & \lambda'_3:~ \texttt{if}(\areg'_1 == 0): \\
   \quad\qquad \texttt{criticalsection} & \quad\qquad \texttt{criticalsection} \\ \hline 
\end{array}
\end{array}$
};
\end{tikzpicture}
\end{subfigure}
\begin{subfigure}{\textwidth}
   \centering
   \includegraphics[scale=0.4]{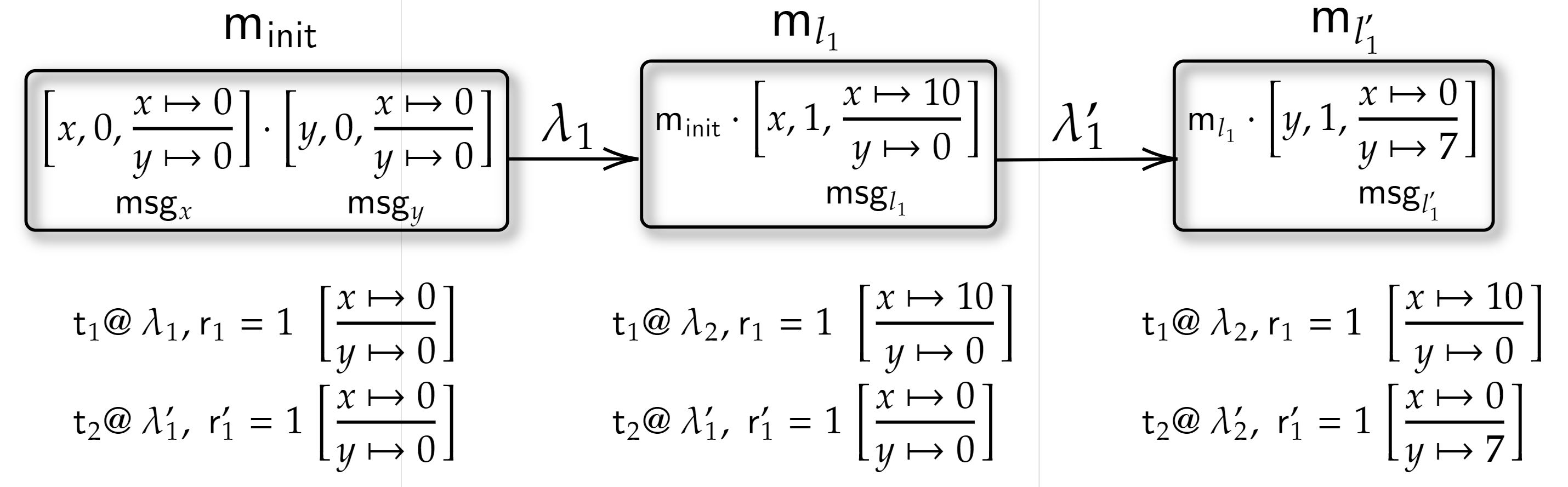}
\end{subfigure}
\caption{On the top, is a simplified version of Dekker's mutual exclusion protocol. 
Below, is a partial execution sequence under RA. The rectangles show the contents (messages) of the shared memory. Messages have three components - (1) the variable, (2) value of the message and (3) the message view - a map from $\{\xvar,\yvar\}$ to the set of timestamps $\settime$ ($\mathbb{N}$). The lines below show the thread-local state - instruction-pointer, register valuation and thread-local view}.
\label{Figure:dekker}
\end{figure}

\newcommand{\pc}{\mathsf{pc}}

Under Sequential Consistency (SC) \cite{Lamport79}, which is a stronger notion of consistency, the mutual exclusion property (i.e., at most one thread is in the critical section at any time) is preserved. However, this is not the case under the RA memory model. To see why, consider the execution sequence presented in Fig \ref{Figure:dekker}. 
At each instant, the figure shows where the instruction-counter (i.e., the label of the next instruction to get executed) resides in each of the threads, along with 
	 the values of the registers. The black arrows with instruction labels $\lambda_1, \lambda'_1$ show 
	 the evolution of the run on executing  the instruction labeled $\lambda_1, \lambda'_1$ respectively.  
	 Let $\amem_{\lambda}$ represent the memory obtained after executing the instruction labeled $\lambda$, and let 
	 $\anevent_{\lambda}$ be the unique new message (if any) that is part of $\amem_{\lambda}$ after the execution 
	 of the instruction labeled $\lambda$. The initial memory is $\amem_{\init}$ where $\xvar, \yvar$ have values and timestamps 0; 
	 $\anevent_{\xvar}, \anevent_{\yvar}$ represent the messages in $\amem_{\init}$ corresponding to $\xvar, \yvar$.
  The execution of the instruction labeled by $\lambda_1$ results 
 in  the addition of a new message $\anevent_{\lambda_1}$ to the memory whose timestamp (10) is higher 
  than 0 (which is  the current timestamp of the variable $\xvar$ for $\athread_1$). The view of $\athread_1$ is then updated to $\xvar \mapsto 10, \yvar \mapsto 0$.
     Likewise, the execution of the instruction labeled by $\lambda'_1$ 
  results in the addition of  a new message $\anevent_{\lambda'_1}$ to the memory 
  with a higher time stamp (7). This will result in the  update of  the view of $\athread_2$ to $\xvar \mapsto 0, \yvar \mapsto 7$, wrt. the variable $\yvar$. 
  The read instruction labeled by $\lambda_2$ is then allowed to use the message $\anevent_{\yvar}$ to fetch the value of $\yvar$,
  since the view of $\athread_1$ wrt. $\yvar$ is 0. Likewise, in the case of $\athread_2$ concerning the execution of the instruction labeled by $\lambda'_2$, 
  the message $\anevent_{\xvar}$ is used since the view of $\athread_2$ wrt. $x$ is 0. 
    After these steps, both threads enter their respective critical section. 
\end{example}

\subsubsection{Conflict} 
\label{subsec:conflict}

We need a notion of conflict not only for messages, but also for memories, configurations, and computations. 
Two messages are \textit{non-conflicting}, denoted by $(\xvar_1, \avalue_1, \aview_1)\noconflict (\xvar_2, \avalue_2, \aview_2)$, if either their variables  are different, $\xvar_1\neq \xvar_2$, the timestamps are different, $\aview_1(\xvar_1) \neq \aview_2(\xvar_2)$, or the timestamps are zero, $\aview_1(\xvar_1){=}0{=} \aview_2(\xvar_2)$. 
Observe that initial messages do not conflict with any other message. 

Two memory states are non-conflicting, $\amem_1\noconflict\amem_2$, if for all $\anevent_1\in \amem_1$ and all $\anevent_2\in \amem_2$ we have $\anevent_1\noconflict\anevent_2$. 
Two configurations are non-conflicting, $\acf_1\noconflict\acf_2$, if their memory states are non-conflicting. 
Two computations are non-conflicting, denoted $\arun\noconflict\arun'$, if they use different threads and non-conflicting messages, $\tidOf{\arun}\cap\tidOf{\arun'}=\emptyset$ and $\lastOf{\arun}\noconflict \lastOf{\arun'}$. 
\section{Undecidability of $\thatthread(\unloopy)$}
\label{app:undec}

In this section, we establish the undecidabililty of the class $\thatthread(\unloopy)$, that is the class with loop-free $\thatthread$ threads (which can execute arbitrarily-many CAS operations) and without any $\thisthread$ threads. This result essentially shows that even with the loop-free assumption, allowing $\thatthread$ threads to perform CAS operations is in itself intractable from a safety verification viewpoint. 
Hence, the $\uncassy$ restriction that we impose on $\thatthread$ threads is a justified means of achieving tractabililty.

In fact, we will show a stronger result. We will show that we can transform a non-paramterized system consisting of $k$ distinguished  threads
having  full instruction set and loops under RA (the class $\thisthread_1 \parallel\cdots \parallel \thisthread_k$) to a parameterized system 
 corresponding to the class $\thatthread(\unloopy)$ such that, control-state reachability is preserved. With this equivalence, the claim would follow by the undecidability result of \cite{AAAK19}.

\subsection{Constructing the equivalent loop-free program}
\newcommand{\badlabel}{\lambda_\mathsf{fail}}
\newcommand{\nextstate}{\mathsf{next}}
\newcommand{\instr}{\iota}
\newcommand{\assumeinst}{\mathsf{assume}}
\newcommand{\assertinst}{\mathsf{assume}}
\newcommand{\loadinst}{\mathsf{assume}}
\newcommand{\storeinst}{\mathsf{assume}}
\newcommand{\casinst}{\mathsf{assume}}
\newcommand{\assigninst}{\mathsf{assign}}
\newcommand{\startstate}{\mathsf{start}}
\newcommand{\finalstate}{\mathsf{end}}
\newcommand{\threemapping}{\mbox{\Lightning}}
\newcommand{\regvar}{\mathfrak{r}}

To show this result, we transform an input of $n$ programs $\{\com_1, \com_2, \cdots, \com_n\}$ and a failure state label $\badlabel$ in some $\com_i$ (with possibly full instruction set and loops) and transform them into a single program $\com_\thatthread$ and failure state $\badlabel'$, the control flow for which is loop-free but uses the full instruction set including CAS operations. We claim that the state label $\badlabel$ is reachable in $\thisthread_1 \parallel\cdots \parallel \thisthread_n$ with $\thisthread_i$ executing $\com_i$ if and only if the state label $\badlabel'$ is reachable in the system $\thatthread(\unloopy)$ with the environment threads executing $\com_\thatthread$. Let the variable set, data domain and register set of the original system be $\varset$, $\domain$ and $\setreg = \{\areg_1, \cdots, \areg_k\}$ as usual. We assume that the memory is initialized to 0 on all variables.

\textbf{Converting a single program $\com_i$ to $\com_i'$}
We show how we can convert one thread program $\com_i$ into a loop-free program $\com_i'$ and then show how we can combine all the programs together into a single loop-free program $\com_\thatthread$. Consider for an $i$, the program $\com_i$. For the purposes of this construction, we will assume that the program $\com_i$ has been specified as a transition system rather than in the while-language syntax. It is clear that both representations are equivalent and can be interconverted with only polynomial overhead. Hence we assume that $\com_i = (Q, \Delta, \instr)$ where $Q$ is the set of control states, $\Delta$ is the transition relation and $\instr$ maps each transition to its corresponding instruction from $\{\myskip, \assume{\anexpOf{\vecreg}}, \assign{\areg}{\anexpOf{\vecreg}}, \load{\areg}{\xvar}, \store{\xvar}{\areg}, \casOf{\avar, \areg_1, \areg_2}\}$. We transform $\com_i$ to a loop-free program as follows. Let $Q = \{q_0, \cdots, q_n\}$ and with $q_0$ as the initial state.

In this conversion, we add extra variables and values such that $\varset' = \varset \uplus \{t_i, \regvar_i^1, \cdots, \regvar_i^{|\setreg|}\}$ and $\domain' = \domain \cup \{0, \cdots, |Q|-1, \Lambda^\bot\}$ where $\uplus$ denotes disjoint union. 
Now we specify the new transition system $\com_i'$ which needs to be loop-free. For each transition $\partial = (q_{a}, q_{b}) \in \Delta$, with source and end states $q_{a}$ and $q_{b}$ respectively and instruction $\instr(\partial)$, we transform it into the following three transition sequence (a sequence of green transitions starting with a CAS, followed by $k$ load operations; the transition corresponding to $\instr(\partial)$, followed by 
the sequence of pink transitions consisting of $k$ store operations ending with a CAS) 
denoted as $\threemapping(\partial)$.
\begin{align*}
\threemapping(\partial) ~=&~ \textcolor{green!50!blue}{q_\startstate \xrightarrow{\casOf{t_i, a, \Lambda^\bot}} q_\partial^0 \xrightarrow{\load{\areg_1}{\regvar_i^1}} q_\partial^1 \cdots q_\partial^{k-1} \xrightarrow{\load{\areg_k}{\regvar_i^k}}} q_\partial^k \xrightarrow{\instr(\partial)} q_\partial^{k+1} \cdots \quad \text{(contd. below)}\\
&~ \cdots q_\partial^{k+1}  \textcolor{magenta}{\xrightarrow{\store{\regvar_i^1}{\areg_1}} q_\partial^{k+2} \cdots q_\partial^{2k-1} \xrightarrow{\store{\regvar_i^k}{\areg_k}} q_\partial^{2k}  \xrightarrow{\casOf{t_i, \Lambda^\bot, b}} q_\finalstate}
\end{align*}

We construct $\threemapping(\partial)$ for each transition $\partial$ in $\Delta$ to get the complete transition system. The initial/final (collectively called terminal) nodes of this transition system are $q_\startstate, q_\finalstate$ (which are common to all $\threemapping(\partial)$). The internal $q_\partial^{\_}$ states are all distinct across the $\threemapping(\partial)$ for different $\partial$. The transition system that we obtain has size 
$\mathcal{O}
(|\setreg||\Delta|)$ (each original transition from $q_a$ to $q_b$ is transformed into a sequence of $2|\setreg|+3$ transitions  between $q_\startstate$ and $q_\finalstate$). It is clearly loop-free. See Figure \ref{fig:app-a} showing an example if we start with $\thisthread_i\parallel\thisthread_j$.

\textbf{Combining individual $\com_i'$}
We construct programs $\com_i'$ as described above for each thread $\thisthread_i$. Now we combine these individual programs into a single program $\com_\thatthread$.
We ensure that the newly added shared variables ($t_i$, $\regvar_i^{\_}$ for $\thisthread_i$) are also disjoint across threads. Hence the variable set is now $\varset' = \varset \uplus \{t_1, \cdots, t_n\} \uplus \{\regvar_i^j\}_{i \in [n], j\in [|\setreg|]}$ (where $\varset$ was the original variable set $\{\com_1, \com_2, \cdots, \com_n\}$ were operating with). Finally, the combined data domain is simply a union of individual data domains (which possibly overlap). We combine the individual programs as
\begin{equation*}
	\com_\thatthread = \com_1' \oplus \com_2' \oplus \cdots \oplus \com_n'
\end{equation*}
where $\oplus$ denotes non-deterministic choice. It is clear that $\com_\thatthread$ is loop-free. Additionally, $|\com_\thatthread|$, and the new $|\domain'|$ and $|\varset'|$ is polynomial in $\sum_i |\com_i|$ and previous $|\varset|$ and $|\domain|$.


	
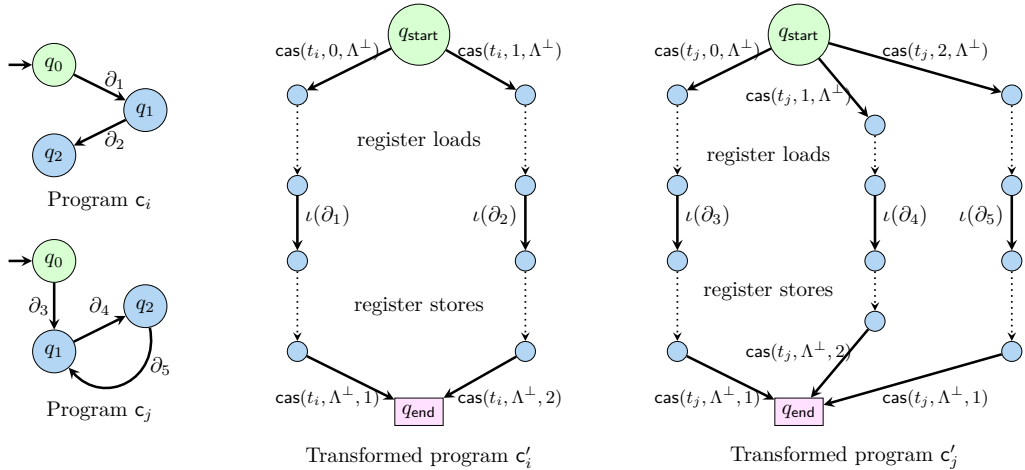
\begin{figure}[h]
\centering
\newtcbox{\colorboxouline}[1][]{boxsep=0.3pt,left=0.1pt,right=0.1pt,top=1pt,bottom=1pt,colframe=magenta,colback=white,boxrule=0pt,toprule=1pt,bottomrule=1pt,sharp corners,#1}

\tikzstyle{state}=[fill={rgb,255: red,178; green,214; blue,243}, draw=black, shape=circle]
\tikzstyle{init}=[fill={rgb,255: red,215; green,255; blue,210}, draw=black, shape=circle]
\tikzstyle{start}=[fill={rgb,255: red,255; green,224; blue,254}, draw=black, shape=rectangle]

\tikzstyle{transition}=[fill=none, ->,very thick, -stealth]
\tikzstyle{dotstrans}=[->, dotted,thick, -stealth]
\scalebox{0.8}{
\newcommand{\labelsize}{\footnotesize}
\begin{tikzpicture}
	\begin{pgfonlayer}{nodelayer}
		\node [style=none] (0) at (-7, -2.75) {};
		\node [style=none] (1) at (-7, 0.5) {};
		\node [style=init] (2) at (-6.25, 0.5) {$q_0$};
		\node [style=state] (3) at (-6.25, -1) {$q_2$};
		\node [style=init] (4) at (-6.25, -2.75) {$q_0$};
		\node [style=state] (5) at (-6.25, -4.25) {$q_1$};
		\node [style=state] (6) at (-4.75, -3.5) {$q_2$};
		\node [style=state] (8) at (-4.75, -0.25) {$q_1$};
		\node [style=none] (10) at (-5.25, 0.25) {$\partial_1$};
		\node [style=none] (11) at (-5.25, -0.75) {$\partial_2$};
		\node [style=none] (12) at (-4.5, -4.5) {$\partial_5$};
		\node [style=none] (13) at (-6.5, -3.5) {$\partial_3$};
		\node [style=none] (14) at (-5.5, -3.5) {$\partial_4$};
		\node [style=none] (15) at (-5.5, -1.75) {Program $\com_i$};
		\node [style=none] (16) at (-5.5, -5.25) {Program $\com_j$};
		\node [style=start] (17) at (-0.25, -5.25) {$q_\finalstate$};
		\node [style=start] (18) at (6, -5.25) {$q_\finalstate$};
		\node [style=state] (19) at (1.5, 0) {};
		\node [style=state] (20) at (-2.25, 0) {};
		\node [style=state] (21) at (1.5, -4.25) {};
		\node [style=state] (22) at (-2.25, -4.25) {};
		\node [style=state] (23) at (7.25, -0.5) {};
		\node [style=state] (24) at (9.5, 0) {};
		\node [style=state] (25) at (9.5, -2.75) {};
		\node [style=init] (26) at (-0.25, 1) {$q_\startstate$};
		\node [style=state] (35) at (1.5, -2.75) {};
		\node [style=state] (36) at (-2.25, -2.75) {};
		\node [style=state] (37) at (-2.25, -1.5) {};
		\node [style=state] (38) at (1.5, -1.5) {};
		\node [style=state] (40) at (9.5, -1.5) {};
		\node [style=state] (41) at (7.25, -1.5) {};
		\node [style=state] (42) at (4, 0) {};
		\node [style=state] (43) at (4, -1.5) {};
		\node [style=state] (44) at (4, -2.75) {};
		\node [style=state] (45) at (7.25, -2.75) {};
		\node [style=state] (46) at (9.5, -4.25) {};
		\node [style=state] (47) at (7.25, -3.75) {};
		\node [style=state] (48) at (4, -4.25) {};
		\node [style=init] (49) at (6, 1) {$q_\startstate$};
		\node [style=none] (50) at (6.75, -6) {Transformed program $\com'_j$};
		\node [style=none] (51) at (-0.25, -6) {Transformed program $\com'_i$};
		\node [style=none] (52) at (-1.75, -2) {$\instr(\partial_1)$};
		\node [style=none] (53) at (1, -2) {$\instr(\partial_2)$};
		\node [style=none] (54) at (4.5, -2) {$\instr(\partial_3)$};
		\node [style=none] (55) at (7.75, -2) {$\instr(\partial_4)$};
		\node [style=none] (56) at (9, -2) {$\instr(\partial_5)$};
		\node [style=none] (57) at (-0.25, -0.75) {register loads};
		\node [style=none] (58) at (5.5, -1) {register loads};
		\node [style=none] (59) at (-0.25, -3.5) {register stores};
		\node [style=none] (60) at (5.5, -3.25) {register stores};
		\node [style=none] (61) at (-1.75, 0.75) {\labelsize$\casOf{t_i,0, \Lambda^\bot}$};
		\node [style=none] (62) at (1.25, 0.75) {\labelsize$\casOf{t_i,1, \Lambda^\bot}$};
		\node [style=none] (63) at (1.25, -5) {\labelsize$\casOf{t_i,\Lambda^\bot,2}$};
		\node [style=none] (64) at (-1.75, -5) {\labelsize$\casOf{t_i,\Lambda^\bot, 1}$};
		\node [style=none] (65) at (4.5, -5) {\labelsize$\casOf{t_j,\Lambda^\bot, 1}$};
		\node [style=none] (66) at (8.25, -5) {\labelsize$\casOf{t_j,\Lambda^\bot, 1}$};
		\node [style=none] (67) at (4.5, 0.75) {\labelsize$\casOf{t_j,0, \Lambda^\bot}$};
		\node [style=none] (68) at (8.25, 0.75) {\labelsize$\casOf{t_j,2, \Lambda^\bot}$};
		\node [style=none] (69) at (6, 0) {\labelsize$\casOf{t_j,1, \Lambda^\bot}$};
		\node [style=none] (70) at (6, -4.25) {\labelsize$\casOf{t_j,\Lambda^\bot,2}$};
	\end{pgfonlayer}
	\begin{pgfonlayer}{edgelayer}
		\draw [style=transition] (2) to (8);
		\draw [style=transition] (8) to (3);
		\draw [style=transition] (4) to (5);
		\draw [style=transition] (5) to (6);
		\draw [style=transition, bend left=75, looseness=1.50] (6) to (5);
		\draw [style=transition] (1.center) to (2);
		\draw [style=transition] (0.center) to (4);
		\draw [style=transition] (26) to (19);
		\draw [style=transition] (26) to (20);
		\draw [style=transition] (21) to (17);
		\draw [style=transition] (22) to (17);
		\draw [style=transition] (37) to (36);
		\draw [style=transition] (38) to (35);
		\draw [style=transition] (40) to (25);
		\draw [style=transition] (41) to (45);
		\draw [style=transition] (43) to (44);
		\draw [style=transition] (49) to (24);
		\draw [style=transition] (49) to (23);
		\draw [style=transition] (49) to (42);
		\draw [style=transition] (46) to (18);
		\draw [style=transition] (47) to (18);
		\draw [style=transition] (48) to (18);
		\draw [style=dotstrans] (24) to (40);
		\draw [style=dotstrans] (25) to (46);
		\draw [style=dotstrans] (23) to (41);
		\draw [style=dotstrans, in=90, out=-90] (42) to (43);
		\draw [style=dotstrans] (44) to (48);
		\draw [style=dotstrans] (45) to (47);
		\draw [style=dotstrans] (19) to (38);
		\draw [style=dotstrans] (35) to (21);
		\draw [style=dotstrans] (20) to (37);
		\draw [style=dotstrans] (36) to (22);
	\end{pgfonlayer}
\end{tikzpicture}
}
\caption{Examples of two $\thisthread$ threads $\thisthread_i$ and $\thisthread_j$ executing programs $\com_i, \com_j$ and the corresponding transformed programs $\com'_i, \com'_j$. The program $c_i$ has 2 transitions while $c_j$ has 3 transitions. Note how the read-write value for CAS operations in the transformed program match with the transitions in the original program.}	
\label{fig:app-a}
\end{figure}

\subsection{Proof of Equivalence}
We now prove that the system $\thatthread(\unloopy)$ with the $\thatthread$ threads executing the program $\com_\thatthread$ as defined above respect  the original system. We defer the notion of `maintaining reachability' for a bit later. We first observe the program $\com_\thatthread$ and make some observations.

\subsubsection{Simulation of individual threads}
\textbf{Locking/unlocking of $\com_i'$.}
For any single transformed program $\com_i'$, we note that at any given point, only one thread can be in any internal (not initial/final) state of $\com_i'$. To see this, note the two atomic CAS operations flanking each 3-path in $\com_i'$. All these CAS operations are on the same variable $t_i$ and moreover there are no other operations on $t_i$. Hence at any given point in time, there is only one message on $t_i$ (the most recent write) that is available for a CAS operation. The value of this message dictates whether the operation will succeed. When it succeeds, the most recent write value changes to the value written by the CAS.

Now note that $t_i$ is initialized with value 0; hence initially one thread, say $\athread$ can take perform a CAS and change the recent value to $\Lambda^\bot$.  Now, there is no transition from $q_\startstate$ that performs a CAS with a read of $\Lambda^\bot$. Hence all other threads are kept waiting until the recent value on $t_i$ changes from $\Lambda^\bot$. This is possible only when the initial thread $\athread$, executes the final transition and reaches $q_\finalstate$, maintaining the claim. Hence these CAS operations perform the role of a mutual-exclusion lock. But then they perform another function too.

\textbf{State transference.} We now know that for each $i$, only one thread may execute $\com_i'$ at any given time. However the locking/unlocking operations using CAS also enable threads to transfer their state to their successors. There are three components to the state, which we handle in turn:
\begin{itemize}
	\item Control-state: Note that the recent value on variable $t_i$ is $v \neq \Lambda^\bot$ only if the previous thread terminated after simulating some transition ending at $q_v$. Additionally, a locking CAS operation for $\threemapping(\partial)$ reads value $v$ only if $\partial$ is a transition from $q_v$ to some other state. Hence, it is guaranteed that the successive thread will execute some transition that emerges from a state where the previous thread left off. Note how this is true for the first thread as well since the initial value on all variables is 0 and the initial state of the transition system is $q_0$.
	\item View: The second component that we consider is the view. This also is transferred from a thread to its successor through the CAS operation. In particular, when a thread $\athread$ executes the final CAS operation to reach $q_\finalstate$, it generates a message on $t_i$ which is read by its successor. This read implies that the successor will take the join on its own (initial) view with that of the message and hence essentially accumulate the exact view that the previous thread left with. So, the view is transferred as well.
	\item Register valuations: The previous thread $\athread$ stores its register valuations in the shared variables $\regvar_i^j$ in the final sequence of store operations before terminating. These are then accessed by the successor thread through the initial sequence of load operations.
\end{itemize}
In this way we see that not only is exclusion ensured, but the thread states are transferred from one thread to the next. Together, these sequences  of threads simulate the entire run of the original $\thisthread_i$ in fragments. The above holds for all $i \in [n]$. Hence at any given point, there are at most $n$ threads simulating the original ones.

\subsubsection{The Complete Simulation}
Now we formalize the notion of equivalence in reachability. We say that an original state with the threads $\{\thisthread_1, \cdots, \thisthread_n\}$ is equivalent to a new state when we have the following.
\begin{itemize}
	\item if the control states of threads $\{\thisthread_1, \cdots, \thisthread_n\}$ are $(q_{i_1}, q_{i_2}, \cdots, q_{i_n})$ respectively then in the new system with $\thatthread$, the recent value of shared variable $t_j$ is $i_j$,
	\item  the register valuation of each original $\thisthread_i$ is reflected ($\areg_j$ = most recent write to $\regvar_i^j$) in the most recent writes to the variables $\regvar_i^{\_}$ for each thread $i$, 
	\item  the view of $\thisthread_i$ is the view stored in the most recent message to $t_i$ (again projected on the original variable set) and 
	\item  the global memory (projected) is identical across the original and $\thatthread$ states. 
\end{itemize}
We claim the following: a state in the original system can be reached if and only some equivalent state in the new system can be reached. 
We can prove this by induction. The base case is that all threads are in their initial states, registers and views with the memory only with initial messages (0 on each variable). This trivially satisfies the requirement, both in the forward and reverse directions.

Now for the inductive case ($\Rightarrow$). Assume that it was true at some instant. Let some $\thisthread_i$ execute an instruction for the transition $\partial$. In the new system, we can simulate this as a new $\thatthread$ thread $\athread$ taking the path corresponding to $\threemapping(\partial)$ in $\com_i'$. We note by the observations above that the invariants for the thread-local state (control-state, register valuations and view) is maintained. Additionally, if $\thisthread_i$ wrote a message to the memory, then so can $\athread$. In particular, since the view of $\athread$ is obtained from the CAS read, it matches that of $\thisthread_i$. Hence the message added by $\athread$ can have the same timestamps as $\thisthread_i$.

Inductive case ($\Leftarrow$). The same argument works in the reverse direction. Assume that a pair of equivalent states have been reached. Now, consider a  $\thatthread$ thread path $\threemapping(\partial)$ in $\com_i'$ where $\partial = (q_a, q_b)$. Then, by the induction hypothesis this means that $\thisthread_i$ is in state $q_a$ in the original run. Given the equality of thread and memory state initially, it too can take the transition $(q_a, q_b)$. Once again, the invariant follows from the  earlier observations.

This gives a sketch of the proof. In particular, note that even though we give an equivalence between the control state in the original system and a variable value in the new system, this can be easily converted to an equivalence between control states themselves. This means that the reachability problem for $\thisthread_1 \parallel \cdots \parallel \thisthread_n$ can can converted to a reachability problem for $\thatthread(\unloopy)$. This prompts us to restrict $\thatthread$ threads to a reduced (cas-free) instruction set and motivates the idea of modelling CAS instructions in a run via computations of the $\thisthread$ threads.

\section{A Simplified Semantics}
\label{Section:Simplification}
In this section, we propose a simplified semantics for the class of systems given by 
$\thatthread(\uncassy) \parallel \thisthread_1 \parallel \dots \parallel \thisthread_n$. The core of this result 
relies on the \emph{Infinite Supply Lemma} which shows that 
if some $\thatthread$ thread could generate a message $(\anadr, \aval, \aview)$, then a clone of that thread could generate the message $(\anadr, \aval, \aview')$ with $\aview' = \aview[\anadr\mapsto t]$ for some $t > \aview(\anadr)$.

There are two assumptions that the infinite supply lemma and hence our semantic simplification result rely on:
\begin{itemize}
\item arbitrarily many $\thatthread$ threads executing identical programs. 
\item the $\thatthread$ threads do not have atomic instructions (CAS).
\label{enum:conditions}
\end{itemize}
The first assumption allows us to have clone $\thatthread$ threads that duplicate the computation and hence the messages generated in it. The second assumption is required for the duplicated computation to remain valid under RA.

While performing the duplication, one must keep in mind
the dependency between stores and loads across threads. The fact that $\thisthread$ threads 
are not replicatable (their messages cannot be duplicated) adds to the challenge. To ensure that the clone threads
can follow in the footsteps of the original computation we require that $\thisthread$ messages can be read by the $\thatthread$ clones whenever they can be read by the original $\thatthread$ threads.
This necessitates that we respect relative order among timestamps between $\thatthread$ and $\thisthread$ threads.

We develop some intermediate concepts that help us in developing a valid duplicate run. 
In order to accommodate the clone threads, we must make space (create unused timestamps) along $\settime$ for clones to write messages. We do this via \textit{timestamp liftings}. Having done this, we need to define how we can combine the original computation with that of the clones. 
We develop the concept of \textit{superposition} of computations to do this. 
Finally, the infinite supply (of messages) lemma shows how, using the earlier two concepts, we can generate copies of messages, with higher timestamps.

This `duplication-at-will' of $\thatthread$ messages means that we need not store the entire set of $\thatthread$ messages produced. Those with the smallest timestamps act as good representatives of the set. Additionally when any thread reads from an $\thatthread$ message, we need not be bothered about timestamp  comparisons since we could always generate a copy of that message with as high timestamp as required. It is this observation that gives us the timestamp abstraction and with it the simplified semantics.

\subsection{Infinite Supply}

We now make these arguments precise. Our strategy is to split up the timestamps (hence the computation) and separate the part originating from the $\thisthread$ threads from the $\thatthread$ part (which can be duplicated at will). We  write $\projectto{\arun}{\thatthread}$ and $\projectto{\arun}{\thisthread}$ to denote the projections of $\rho$ to $\thatthread$ and $\thisthread$ respectively.

\subsubsection{Timestamp Lifting}
\label{subsubsec:timestamp}
In our development we will make use of \emph{timestamp transformations} 
$\transform:\settime\rightarrow \settime$.
 We extend these to views $\aview$ with \textit{per variable} timestamp transformations $\transform = \{\transform^\xvar\}_{\xvar \in \varset}$, where $\transform^\xvar$ only transforms the timestamps for the variable $\xvar$. 
 The transformed view $\transformOf{\aview}: \varset \rightarrow\settime$ is defined by $(\transformOf{\aview})(\xvar) = \transform^\xvar(\aview(\xvar))$ for every variable $\xvar$. 

As an example consider shared variables $\xvar, \yvar$ and views $\aview_1, \aview_2$
such that $\aview_1=[\xvar \mapsto 2, \yvar \mapsto 5]$ and $\aview_2=[\xvar \mapsto 10, \yvar \mapsto 0]$. 
Using the timestamp transformation $\transform = \{\transform^\xvar, \transform^\yvar\}$
where $\transform^\xvar(0)=\transform^\yvar(0)=0$,
$\transform^\xvar(t) = t+2$ and  
$\transform^\yvar(t) = t+7$ for $t > 0$, we obtain 
$\transformOf{\aview_1} = [\xvar\mapsto 4, \yvar\mapsto 12]$ and 
$\transformOf{\aview_2} = [\xvar\mapsto 12,\yvar \mapsto 0]$. 
We also apply the timestamp transformation to messages, memories, configurations, and computations by transforming all view components.

\noindent {\textbf{RA-valid timestamp lifting}}. An \emph{RA-valid timestamp lifting} for a run $\rho$ is a (per variable) timestamp transformation $\gentmorph = \{\tmorph^\xvar\}_{\xvar \in \varset}$ satisfying two properties for each $\xvar \in \varset$: (1) it is strictly increasing, $\tmorph^\xvar(0)=0$; for all $t_1,t_2\in \nat$ with $t_1<t_2$ we have $\tmorph^\xvar(t_1)<\tmorph^\xvar(t_2)$ and (2) if there is a CAS operation on $\xvar$ with (load, store) timestamps as $(t, t+1)$ then $\tmorph^\xvar(t+1) = \tmorph^\xvar(t) + 1$, i.e. consecution of CAS-timestamps is maintained. Note that $\tmorph(\acf_{\init}) = \acf_{\init}$.
In the example above, $\mathsf{tf}$ is a RA-valid timestamp lifting. 
 
 Lemma \ref{lem:gen-time-lift} says that the run $\gentmorph(\arun)$ obtained by modifying the timestamps of a valid run $\arun$ with an RA-valid timestamp lifting $\gentmorph$ is also a valid under the RA semantics. 

\begin{lemma}[Timestamp Lifting Lemma]
\label{lem:gen-time-lift}
Let $\gentmorph = \{\tmorph^\xvar\}_{\xvar\in\varset}$ be an RA-valid timestamp lifting. If $\arun$ is a computation under RA, then so is $\gentmorph(\arun)$. 
Hence if a configuration $\acf$ is reachable under RA then so is $\gentmorph(\acf)$.
\end{lemma}
\begin{proof}
This result follows since timestamp lifting is just a relabelling of timestamps for each shared variable. The lemma relies on the following facts/observations: 
\begin{itemize}
\item There are no timestamp comparisons across variables, $\aview(\anadr)$ is never compared with $\aview(\anadr')$ for $\anadr \neq \anadr'$.
\item The relative order between timestamps on the same variable is preserved due to the strictly increasing property. Additionally, $\tmorphOf{0} = 0$, maintaining the timestamps of the $\init$ messages.
\item The load, store timestamps of $\rulecaslocal$ operations still remain consecutive.
\end{itemize}
In particular the lemma can be formally proven by induction on the length of the run. The base case is trivial and the inductive case follows by showing that each instruction - read, write, CAS - that can be executed in $\arun$  can be executed in the lifted run, $\gentmorph(\arun)$.
\end{proof}

\noindent The  duplication of messages by the clone $\thatthread$ threads requires us to copy computations and then merge them such that the RA semantics are not violated. This requries (1) timestamps of merging computations to not conflict and (2) the reads-from dependencies between threads are respected. With this in mind, we introduce the idea of superposition.

\subsubsection{Superposition}
\label{subsubsec:superpos}
We define the \emph{superposition} $\arun\superpos\arun'$ of two computations $\arun, \arun'$ as the computation that first executes $\arun$ and then $\arun'$. 
This requires us to combine the memory in $\lastOf{\arun}$ with that of every configuration in $\arun'$. 
Moreover, the threads transitioning in $\arun, \arun'$ must be disjoint.  
Given these considerations, the operation requires the computations to be non-conflicting, $\arun\noconflict\arun'$ (see Section \ref{subsec:conflict}), and is defined as follows:
\begin{align*}
\arun\superpos\arun'\;  = \; \arun;(\lastOf{\arun}+\arun'). 
\end{align*}
The addition of a configuration $\acf$ to a computation $\arun\; =\; \acf_0\xxrightarrow{(\athread_1, \anevent_1)}\ldots \xxrightarrow{(\athread_n, \anevent_n)}\acf_{n}$ yields the 
new computation 
\begin{align*}
\acf+\arun\; =\;  (\acf+\acf_0)\xrightarrow{(\athread_1, \anevent_1)}\ldots \xrightarrow{(\athread_n, \anevent_n)}(\acf+\acf_{n}). 
\end{align*}
Addition of configurations $\acf_1 =(\amem_1, \anlcfmap_1)$ and $\acf_2=(\amem_2, \anlcfmap_2)$ is the configuration $\acf_1+\acf_2 = (\amem_1\cup \amem_2, \anlcfmap)$, 
where $\anlcfmap(\athread)=\anlcfmap_1(\athread)$ if $\anlcfmap_1(\athread)\neq \anlcf_{\init}$ and $\anlcfmap(\athread)=\anlcfmap_2(\athread)$ otherwise. 

When $\arun\noconflict\arun'$ holds, we have: (1) for any thread $\athread$, if it has transitioned in $\rho$, then 
it cannot in $\rho'$; likewise, if it has not transitioned in $\rho$, then it can in $\rho'$.

(2) $\lastOf{\arun}\noconflict\lastOf{\arun'}$, and since the memory in earlier configurations of $\arun$ is a subset of that in $\lastOf{\arun}$, the memory unions performed above involve nonconflicting memories.
An initial configuration is neutral for addition, in particular $\lastOf{\arun'}+\firstOf{\arun}=\lastOf{\arun'}$. 
The operation of concatenation $\arun_1;\arun_2$ expects two computations $\arun_1$ and $\arun_2$ that satisfy $\lastOf{\arun_1}=\firstOf{\arun_2}$ and returns the sequence consisting of the transitions in $\arun_1$ followed by the transitions in $\arun_2$. 
This need not be a valid computation under RA, but under the following conditions it is.

Let $\setevents(\rho)$ be the memory in $\lastOf{\rho}$. Likewise, let $\setevents(\projectto\arun\thisthread) \subseteq \setevents(\rho)$ be the subset of memory in $\lastOf{\rho}$, which have been added by $\thisthread$ threads during $\rho$. 

\begin{lemma}[Superposition]
Consider valid computations $\arun,\arun'$ of a parametrized system under RA such that $\projectto\arun\thatthread \# \projectto{\arun'}\thatthread$ and that $\setevents(\projectto\arun\thisthread) = \setevents(\projectto{\arun'}\thisthread)$.
Then the superposition $\arun\superpos\projectto{\arun'}\thatthread$ is a valid computation under RA.
\label{lem:superpos}
\end{lemma}
\begin{proof}
Since there are arbitrarily many $\thatthread$ threads, we distinguish apart the $\thatthread$ threads in $\arun'$ from the $\thatthread$ threads in $\arun$. By doing so we ensure that the threads operating (changing state) in $\arun$ and $\projectto{\arun'}\thatthread$ are disjoint. 

Now consider the global state obtained after executing $\arun$ (which is a valid run under RA). By hypothesis, the memory state contains messages from $\projectto\arun\thisthread$, which are identical to those in $\projectto{\rho'}\thisthread$. After execution of $\arun$ is complete, we claim that we can execute $\projectto{\rho'}\thatthread$ one step at a time. 
\begin{itemize}
	\item Whenever a 
 $\thisthread$ thread loads from a message generated by a $\thatthread$  thread in  $\rho$, the same can happen in $\arun\superpos\projectto{\arun'}\thatthread$. Likewise, the relative time stamps between the $\thisthread$ threads and the $\thatthread$ threads in $\rho'$
 are the same; so $\projectto{\arun'}\thatthread$ can be executed after $\rho$.
\item 
 Likewise, reads made by some $\thatthread$ thread on $\rho, \rho'$ either from another $\thatthread$ thread 
or a $\thisthread$ thread also continues exactly in the same way in $\arun\superpos\projectto{\arun'}\thatthread$, since 
the messages added by $\thisthread$ threads are exactly same in $\rho, \rho'$, and 
the $\thatthread$ threads are disjoint.
\item The above two points show that we have exactly the same  reads-from dependencies ($\thisthread \leftrightarrow \thatthread$ in $\rho$, 
 $\thisthread \leftrightarrow \thatthread$ in $\rho'$) in $\arun\superpos\projectto{\arun'}\thatthread$.
The reason is that $\thatthread$ threads are disjoint and the 
messages added by $\thisthread$ threads are the same in $\rho, \rho'$. 
Finally, all writes made by the respective $\thatthread$ threads of $\rho, \rho'$  can be done 
in $\arun\superpos\projectto{\arun'}\thatthread$; likewise, all writes made by the $\thisthread$ threads in $\rho$ can also 
be made in $\arun\superpos\projectto{\rho'}\thatthread$. The reason is that 
 $\projectto{\arun}\thatthread \# \projectto{\rho'}\thatthread$, and trivially, we have $\projectto{\rho'}\thisthread \# \projectto{\arun'}\thatthread$. 
 This ensures no conflict of write-timestamps.
\end{itemize}
  Formally this can be proven by induction on the length of $\arun'$.
\end{proof}

Now we develop the infinite supply lemma. Recall that our goal is to generate arbitrarily many copies of $\thatthread$ messages 
with the same variable and value but higher timestamp. 
Let us fix one such message, $\anevent = (\anadr, \avalue, \aview)$, for our discussion here and see how we can replicate it. 
Towards this end, consider a computation $\rho$ in which it is generated.
We `spread-apart' the timestamps of $\setevents(\arun)$, using timestamp liftings so that we create `holes' (unused timestamps) along $\settime$. Then we generate copies of $\thatthread$ threads, denoted as $\copythread{\thatthread}$ (possible since $\thatthread$ threads can be replicated). 

The holes accomodate for the timestamps of $\copythread{\thatthread}$ and the (higher) timestamp of the copy of $\anevent$.
Throughout this, we preserve the order of timestamps of $\thatthread,\copythread{\thatthread}$ threads relative to those of $\thisthread$ threads. This ensures that reads-from dependencies are maintained - $\copythread\thatthread$ can read a $\thisthread$ message whenever $\thatthread$ can do so.

We define the computation $\widetilde\arun$ as a copy of $\projectto{\arun}{\thatthread}$ executed by $\copythread\thatthread$ threads.
The write timestamps used by $\copythread\thatthread$ threads are the unoccupied timestamps generated by the timestamp lifting operation $\gentmorph(\arun)$. We show an example of this via a graphic. Let ${\color{red!90} \mathsf{e}} {{\color{black!90}\mathsf{T^i}}}$ and ${\color{blue!90} \mathsf{d}} {{\color{black!90}\mathsf{T^i}}}$ respectively denote the timestamps 
chosen by $\thatthread$ and $\thisthread$ along $\rho$ (first row).
\newcommand{\thisthreadtsh}{\thisthreadshort\mathsf{T}}
\newcommand{\thatthreadtsh}{\thatthreadshort\mathsf{T}}
\newcommand{\thisthreadtsilsh}[1]{{\color{blue!20} \mathsf{d}} {{\color{black!20}\mathsf{T^#1}}}}
\newcommand{\thatthreadtsilash}[1]{{\color{red!20} \mathsf{e}} {{\color{black!20}\mathsf{T^#1_{a}}}}}
\newcommand{\thatthreadtsilbsh}[1]{{\color{red!20} \mathsf{e}} {{\color{black!20}\mathsf{T^#1_{a}}}}}
\newcommand{\aye}{\mathsf{a}}
\newcommand{\bee}{\mathsf{b}}

\vspace{-1em}
\begin{align*}
\arun\!:~& {\color{gray} \init} ~~\thisthreadtsh^0 ~~ \thatthreadtsh^0 ~~ \thisthreadtsh^1 ~~ \thatthreadtsh^1 ~~ \thatthreadtsh^2~~ \\
\gentmorph(\arun)\!:~&  {\color{gray} \init} ~~\thisthreadtsh^0~~ \thatthreadtsilbsh{0}~~ \thatthreadtsh^0_{\aye}~~ \thisthreadtsh^1~~ \thatthreadtsilbsh{1}~~ \thatthreadtsh^1_{\aye}~~ \thatthreadtsilbsh{1}~~ \thatthreadtsh^2_{\aye} \\
\widetilde{\arun}\!:~& {\color{gray} \init} ~~\thisthreadtsilsh{0} ~~ \thatthreadtsh^0_{\bee}~~\thatthreadtsilash{0} ~~\thisthreadtsilsh{1}~~ \thatthreadtsh^1_{\bee}~~ \thatthreadtsilash{1}~~ \thatthreadtsh^2_{\bee} ~~ \thatthreadtsilash{2}
\end{align*}
The second row shows lifted timestamps (with subscript $\mathsf{a}$) of $\gentmorph(\arun)$ and the holes (faded). The third row shows holes being used by $\copythread\thatthread$ for $\widetilde\arun$ (these have subscript $\mathsf{b}$).
The construction guarantees $\gentmorph(\arun) \noconflict \widetilde\arun$ and superposition 
$\gentmorph(\arun) \superpos \widetilde\arun$ is allowed. 
In this computation, $\widetilde\arun$ 
generates a copy of $\anevent$, $\anevent'=(\anadr, \aval, \aview')$
with higher $\aview'(\xvar)$.
Additionally, since $\thatthreadtsh^i_{\aye}, \thatthreadtsh^i_{\bee}$ have the same position relative to all $\thisthreadtsh^j$ timestamps, so will $\aview(\yvar), \aview'(\yvar)$ for $\yvar \neq \anadr$.
\newcommand{\ptimestamps}{\mathbb{P}}

Now we state the Infinite Supply Lemma. 
As helper notation, for a run $\arun$ and each variable $\anadr$, we denote the timestamps of stores of $\thisthread$ threads on $\anadr$ as $\thisTS{\anadr}{0} < \thisTS{\anadr}{1} < \cdots $. 
\begin{lemma}[Infinite Supply]
Let $\arun$ be a valid run under the RA semantics, in which the message $(\anadr, \avalue, \aview)$ has been generated by an $\thatthread$ thread. Then for each timestamp ${\color{violet} t^*}\in\mathbb{N}$, there exist two  timestamp lifting functions $\gentmorph_1$, $\gentmorph_2$ and a run $\arun_1$ such that $\gentmorph_1(\arun) \superpos \gentmorph_2(\projectto{\arun}{\thatthread}) \superpos \arun_1$ is a valid run. This run contains a message $(\anadr, \avalue, \aview')$ satisfying (${\color{blue} \text{ts}}$ comes from $\arun$)
\begin{enumerate}
	\item $\forall i$ (${\color{violet} t^*} \leq \thisTS{\anadr}{i} \land \aview(\anadr) \leq \thisTS{\anadr}{i}) \implies \aview'(\anadr) \leq \tmorph_1^{\anadr}(\thisTS{\anadr}{i})$ 
	\item $\aview'(\anadr) \geq \tmorph_2^{\anadr}({\color{violet} t^*})$ 
	\item $\forall \anadr' \neq \anadr$, $\forall i$, $\aview(\anadr') \leq \thisTS{\anadr'}{i} \implies \aview'(\anadr') \leq \tmorph_1^{\anadr'}({\thisTS{\anadr'}{i}})$
\end{enumerate} 
\label{lem:inf-sup}
\end{lemma}

\begin{proof}
Without loss of generality, we assume that in the run $\arun$, the timestamps on each variable are consecutive. If that is not the case, we can always use a  timestamp lowering operation that `fills in the gaps' between non-consecutive timestamps, while maintaining consecution of the load,store timestamps of $\rulecaslocal$ operations.

We will give a constructive proof. We specify $\gentmorph_1(\arun) \superpos \gentmorph(\projectto{\arun}{\thatthread}) \superpos \arun_1$ by defining $\gentmorph_1$, $\gentmorph_2$ and $\arun_1$, and showing that the resulting run is valid under RA. Then we show how a copy of the message $(\anadr, \avalue, \aview)$ can be obtained as claimed. 

First, we describe how to copy runs. 
\begin{enumerate}
\item 	
{\textbf{Copying a run}}. 
For a variable $\anadr$, we define the lifting functions as follows. With the consecutiveness assumption, the messages on $\anadr$ have consecutive timestamps and are generated by either $\thisthread$ or $\thatthread$, which we will denote below as $\thisthreadt$ and $\thatthreadt$ respectively. For developing intution quickly, consider the following sequence of consecutive timestamps on some variable.

\begin{equation*}
  {\color{gray} \init} ~~\thisthreadt^0 ~~\thisthreadt^1~~ \thatthreadt^0~~ \thisthreadt^2~~ \thatthreadt^1~~ \thatthreadt^2~~ \thisthreadt^3
\end{equation*}

Intuitively, the new interleaved run is obtained by triplicating each $\thatthreadt$ timestamp into three adjacent timestamps, $\thatthreadt_a$, $\thatthreadt_b$ and $\thatthreadt_c$. The $\thatthreadt_a$ timestamps belongs to the lifted run $\gentmorph_1(\arun)$, the $\thatthreadt_b$ timestamp belongs to $\gentmorph_2(\arun)$ and the $\thatthreadt_c$ timestamp belongs to $\arun_1$. The $a,b,c$ copies are ordered as $b < c < a$ giving us the following timestamp sequence from the one above.

\begin{equation*}
  {\color{gray} \init} ~~\thisthreadt^0 ~~\thisthreadt^1~~ \thatthreadt^0_b~~ \thatthreadt^0_c~~ \thatthreadt^0_a~~ \thisthreadt^2~~ \thatthreadt^1_b~~ \thatthreadt^1_c~~ \thatthreadt^1_a~~ \thatthreadt^2_b~~ \thatthreadt^2_c~~ \thatthreadt^2_a ~~\thisthreadt^3
\end{equation*}

\newcommand{\twints}[1]{\mathsf{twin}(#1)}
We can formalize $\gentmorph_1$ and $\gentmorph_2$ by counting the number of $\thisthreadt$ timestamps smaller than the $\thatthreadt^i$ timestamp, but for ease of presentation, we will keep this implicit. The total shift  can be done for instance, using the function 
which maps a time stamp $p \in \nat$ corresponding to a $\thatthread$ thread to the (number of $\thisthread$ threads appearing before $p$) + 3(number of $\thatthread$ threads appearing before $p$)+3, while 
for a $\thisthread$ time stamp $p \in \mathbb{N}$, one can map it to 
 (number of $\thisthread$ threads appearing before $p$) + 3(number of 
 $\thatthread$ threads appearing before $p$)+1. 
 So for instance, the $\thatthreadt^0$ at timestamp 3 moves to the $\thatthreadt^0_a$ time stamp 5, while the $\thisthreadt^3$ moves to 
 the $\thisthreadt^3$ time stamp 3+ 3(3)+1=13.

\smallskip 

$\gentmorph_1$ maps the timestamp $\thatthreadt^i$ to the timestamp $\thatthreadt^i_a$ timestamp. Similarly, $\gentmorph_2$ maps the $\thatthreadt^i$ timestamp to the timestamp $\thatthreadt^i_b = \thatthreadt^i_a-2$. Finally, we have $\thatthreadt^i_c = \thatthreadt^i_a-1$. We call these adjacent timestamps `triplets'. Additionally, $\gentmorph_1$ and $\gentmorph_2$ map the $\thisthreadt^i$ timestamp to the corresponding $\thisthreadt^i$ timestamp in the expanded run.

We first note that $\gentmorph_1$ satisfies the premise of the timestamp lifting lemma, that for $\rulecaslocal$ operations, consecutive load,store timestamps for CAS remain consecutive.
This follows since only $\thisthread$ can perform $\rulecaslocal$ and under $\gentmorph_1$, consecution is maintained both for $(\thisthreadt^{i-1}, \thisthreadt^{i})$ timestamps as well as $(\thatthreadt_a, \thisthreadt^{i})$ as depicted in the following timestamp sequence. Thus, by the timestamp lifting lemma, $\gentmorph_1(\arun)$ is a valid run under RA.

\newcommand{\thatthreadtsila}[1]{{\color{red!20} \mathsf{env}} {{\color{black!20}\mathsf{T^#1_a}}}}
\newcommand{\thatthreadtsilc}[1]{{\color{red!20} \mathsf{env}} {{\color{black!20}\mathsf{T^#1_c}}}}
\newcommand{\thatthreadtsilb}[1]{{\color{red!20} \mathsf{env}} {{\color{black!20}\mathsf{T^#1_b}}}}

\begin{equation*}
  {\color{gray} \init} ~~\thisthreadt^0 ~~\thisthreadt^1~~ \thatthreadtsilb{0}~~\thatthreadtsilc{0}~~ \thatthreadt^0_a~~ \thisthreadt^2~~ \thatthreadtsilb{1}~~ \thatthreadtsilc{1}~~ \thatthreadt^2_a~~ \thatthreadtsilb{2}~~ \thatthreadtsilc{2}~~ \thatthreadt^3_a~~ \thisthreadt^3
\end{equation*}

\newcommand{\copythreadone}[1]{\mathsf{copyB}(#1)}

\item {\textbf{The first superposition gives a valid run}}. We now claim that the run $\gentmorph_1(\arun) \superpos \gentmorph_2(\projectto{\arun}{\thatthread})$ is a valid run under RA. For a $\thatthread$ thread $\athread$ in $\arun$ we denote by $\copythreadone{\athread}$ as its `copy' ($\mathsf{TID}$ distinguished apart) in $\gentmorph_2(\projectto{\arun}{\thatthread})$ ($\mathsf{B}$ since it occupies the $b$ timestamps). We will show that $\copythreadone{\athread}$ copies the transitions that $\athread$ took. We use the following invariant relating the view $\aview_1$ of thread $\athread$ in $\arun$ with the view $\aview_2$ of $\copythreadone{\athread}$ for showing this. Let  $\mathsf{TS}(t)$ denote the set of timestamps used by thread $t$ in a run $\rho$. 

\begin{quote}
	For every shared variable $\anadr$, if $\aview_1(\anadr) \in \mathsf{TS}(\thatthread)$, then $\aview_2(\anadr) = \aview_1(\anadr) - 2$, else if $\aview_1(\anadr) \in \mathsf{TS}(\thisthread)$ then $\aview_2(\anadr) = \aview_1(\anadr)$
\end{quote}

 Now we can reason by induction on the length of the run that whenever $\athread$ takes a transition in $\arun$, $\copythreadone{\athread}$ can replicate it but with the view as given by the above invariant. More precisely, whenever a $\thatthread$ thread $\athread$ makes a store with timestamp $\thatthreadt$, $\copythreadone{\thatthread}$ will make a store with timestamp $\thatthreadt-2$. Similarly, when a $\thatthread$ thread  $\athread$ makes a load,

\begin{enumerate}
  \item if the load is from a message by a $\thisthread$ thread, $\copythreadone{\thatthread}$ also loads from the same message 
  \item if the load is from a message by some $\thatthread$ thread $\athread'$ in $\arun$, $\copythreadone{\thatthread}$ loads from $\copythreadone{\athread'}$
\end{enumerate}

It is easy to check that the view invariant is maintained through this simulation. Crucially we have $\thatthreadt^i_a < \thisthreadt^j \iff \thatthreadt^i_b < \thisthreadt^j$. Thus $\athread$ and $\copythreadone{\athread}$ can always read the same set of $\thisthread$ messages. Thus we have that $\gentmorph_1(\arun) \superpos \gentmorph_2(\projectto{\arun}{\thatthread})$ is valid under RA. Now we focus on message generation by $\arun_1$. Intuitively, $\arun_1$ will also be a copy of $\arun$ but will occupy the $\thatthreadt_c$ timestamps.

\item 
{\textbf{Generating a copy of the message}}.  Now we will describe how we can use $\arun_1$ to generate the message $(\anadr, \avalue, \aview')$. Let $\arun_{p}$ be the prefix of $\arun$ just (one transition) before the message $(\anadr, \avalue, \aview)$ is generated. We generate the run $\arun'$ by copying the $\projectto{\arun_{p}}{\thatthread}$ using the $c$-timestamps. The run obtained is $\gentmorph_1(\arun) \superpos \gentmorph_2(\projectto{\arun}{\thatthread}) \superpos \arun'$. Using similar reasoning as earlier, we can show that this is a valid run under the RA semantics.  

\newcommand{\copythreadtwo}[1]{\mathsf{copyC}(#1)}

Now let $\thatthread$ thread $\athread$ be the thread which generates the message $(\anadr, \avalue, \aview)$ in $\arun$. Then there is a copy of $\athread$, thread $\copythreadtwo{\athread}$ in $\arun'$, that is now in a control state enabling it to generate a message on variable $\anadr$  with value $\avalue$ (since the transitions have been replicated exactly across $\thatthread$ and $\copythreadtwo{\thatthread}$). Now we need to reason about the view of the message generated by $\copythreadtwo{\athread}$. If the view of $\athread$ is $\aview_a$ and that of $\copythreadtwo{\athread}$ is $\aview_c$ we have the following, which again follows from the invariant mentioned above. 
\begin{quote}
	For each variable $\anadr'$, if $\aview_a(\anadr') \in \mathsf{TS}(\thatthread)$, then $\aview_c(\anadr') = \aview_a(\anadr') - 1$, else if $\aview_a(\anadr') \in \mathsf{TS}(\thisthread)$ then $\aview_c(\anadr') = \aview_a(\anadr')$
\end{quote}

Observe how this satisfies the condition (3) in the lemma immediately since in both cases above, we have $\aview_c(\anadr') \leq \aview_a(\anadr')$. Now the thread $\copythreadtwo{\athread}$ will choose the timestamp for variable $\anadr$, $\aview'(\anadr)$. Assume $t^* \in \nat$ has been given. We have two cases, (i) $t^* \leq \aview(\anadr)$ and (ii) $t^* > \aview(\anadr)$.
\begin{enumerate}
	\item[(i)] In this case there is nothing to prove as the original message is lifted to $(\anadr, \avalue, \aview')$ where $\aview'(\anadr) = \tmorph_1^\anadr(t^*)$ which satisfies both conditions (1) and (2). 
	\item[(ii)] We choose $\aview'(\anadr) =  \tmorph_2^\anadr(t^*) + 1$, which satisfies (2). Note that this timestamp is higher than $\aview_c(\anadr)$ since $\tmorph_2^\anadr(t^*) \geq \tmorph_1^\anadr(\aview(\anadr)) = \aview_a(\anadr) > \aview_c(\anadr)$. We have $\aview'(\anadr) =  \tmorph_2^\anadr(t^*) + 1 = \tmorph_1^\anadr(t^*) - 1$ due to the construction of $\gentmorph_1, \gentmorph_2, \arun_1$. Additionally, $\thisTS{\anadr}{i} \geq t^* \implies \tmorph_1^\anadr(\thisTS{\anadr}{i}) \geq \tmorph_1^\anadr(t^*) > \aview'(\anadr)$ satisfiying condition (1). In this case $\arun_1$ is defined as $\arun'$ extended by the store transition generating the message. Note that since it is a $c$-type message, the timestamp is available.
\end{enumerate}
\end{enumerate}

Thus in both cases, we have a message $(\anadr, \avalue, \aview')$ with $\aview'$ satisfying the required conditions. This proves the theorem.
\end{proof}

To sum up, we interpret the infinite supply as this: $\gentmorph_1(\arun)$ is the lifted run with holes. $\gentmorph_2(\projectto{\arun}{\thatthread})$ is the $\copythread{\thatthread}$ run and $\rho_1$ is obtained by running another copy that generates the new message. We note that run triplication is not strictly necessary for message duplication, but makes the proof easier. We note, points (1) and (3) above refer to relative ordering between $\thatthread$ and $\thisthread$ timestamps and (2) refers to the new message with an arbitrarily high timestamp.

\subsection{Abstracting the Timestamps}
\noindent We introduce the \emph{timestamp abstraction}, which is a building block for the simplified semantics. 
Let us call a message $\anevent$ an $\thatthread$ ($\thisthread$) message if it is generated by a $\thatthread$ ($\thisthread$) thread. 
With the intuition that $\thatthread$ messages can be replicated with arbitrarily high timestamps, while $\thisthread$ or initial messages cannot be, we  distinguish the write timestamps of the two types of messages.

\smallskip
\noindent{\textbf{Timestamp Abstraction.}}
If a $\thatthread$ thread has read a message $(\xvar, d, \aview)$ from a $\thisthread$ thread with a timestamp $\atimestamp=\aview(\xvar)$ and 
has generated a message $\anevent$ on $\xvar$, then copies of $\anevent$  are available with arbitrarily high timestamps at least as high as $\atimestamp$. To capture this in our abstraction, 
 we assign the $\thatthread$ message $\anevent$, a timestamp $\tplus{\atimestamp}$ that is by definition, larger than~$\atimestamp$. 

We define the set of timestamps in the simplified semantics as 
$\nat \uplus \tplus{\nat}$, where $\tplus{\nat}$ contains for each  $\atimestamp \in \nat$, a timestamp $\tplus{\atimestamp}$. 
The timestamps are equipped with the order $\preceq$ in which $\tplus{\atimestamp}$ is greater than $\atimestamp$ and smaller than $\atimestamp+1$:\\
\begin{equation*}
0 \prec \tplus{0} \prec 1 \prec \tplus{1}\prec \ldots
\end{equation*}

Timestamps of form $\atimestamp \in \mathbb{N}$ are used for the stores of $\thisthread$ threads while those of form $\tplus{\atimestamp}$ are used for stores of $\thatthread$ threads. We allow multiple stores with the same timestamp of form $\tplus{\atimestamp}$, while allowing at most one store for timestamps of form $\atimestamp$. This abstracts timestamps of multiple $\thatthread$ messages between two $\thisthread$ messages by a single $\tplus{\atimestamp}$ timestamp. Initial messages have timestamp 0 as usual.

We utilize this timestamp abstraction by defining a simplified semantics; note that this simplification is not per se a simpler formulation but rather is simple in the sense that it will pave the way for efficient verification procedures (as detailed in Section \ref{sec:consec-lfree}, Section \ref{sec:nexp-c}). 
 We then show that a run $\rho$ in the classical RA semantics has an equivalent run in the simplified semantics where the timestamps are transformed according to some timestamp transformation $\gentmorph$ as defined above. We show that reachability across the two semantics is preserved since both order and consecution between timestamps is maintained.

\noindent\textbf{RA semantics, simplified.} As in the classical RA semantics,  
the transition rules of the simplified semantics will require us to increase timestamps (upon writing messages). We define the  function $\raisets{-}$ on $\nat \uplus \tplus{\nat}$ by
\begin{align*}
    \raisets{\atimestamp} = \raisets{\tplus{\atimestamp}} = \tplus{\atimestamp}, \atimestamp \in \nat. 
\end{align*}
The definition of the simplified semantics replaces the domain $\settime$ by $\ptimestamps = \nat \uplus \tplus{\nat}$. 
We use the term \emph{abstract} to refer to the resulting views, messages, memory, local configurations, and configurations and use a superscript $\abstscriptpshort$ (shortened $\thisthread/\thatthread$) to indicate that an element is abstract. 
So an abstract view is a function, $\aview^{\abstscriptpshort}$ that maps shared variables to $\ptimestamps$. 
We now specify the transitions in the abstract semantics. Owing to their different nature (one is replicatable, the other is not) the $\thisthread$ and $\thatthread$ threads will have different transition rules in the simplified semantics.

For storing a value, the $\thatthread$ threads use a rule \rulestorelocalabstthat\ that coincides with rule \rulestorelocal\ from the RA semantics (Figure \ref{app:TransitionRelation}) except that it replaces relation $<_{\anadr}$ by $<^{\thatthread}_{\anadr}$ defined as follows:

\vspace{-1em}
\begin{align*}
\aview_1 <^{\thatthread}_{\anadr} \aview_2\quad\text{iff}~
\begin{cases}
	\raisets{\aview_1(\anadr)} \preceq \aview_2(\anadr) \in \mathbb{N}^+  & \\ 
	\aview_1(\yvar) = \aview_2(\yvar) \qquad \text{ for } \yvar\neq \anadr.
\end{cases}
\end{align*}
Additionally, for stores of $\thatthread$ threads, we no longer require the timestamp of the message to be unused. So we disregard the $\mathsf{msg} 
 \# \mathsf{m}$ check in the global $\rulestoreglobal$ rule (note crucially this is for $\thatthread$ only). The $\thisthread$ threads use \rulestorelocal\ from the RA semantics without modifications, and hence choose a timestamp in $\nat$, not a raised value. 

For load instructions, we distinguish between messages generated by $\thisthread$ and $\thatthread$ threads. 
This is a natural consequence of the different nature of timestamps, $\atimestamp$ for $\thisthread$ and $\tplus{\atimestamp}$ for $\thatthread$ messages. 
For loading a $\thisthread$ message, we use rule $\ruleloadlocal$ (Figure \ref{app:TransitionRelation}) from the RA semantics without changes. 

For loading from $\thatthread$ threads, we introduce a new rule $\ruleloadlocalabstthat$. 
It is defined by replacing $\join$ in $\ruleloadlocal$ by $\join^{\thatthread}_x$.
We drop the check on the order of timestamps (overwrite it by true); a $\thatthread$ message may always be read, independent of the reading thread's view. 
The join is dependent on the variable being read from, $\anadr$. 
To define $\aview_1\join^{\thatthread}_{\anadr}\aview_2$, let $\aview_1$ be the view of the thread thread loading the message and $\aview_2$ be the view in the message.

\vspace{-1em}
\begin{align*}
\aview_1\join^{\thatthread}_{\anadr}\aview_2 = (\aview_1[\anadr \mapsto \raisets{\aview_1(\xvar)}]) \join \aview_2 
\end{align*}

Thus, if $\aview_1(\xvar)=4$ and $\aview_2(\xvar)=2^+$,
then $(\aview_1\join_{\anadr}^{\thatthread}\aview_2)(\xvar)=4^+$.
The update to $\raisets{\aview_1(\xvar)}$  ensures that if it the timestamp on $\xvar$ was $\atimestamp$, it is at least $\tplus{\atimestamp}$, and hence it cannot read a ($\thisthread$) message with timestamp $\atimestamp$ again. 
We note that the above join operation is not commutative. 

Now we consider the atomic operation - $\rulecaslocal$ - which can only be performed by $\thisthread$. We have two cases depending on whether $\rulecaslocal$ loads from a $\thisthread$ or $\thatthread$ message. If it is the latter, then the transition is identical to ($\ruleloadlocalabstthat$; \rulestorelocal) with the additional condition that the load and store timestamps must be $\tplus{\atimestamp}$ and $\atimestamp+1$ for some $\atimestamp$. 

If it is the former (load from $\thisthread$) then the load and store timestamps must be $\atimestamp$ and $\atimestamp+1$. Consequently, there cannot be any messages with timestamp $\tplus{\atimestamp}$. Conversely, if there is (atleast one) message with timestamp $\tplus{\atimestamp}$, then the $\rulecaslocal$ operation with load and store timestamps $\atimestamp$ and $\atimestamp + 1$ is forbidden. We keep track of such `blocked' intervals $(\atimestamp, \atimestamp+1)$ by adding a set $\mathsf{B}$ to the global state in the simplified semantics. 
The  global and local transition relations of the full simplified semantics are in 
Figure \ref{Figure:AtomicTransitionRelation-app1}, \ref{Figure:AtomicLocalTransitionRelation-app}.

We formally proof equivalence w.r.t. reachability of the simplified semantics with the original RA semantics. But before that, we give an example to illustrate timestamp abstraction in the simplified semantics.

\begin{figure}[t]
\centering
\small
\begin{subfigure}{\textwidth}
\hspace{-1.3cm}
   \centering
\begin{tikzpicture}[codeblock/.style={line width=0.5pt, inner xsep=0pt, inner ysep=5pt}  , show background rectangle]
\node[codeblock] (init) at (current bounding box.north west) {
$
\def\arraystretch{1}
\begin{array}{c}
\text{Variables } \xvar \text{ and } \yvar \text{ have been initialized to 0} \\
\begin{array}{l|l}
   \hline
    \mathsf{producer} &  \mathsf{consumer} \\ \hline 
   \lambda_1:~ \load{\areg_1}{\yvar} & \lambda'_1:~ \store{\yvar}{1} \\
   \lambda_2:~ \texttt{if}(\areg_1 == 1):  & \lambda'_2:~ \texttt{for i in {1..z}}: \\
   \lambda_3:~ \qquad \store{\xvar}{1} \choice \quad\quad\cdots \choice \store{\xvar}{l} & 
   \lambda'_3:~ \quad\quad\load{\areg'_1}{\xvar} \\ 
    & \lambda'_4:~ \quad\quad\assume{\areg'_1 = \texttt{(i \%\% l)+1}} \\ \hline 
\end{array}
\end{array}$
};
\end{tikzpicture}
\end{subfigure}
\begin{subfigure}{\textwidth}
   \centering

   \includegraphics[scale=0.3]{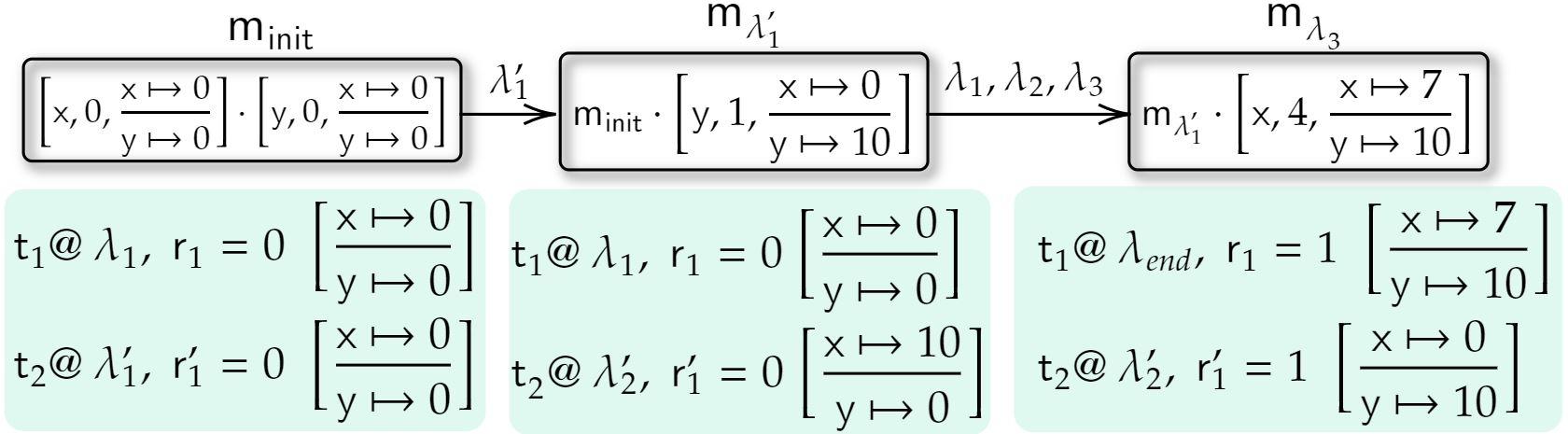}
\end{subfigure}
\caption{Above is a (non-parameterized) producer-consumer program, and below is a sample execution snippet with  threads $t_1$ and $t_2$ playing the roles of producer and consumer respectively.}
\label{Figure:example}
\end{figure}

\newcommand{\prodr}{{\color{red} \mathsf{producer}}}
\newcommand{\consr}{{\color{blue} \mathsf{consumer}}}

\begin{figure}[h]
\centering
\includegraphics[scale=0.15]{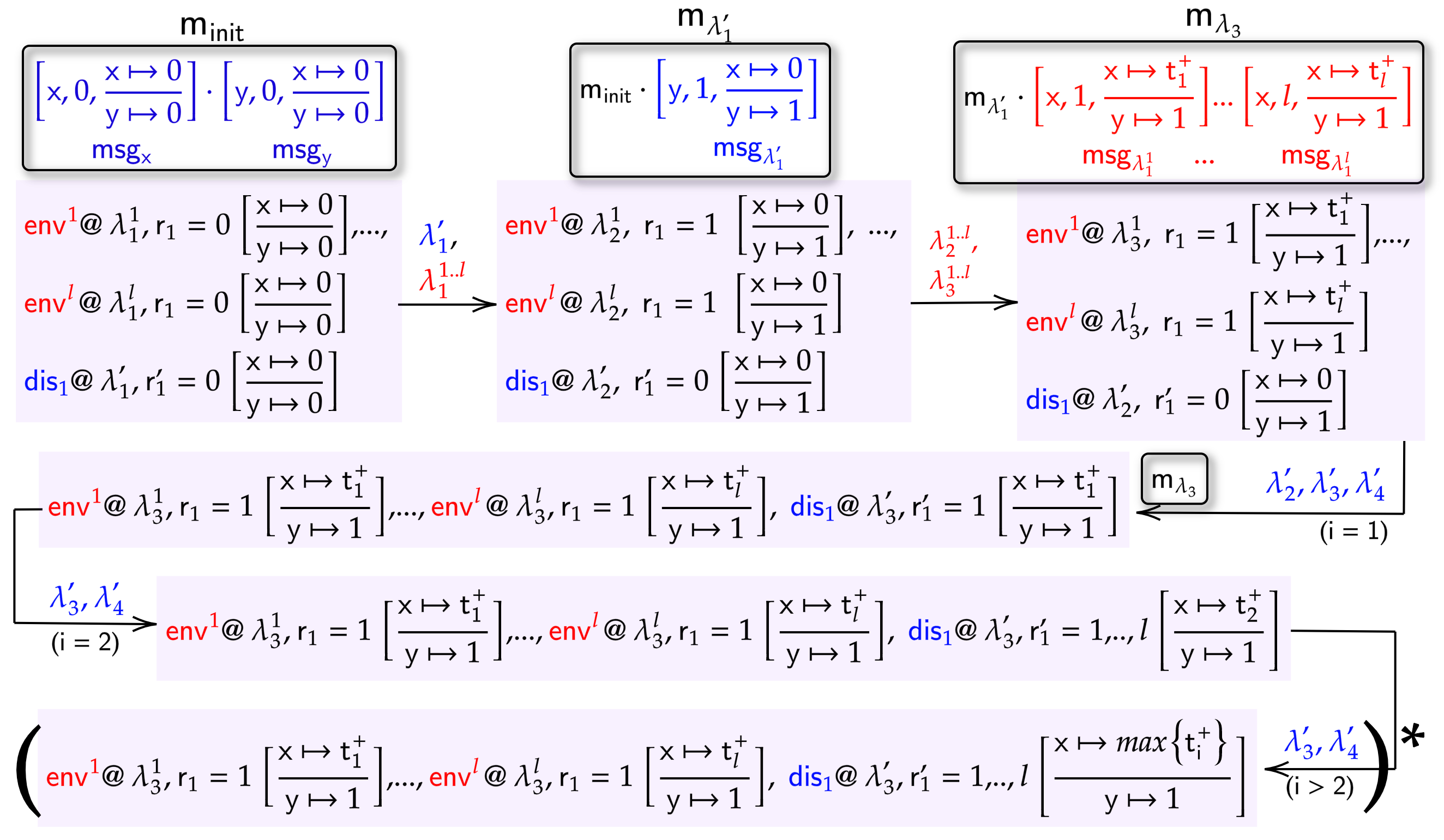}
\caption{Execution under the simplified semantics. $\prodr$  transitions and messages in red, $\consr$ transitions and messages in blue. 
The execution begins with the $\consr$ thread generating a message on $\yvar$ with value 1 and timestamp 1 leading to the memory $\amem_{\lambda'_1}$. 
The $\prodr$ threads executing $\lambda_1^{1\dots l}$ read from this message and reach states $\lambda_2^{1\dots l}$. They generate messages on $\xvar$ with values $\{1,\dots,l\}$ shown in memory $\amem_{\lambda_3}$. 
These are then read by the $\consr$ as it loops around  $\lambda'_3$, $\lambda'_4$ for different iterates $\texttt{i}$, ($\texttt{i=1, i=2, i>2}$) as shown along the transition edge.}
\label{Figure:simplified}
\end{figure}

\noindent\textbf{Simplified Semantics, on an Example}. 
In Figure \ref{Figure:simplified} we give an example of a computation under the simplified semantics by parameterizing the program from Figure \ref{Figure:example}. The parameterized program has a single distinguished $\thisthread_1$ thread 
which runs the program $\consr$,    
and arbitrarily many $\thatthread$ threads which run $\prodr$. 
We consider a computation in which $\thisthread_1$, and $l$ (out of the unboundedly many) $\thatthread$ threads participate.  
We label different instances of the $\thatthread$ threads (and their instruction labels) by superscripts from $\{1,\dots,l\}$ (eg.,
 $\lambda_1^{1}, \dots \lambda_1^{l}$ for $\lambda_1$).

The $\consr$ generates write timestamps of the form $\atimestamp$ (1 in example), while $\prodr$ threads have write timestamps of the form $\atimestamp_1^+, \dots,\atimestamp_l^+$. 
While the timestamp 1 is now occupied, there can be several writes with timestamps of the form $\atimestamp^+$, 
(in particular some $\atimestamp_i^+$ may be equal). 
Additionally, when reading from  these $\prodr$ generated messages, $\consr$ does not perform any timestamp checks, 
rather simply updates its view by taking joins. 
Hence we point out that the load with value 2 during the second loop iteration ($\texttt{i=2}$), is feasible even if $\atimestamp_2^+ < \atimestamp_1^+$; unlike the classical RA semantics.  
In this example, the $\thisthread$ thread,  
after looping around $l$ times and reading from  $\thatthread$ messages, 
has the view on $\xvar$ as $\max\{\atimestamp_i^+\}$. 
Due to the lack of timestamp comparisons,  
$\consr$ can perform the loop arbitrarily many times ($\mathsf{z} > l$ times), 
moreover, the number of $\thatthread$ threads needed is  independent of $\texttt{z}$. 
We see that this relieves the burden of timestamp comparisons, for $\thatthread$ messages. 


\tikzset{background rectangle/.style={fill=none
}}
\begin{figure}[h]
\centering
\small
\begin{tikzpicture}[codeblock/.style={line width=0.5pt, inner xsep=0pt, inner ysep=5pt}  , show background rectangle]
\node[codeblock] (init) at (current bounding box.north west) {
$
\def\arraystretch{3}
\begin{array}{c}
\rowcolor{black!15}
\begin{array}{c}
\ruleloadglobal\quad\inferrule{\anlcfmap(\athread) = \anlcf \quad
\anlcf\localtrans{\mathsf{ld},\anevent}\anlcf'\quad \anevent\in \amem
  }{
  (\amem, \anlcfmap,\mathsf{B})\xrightarrow{(\athread, \anevent)} (\amem, \anlcfmap[\athread\mapsto\anlcf'],\mathsf{B})
  } \\
  \ruleunlabelled\quad\inferrule{\anlcfmap(\athread) = \anlcf \quad
\anlcf\localtrans{}\anlcf'}{
  (\amem, \anlcfmap,\mathsf{B})\xrightarrow{\athread} (\amem, \anlcfmap[\athread\mapsto\anlcf'],\mathsf{B})
  } 
\end{array}
  \\
  \rowcolor{black!5}
  \rulestoreglobal^{\thisthread}\quad\inferrule{\athread\in\thisthread\quad\anlcfmap(\athread) = \anlcf \quad
  \anlcf\localtrans{\mathsf{st},\anevent,\thisthread}\anlcf'\quad}{
  (\amem, \anlcfmap, \mathsf{B})\xrightarrow{(\athread, \anevent)} (\amem\cup\set{\anevent}, \anlcfmap[\athread\mapsto\anlcf'],\mathsf{B})
  } \\
  \rowcolor{black!15}
  \rulestoreglobal^{\thatthread} \vspace{-1em}\\
  \rowcolor{black!15}
   \inferrule{\athread \in \thatthread\quad \anlcfmap(\athread) = \anlcf \quad
  \anlcf\localtrans{\mathsf{st},\anevent,\thatthread}\anlcf'\quad \anevent = (\anadr, \avalue, \abstpview) \quad \abstpview(\anadr) = \atimestamp^+ \quad \atimestamp \not\in \mathsf{B}}{
  (\amem, \anlcfmap,\mathsf{B})\xrightarrow{(\athread, \anevent)} (\amem\cup\set{\anevent}, \anlcfmap[\athread\mapsto\anlcf'],\mathsf{B}\cup \{\atimestamp^+\})} \\
  \rowcolor{black!5}
  \rulecasglobal \vspace{-0.5em}\\
  \rowcolor{black!5}
   \inferrule{\anlcfmap(\athread) = \anlcf \quad
  \anlcf\localtrans{\mathsf{cas},\anevent_r,\anevent_w}\anlcf' \quad \anevent_r \in \amem \quad \anevent_w = (\xvar, \avalue, \abstpview)  \quad \abstpview(\anadr) = \atimestamp + 1 \\\\ 
  \text{ if } \anevent_r(\anadr) \in \mathbb{N} \text{ then } \tplus{\atimestamp} \not\in \mathsf{B}}{
  (\amem, \anlcfmap,\mathsf{B})\xrightarrow{(\athread, \anevent_w)} (\amem\cup\set{\anevent_w}, \anlcfmap[\athread\mapsto\anlcf'],\mathsf{B}\cup \{\atimestamp\})}
\end{array}$
};
\end{tikzpicture}
\caption{Simplified semantics. Global transition relation. $\mathsf{B}$ is a set of blocked time stamps. For an $\thatthread$ thread
making a store operation, the time stamp  $\atimestamp^+ \in \mathbb{N}^+$ 
can be chosen only when $\atimestamp$ has not been blocked ( $\atimestamp \notin \mathsf{B}$). $\atimestamp^+$ is added to $\mathsf{B}$ whenever a 
$\thatthread$ thread makes a store operation adding a message $(\xvar, d, \atimestamp^+)$. 
Likewise, when a $\thisthread$ makes a CAS operation on loading from a 
 message $(\xvar,d,\aview)$ with $\aview(\xvar)=\atimestamp \in \mathbb{N}$, then it must be checked that $\atimestamp^+ \notin \mathsf{B}$, ensuring that there are no time stamps between $\atimestamp$ and $\atimestamp+1$. $\atimestamp \in \mathbb{N}$ is added to $\mathsf{B}$ when a $\thisthread$ 
thread makes a CAS, loading from a message $(\xvar,d,\aview)$ with $\aview(\xvar)=\atimestamp \in \mathbb{N}$.}
\label{Figure:AtomicTransitionRelation-app1}
\end{figure}

\tikzset{background rectangle/.style={fill=none
}}
\begin{figure}[h]
\centering
\small
\begin{tikzpicture}[codeblock/.style={line width=0.5pt, inner xsep=0pt, inner ysep=5pt}  , show background rectangle]
\node[codeblock] (init) at (current bounding box.north west) {
$
\def\arraystretch{2}
\begin{array}{cc}
\rowcolor{black!15}
\hspace{1cm}
\begin{subarray}{c}\rulestorelocalabstthat \\
    \text{store operation for }\thatthread
\end{subarray}
&
\quad\inferrule{
 \avaluation(\areg)=\avalue\quad 
\abstpview_1<^{\thatthread}_{\anadr}\abstpview_2
  }{
  (\store{\xvar}{\areg}
, \avaluation, \abstpview_1)\localtrans{\mathsf{st},(\anadr, \avalue, \abstpview_2), \thatthread} (\myskip, \avaluation, \abstpview_2)
  }
  \\[0.5cm]
  \rowcolor{black!5}
\begin{subarray}{c}\rulestorelocal \\
    \text{store operation for }\thisthread
\end{subarray}
&
\quad\inferrule{
 \avaluation(\areg)=\avalue\quad 
\abstpview_1<_{\anadr}\abstpview_2 \quad \abstpview_2(\anadr) \in \mathbb{N}
  }{
  (\store{\xvar}{\areg}
, \avaluation, \abstpview_1)\localtrans{\mathsf{st},(\anadr, \avalue, \abstpview_2), \thisthread} (\myskip, \avaluation, \abstpview_2)
  } \\[0.5cm]
\rowcolor{black!15}
\begin{subarray}{c}
    \ruleloadlocalabstthat \\
    \text{load from }\thatthread \text{ messages}
\end{subarray}
&
\quad\inferrule{
\avaluation'=\avaluation[\areg\mapsto \avalue] \quad \abstpview_2(\anadr) \in \mathbb{N}^+}{
  (\load{\areg}{x}
, \avaluation, \abstpview_1)\localtrans{\mathsf{ld},(\anadr, \avalue, \abstpview_2)} (\myskip, \avaluation', \abstpview_1\join^{\thatthread}_{\anadr}\abstpview_2)} \\[0.5cm]
  \rowcolor{black!5}
\begin{subarray}{c}
    \ruleloadlocal \\
    \text{load from }\thisthread \text{ messages}
\end{subarray}
&
\quad\inferrule{
\avaluation'=\avaluation[\areg\mapsto \avalue] \quad  \abstpview_1(\anadr) \preceq \abstpview_2(\anadr) \in \mathbb{N}}{
  (\load{\areg}{\xvar}
, \avaluation, \abstpview_1)\localtrans{\mathsf{ld},(\anadr, \avalue, \abstpview_2)} (\myskip, \avaluation', \abstpview_1\join\abstpview_2)} \\[0.5cm]
  \rowcolor{black!15}
\begin{subarray}{c}
    \rulecaslocalabstthat \\
    \text{cas with load from }\thatthread \text{ messages}
\end{subarray}
&
\quad\inferrule{\avaluation(\areg_1) = \avalue_1 \quad \avaluation(\areg_2) = \avalue_2\quad \abstpview_1(\xvar) = \atimestamp^+ \\\\
{\abstpview}' = \abstpview\join^{\thatthread}_{\anadr}\abstpview_1 \quad {\abstpview}'(\anadr) = \atimestamp_1^+ \quad \abstpview_2 = {\abstpview}'[\xvar \mapsto \atimestamp_1 + 1]} {
  (\casOf{\avar, \areg_1, \areg_2}
, \avaluation, \abstpview)\localtrans{\mathsf{cas},(\xvar, \avalue_1, \abstpview_1),(\xvar, \avalue_2, \abstpview_2)}
(\myskip, \avaluation, \abstpview_2)} \\[0.5cm]
  \rowcolor{black!5}
\begin{subarray}{c}
    \rulecaslocal \\
    \text{cas with load from }\thisthread \text{ messages}
\end{subarray}
&
\quad\inferrule{\avaluation(\areg_1) = \avalue_1 \quad \avaluation(\areg_2) = \avalue_2\quad \abstpview_1(\xvar) = \atimestamp \geq \abstpview(\xvar) \\\\
{\abstpview}' = \abstpview \join \abstpview_1 \quad \abstpview_2 = {\abstpview}'[\xvar \mapsto \atimestamp + 1]}{
  (\casOf{\avar, \areg_1, \areg_2}
, \avaluation, \abstpview)\localtrans{\mathsf{cas},(\xvar, \avalue_1, \abstpview_1),(\xvar, \avalue_2, \abstpview_2)} (\myskip, \avaluation, \abstpview_2)}
\end{array}$
};
\end{tikzpicture}
\caption{Simplified semantics. Thread-local transition relation. Margin annotations provide description. The store rules refer to the thread type ($\thisthread$/$\thatthread$) \textit{executing} the instruction; the load rules refer to the thread type \textit{which generated} the message that is being loaded (similarly for the load part of CAS operations which can only be executed by $\thisthread$ threads).
In Rule~\rulestorelocalabstthat, we use $\abstpview_1<_{\anadr}^{\thatthread}\abstpview_2$ to mean 
$\raisets{\abstpview_1(\anadr)}\preceq \abstpview_2(\anadr) \in \mathbb{N}^+$ and $\abstpview_1(\yvar)=\abstpview_2(\yvar)$ for all variables $\yvar\neq\anadr$. 
In Rule~\ruleloadlocalabstthat, $\abstpview_1 \sqcup_{\anadr}^{\thatthread}\abstpview_2$ is defined as ${\abstpview_1}[\xvar \mapsto \raisets{\abstpview_1(\xvar)}] \join \abstpview_2$.
The join $\join$ always means an element wise max over the relevant domain.
}
\label{Figure:AtomicLocalTransitionRelation-app}
\end{figure}

The simplified semantics exactly captures reachability of the original semantics. Define $\abstfuncp$ to be a function which drops all views from messages and local configurations, and define  $\abstequalp$ as  equality of the local configurations modulo views.

\begin{theorem}[Soundness and Completeness]\label{Theorem:Semantics}
If a configuration $\acf$ is reachable under RA, then there is an abstract configuration $\acf^{\abstscriptpshort}$ reachable in the simplified semantics so that $\acf^{\abstscriptpshort} \abstequalp \abstfuncpOf{\acf}$. 
Conversely, if a configuration $\acf^{\abstscriptpshort}$ is reachable in the simplified semantics,  then there is a configuration~$\acf$ reachable under RA such that $\abstfuncpOf{\acf}\abstequalp\acf^{\abstscriptpshort}$.
\end{theorem}
\begin{proof}
	At the outset, we note that the only component of the configuration that differs between the classical and simplified semantics is that of the timestamps and hence the view map, $\aview$ in concrete and $\abstpview$ in the abstract configuration. We now give a relation between these timestamps. With this relation being clear, the formal equivalence between the semantics can be shown by considering a case analysis of the transitions that the threads can take. Once again for quick intuition we take the example of timestamps on a single shared variable $\anadr$ as follows.

\begin{center}
$\begin{array}{cccccccc}
		\init & \thisthreadt^0 & \thisthreadt^1 & \thatthreadt^0 & \thisthreadt^2 & \thatthreadt^1 & \thatthreadt^2 & \thisthreadt^3 \\
		0 & 1 & 2 & 3 & 4 & 5 & 6 & 7
\end{array}$
\end{center}

	This sequence of timestamps corresponds to the following sequence in the abstract semantics.

\begin{center}
$\begin{array}{rcccccccc}
		 & \init & \thisthreadt^0 & \thisthreadt^1 & \thatthreadt^0 & \thisthreadt^2 & \thatthreadt^1 & \thatthreadt^2 & \thisthreadt^3 \\
	\text{concrete} & 0 & 1 & 2 & 3 & 4 & 5 & 6 & 7 \\
	\text{abstract} & 0 & 1 & 2 & \colorbox{blue!5}{$\quad 2^+\quad$} & 3 & \multicolumn{2}{c}{\colorbox{blue!5}{$\qquad 3^+\qquad$}} & 4 
\end{array}$
\end{center}

In this fashion the $\atimestamp^+$ are abstracted $\thatthread$ timestamps between any two $\thisthread$ timestamps. We define the abstraction (similarly concretization) function as the function that transforms all timestamps in the run as shown above. With the above timestamp abstraction/concretization in mind we show that abstract and concrete configurations are equivalent in terms of reachability.

We prove this by induction on the length of a run. We show that a concrete (similarly abstract) configuration is reachable if and only if it has some abstraction (similarly concretization) that is reachable. 

 \noindent{\textbf{Base Case}}. In the base case equivalence is maintained as the initial concrete configuration is equivalent to its simplified configuration where all timestamps are 0. Recall that all timestamp transformations maintain 0 as a fixpoint. Hence the initial thread-local states and memory are equivalent for the concrete and abstract semantics.

\noindent\textbf{Inductive Case - Concrete to Abstract}
For the inductive case, assume that we have the result after $n \in \nat$ steps in a computation. Now we induct by considering cases over 
types of the $n+1$ th instruction in the computation. 

\begin{itemize}
	\item \textbf{Silent}: Silent (thread-local) instructions are handled trivially. They only change the thread local state identically for the concrete and abstract configurations.
	\item \textbf{Load}: A load transition can be either from a $\thisthread$ or a $\thatthread$. For both the cases, we note that the timestamp abstraction maintains (including equality) the relative order of timestamps. Hence whenever a concrete message is readable, so is the corresponding abstract message.  
	\item \textbf{Store}: This follows since the corresponding thread in the abstract configuration can simulate the store using the corresponding  timestamp ($\atimestamp^+ \in \nat^+$ in case of a $\thatthread$ store, and $\atimestamp \in \nat$ in case of a $\thisthread$ store). Note once again that, the abstraction preserves order on the timestamps and consequently, a store is allowed in the abstract semantics if it was allowed in the concrete computation. 
	\item \textbf{CAS}: In this case, we note that the set $\mathsf{B}$ keeps track of which timestamps are allowed for CAS operations. If the CAS operation read from a $\thatthread$ message, the semantics follows from  $\ruleloadlocalabstthat\rulestorelocal$. However, if the CAS load is performed using the store of a $\thisthread$, then it implies that there are no $\thatthread$ timestamps between the load,store timestamps ($\atimestamp, \atimestamp+1$) of the CAS (similar to $\thisthreadt^0$ and $\thisthreadt^1$ in the figure above). Consequently, we see that the set $\mathsf{B}$ in the abstract semantics does not contain the timestamp $\atimestamp^+$ ($\atimestamp^+$ is added to $\mathsf{B}$ the moment a $\thatthread$ makes a store with the timestamp $\atimestamp^+$ to disallow a CAS with load, store timestamps $\atimestamp$ and $\atimestamp+1$). Thus the equivalent CAS operation is also allowed under the abstract semantics.
\end{itemize}

\noindent\textbf{Inductive Case - Concrete to Abstract}
\begin{itemize}
	\item \textbf{Silent}: Silent instructions are handled trivially they only change the thread local state identically for the concrete and abstract configurations.
	\item \textbf{Load}: We consider two cases depending on whether the load happens from a $\thisthread$ thread or a $\thatthread$ thread. 
	\begin{itemize}
	\item In the case where we load from a $\thisthread$ message, the semantics are equivalent between the abstract and concrete transitions, since we compare the timestamps $\atimestamp^+ \prec \atimestamp +1$ and $\atimestamp \prec \atimestamp+1$ (see the rule $\ruleloadlocal$ in Figure \ref{Figure:AtomicLocalTransitionRelation-app}). 	
 	Given that the concretization function (like the abstraction) maintains relative order between $\thisthread$ and $\thatthread$ timestamps, the load is also feasible in the concrete semantics.
	\item In the second case, the load is from a $\thatthread$.  
	By inductive hypothesis, we have the concrete computation till the load transition. In particular, the message $(\xvar,d,\aview)$ we wish to load has already been generated in the concrete computation. To this concrete computation $\rho$ obtained by inductive hypothesis, we invoke the infinite supply lemma with $t^*$ as the reading thread's local view on $\xvar$ to generate the computation $\gentmorph_1(\arun) \superpos \gentmorph_2(\projectto{\arun}{\thatthread}) \superpos \arun_1$ with the fresh message $(\anadr, \avalue, \aview')$. By point (2) in the lemma the message is loadable, $\aview'(\anadr) \geq \tmorph_2^\anadr(t^*)$. Note how we apply the timestamp lifting function to $t^*$ since the reading thread's new concrete timestamp has changed. Additionally by points (1) and (3) the relative order of timestamps in $\aview'$ in variables other than $\anadr$ remain the same w.r.t the $\thisthread$ thread  messages. This implies that after reading the message, the view of the reading thread will only increase on $\xvar$. Hence for all other variables it will remain the same thus maintaining equivalence between the timestamps in the concrete and abstract run.
\end{itemize}
	\item \textbf{Store}: The store transition for $\thisthread$ is identical to its concrete counterpart. For a $\thatthread$ thread, we note that we generate copies of the abstract $\atimestamp^+$ timestamp to get a sequence of concrete timestamps. Here we can generate an arbitrary number of copies and hence the thread, will always find a vacant timestamp for its store.
	\item \textbf{CAS}: When a $\thisthread$ thread  makes a CAS, it can either read from a $\thatthread$ or from the store of a $\thisthread$ thread.   In the latter case let the timestamps of the load,store in the CAS be $\atimestamp$,$\atimestamp+1$. Then in the abstract semantics we require that $\atimestamp^+ \not\in\mathsf{B}$. This implies that in the concrete semantics too, there are no $\thatthread$ timestamps between the load,store timestamps and hence CAS is possible in the concrete semantics too. In the former case we again use the infinite supply lemma as we did in the case of loads, to generate a loadable $\thatthread$ message. 
\end{itemize}

\end{proof}

\section{Safety Verification with Loop-Free Threads}
\label{sec:consec-lfree}
This section discusses the safety verification problem for the class $\thatthread(\uncassy)\parallel\thisthread_1(\unloopy) \parallel \dots \parallel\thisthread_n(\unloopy)$ 
consisting of a set of $n$ distinguished $\thisthread$ threads executing a loop-free program in the presence of an unbounded number of $\thatthread$ threads.
We show that the safety verification problem for this class of systems can be decided in \pspace{} by leveraging the simplified semantics from Section \ref{Section:Simplification}. We will assume that the domain $\domain$ is finite.
Parallely, we demonstrate the ability to improve automatic verification techniques by showing how to encode the safety verification problem (of whether all assertions hold) into 
Datalog programs.
The encoding is interesting for two reasons: (1) it yields a complexity upper bound that, given~\cite{AAAK19}, came as a surprise; (2) it provides practical verification opportunities, considering that Datalog-based Horn-clause solvers are state-of-the-art in program verification \cite{bjorner2013solving,bjorner2015horn}.

\begin{theorem}\label{Theorem:Horn}
The safety verification problem for 
$\thatthread(\uncassy)$ $\parallel  \thisthread_1(\unloopy)\parallel \dots \parallel\thisthread_n(\unloopy)$, $n \in \mathbb{N}$ 
is non-deterministic polynomial-time relative to the query evaluation problem in linear Datalog \emph{(}$\mathsf{NP}^{\mathsf{PSPACE}}$\emph{)}, and hence is in \pspace. \label{thm:pspace}
\end{theorem}

\newcommand{\nondetalgo}{\mathcal{A}\mathsf{lgo}}
\newcommand{\datalogprob}{\mathsf{P}}

We note that the theorem mentions \textit{non-deterministic} polynomial time relative to the linear Datalog oracle. We provide a non-deterministic poly-time procedure $\nondetalgo$, that, given a verification instance converts it to a Datalog problem $\datalogprob$ s.t. (1) for a `yes' verification instance, atleast one execution of $\nondetalgo$ results in $\datalogprob$ having successful query evaluation and (2) for a `no' verification instance, no execution of $\nondetalgo$ leads to the resulting $\datalogprob$ to have successful query evaluation. 

\newcommand{\dlogmem}{\mathsf{Cache}}
Linear Datalog is a syntactically restricted variant of Datalog for which query evaluation is easy to solve (\pspace) at the cost of being inconvenient as an encoding target. Given that we show a \pspace~upper bound on the parameterized safety verification for the class $\thatthread(\uncassy)||\thisthread_1(\unloopy)|| \dots ||\thisthread_n(\unloopy)$, 
in principle, we could have directly encoded the parameterized safety verification problem instance as a linear Datalog program.
For 
convenience of encoding, we do not directly reduce safety verification into query evaluation in linear Datalog, 
but use an intermediate notion of $\dlogmem$ Datalog. 
To make the ideas behind our reduction clear, we proceed in three steps.
\begin{enumerate}
	\item  We introduce $\dlogmem$ Datalog, which is 
 Datalog with an additional parameter, called the $\dlogmem$, that turns out decisive in controlling complexity of encodings 
in the following sense : every $\dlogmem$ Datalog program can be turned into a linear Datalog program at a cost that is linear in the size of the program plus that of the $\dlogmem$ (Lemma~\ref{Proposition:LinearDatalog}),
\item We then show that $\nondetalgo$ generates $\dlogmem$ Datalog problems that satisfy the description from the previous paragraph (Lemma~\ref{Lemma:EncodingCorrect}), and  
\item  We then argue that for all $\dlogmem$ Datalog instances generated by $\nondetalgo$,
a $\dlogmem$ of polynomial size is sufficient for query evaluation
(Lemma~\ref{Proposition:Bound}).
\end{enumerate}
This shows Theorem~\ref{Theorem:Horn}.

\smallskip
\textbf{Linearizing Datalog} A Datalog program $\datalogprog$ \cite{ceri1990syntax} consists of a predicate set $\setpreds$, a data domain $\setdatadomain$, and a set $\setrules$ of rules (also called clauses). 
Each predicate comes with a fixed arity $> 0$. A predicate $P$ of arity $j$ is a mapping from $\setdatadomain^j$ to $\{true, false\}$. 
An \emph{atom} consists of a predicate $P(t_1, \dots, t_j)$ and a list $t_1, \dots, t_j$ of arguments, where each $t_i$ is a term. A term is 
either a variable or a constant; a term is a ground term if it is a constant, and an atom is a ground atom if all its terms 
 are constants. A positive literal is a positive atom $P(t_1, \dots, t_j)$ and a negative literal is a negative atom $\neg P(t_1, \dots, t_j)$, and a ground literal is a ground atom. 
 A rule has the form
\begin{equation*}
	\mathsf{head} ~\mathop{:-}~ \mathsf{body}_1,\ldots, \mathsf{body}_t 
\end{equation*}
where $\mathsf{head}$ and $\mathsf{body}_i$ are \textit{positive} literals.
A rule with one literal in the body is a \textit{linear} rule, one without a body is called a \emph{fact}. A  linear 
Datalog program is one where all rules are linear or are facts. 
An instantiation of a rule is the result of replacing each occurrence of a variable in the rule by a constant. 
For all instantiations of the rule, if all ground atoms constituting the body are true then the ground atom in the head can be inferred to be true. All instantiations of facts are trivially true.  
We write $\datalogprog\vdash \agatom$ to denote that the ground atom $\agatom$ can be inferred from program $\datalogprog$. 

\smallskip

{\textbf{Query Evaluation Problem.}} The \emph{query evaluation problem} for Datalog is, given a \emph{query instance} $(\datalogprog, \agatom)$ consisting of a Datalog 
program  $\datalogprog$ and a ground atom $\agatom$,  
to determine whether $\datalogprog\vdash \agatom$.
When studying the \emph{combined complexity}, both $\datalogprog$ and $\agatom$ are given as input~\cite{vardi1982complexity}. 
It is known \cite{gottlob2003complexity} that combined complexity of query evaluation for linear Datalog 
is 
in \pspace{}, while  allowing non linear rules raises the complexity to \nexp{} (\cite{vardi1982complexity} and \cite{immerman2012descriptive}). 
Motivated by verification, there has been interest in linearizing Datalog \cite{kafle2016solving}. 
\smallskip

{\textbf{Adding $\dlogmem$ to Datalog: $\dlogmem$ Datalog}}. 
We introduce to Datalog the concept of a $\dlogmem$.
A $\dlogmem$
is a set of ground atoms that is used to control the inference process.   
The resulting program is called a $\dlogmem$ Datalog program.  
In the presence of a $\dlogmem$, the semantics of Datalog is adapted by the following two rules.

\noindent\textit{Add}: For an instantiated rule, the ground atom in the head can be inferred and added to $\dlogmem$ only when all the ground atoms in the body are in $\env$.

\noindent\textit{Drop}: Atoms in $\env$ can be dropped non-deterministically. 

The standard semantics of Datalog can be recovered by monotonically adding all inferred atoms (starting with facts) to the $\dlogmem$ and never dropping anything. 
To show the upper bound, we use a notion of inference that takes into account the size of the $\dlogmem$ and minimizes it.  
For a $\dlogmem$ Datalog program $\datalogprog$ and $k\in \mathbb{N}$, we write  $\datalogprog \vdash_{k} \agatom$ to mean that ground atom $\agatom$ can be inferred from $\datalogprog$ with a computation in which $\sizeOf{\env} \leq k$, the number of atoms in $\dlogmem$ is always at most $k$. The $\dlogmem$ size measures the complexity of linearizing  $\dlogmem$ Datalog 
as follows. 

\begin{lemma}
Given a $\dlogmem$ Datalog program $\datalogprog$, a ground atom  $\agatom$, and a bound $k$, in time quadratic in $\sizeOf{\datalogprog}+\sizeOf{g}+k$ we can construct a linear Datalog program $\datalogprog'$ so that $\datalogprog \vdash_k \agatom$ iff $\datalogprog'\vdash \agatom$.
\label{Proposition:LinearDatalog}
\end{lemma}
\begin{proof}
To go from $\dlogmem$ Datalog to linear Datalog, 
the idea is to simulate the  $\dlogmem$ using a new predicate $\env\mathsf{Pred}$ of arity $k$ in the constructed linear Datalog program $\gdatalogprog'$. 
We know that a $\env$ of size $k$ suffices in the $\env$ Datalog program, so 
any  rule  $\mathsf{head} ~\mathop{:-}~ \mathsf{body}_1,\ldots, \mathsf{body}_p$ 
 in the $\env$ Datalog is s.t. $p<k$. 
\begin{enumerate}
	\item 
\textbf{Simulating Cache} Intuitively, the predicate $\env\mathsf{Pred}(t_1, t_2,\cdots, t_k)$ represents that the terms $t_i$ are members of the $\env$. We can simulate the set $\env$ by reshuffling  terms using rules that swap the $i^{th}$ and $j^{th}$ elements with rules of the form,
\begin{equation*}
    \cachepred(t_1, \cdots, t_j, \cdots, t_i, \cdots, t_k) ~\mathop{:-}~ \cachepred(t_1, \cdots, t_i, \cdots, t_j, \cdots, t_k) 
\end{equation*}
There are quadratically many such rules. 
\item \textbf{Rules} Consider a rule $R$ with a body of size $p$ in $\dlogmem$ Datalog as follows.
\begin{equation*}
    \mathsf{head} ~\mathop{:-}~ \mathsf{body}_1,\ldots, \mathsf{body}_p 
\end{equation*}
We convert this into a rule which matches the first $p$ terms of $\env\mathsf{Pred}$ with the elements of the body. If there is such a matching, the term $\mathsf{head}$ can be inferred and added into $\env$. This is simulated by replacing some term amongst $t_i$ with the term in the $\mathsf{head}$ while keeping other terms the same.
\begin{equation*}
    \cachepred(t_1, \cdots, t_i = \mathsf{head}, \cdots, t_k) \mathop{:-} \cachepred(t_1 = \mathsf{body}_1, \cdots, t_p=\mathsf{body}_p, t_{p+1}, \cdots, t_k) 
\end{equation*}
There are $k$ choices for the term to be replaced. Thus we have $k$ new rules per rule in the original program.
\item \textbf{Final Inference} Finally, since we know that each element of $\env$ is true, we add the inference rules,
\begin{equation*}
    t_i ~\mathop{:-}~ \env\mathsf{Pred}(t_1, t_2, \cdots, t_p) \quad \text{ for } 1 \leq i \leq p 
\end{equation*}
Now, $g$ can be generated if $g$ ever enters $\env$, i.e. $\env\mathsf{Pred}(t_1, t_2, \dots, g, \dots, t_p)$ for some other terms $t_i$. Then we can use the above inference rule to infer $g$.

\end{enumerate}
This shows that we need at most quadratically many rules each with a single body, to give us a linear Datalog program. 
\end{proof}

\subsection{Datalog Encoding}
\label{sec:dataenc}
Theorem~\ref{Theorem:Semantics} tells us that  safety verification under RA is equivalent to safety verification in the simplified semantics.
Safety verification in the simplified semantics, in turn, can be reduced to the \emph{Message Generation (MG)} problem.
\begin{quote}
 Given a parametrized system $\com$ and a message $\anevent^\# = (\anadr^*, \avalue^*, \_)$ called  \emph{goal message}, does there exist a  reachable configuration $\acf^{\abstscriptpshort}=(\amem^\abstscriptpshort, \anlcfmap^\abstscriptpshort)$ such that $\anevent^\#\in \amem^\abstscriptpshort$ (for some $\aview^\abstscriptpshort$)?
\end{quote} 

To see the connection between MG and safety verification, note that we can replace each $\assert{\texttt{false}}$ statement in the program by $\anadr^* \coloneqq \avalue^*$ for variable $\anadr^*$ and value $\avalue^*$ unused elsewhere. The system is unsafe if and only if a \emph{goal message} $\anevent^\# = (\anadr^*, \avalue^*, \aview^\abstscriptpshort)$ is generated for some $\aview^\abstscriptpshort$.

\noindent While encoding into Datalog, we non-deterministically guess $\aview^\abstscriptpshort$.  For this, we crucially show that there are only exponentially-many choices of $\aview^\abstscriptpshort$ which need to be enumerated. Henceforth we assume that the queried goal message $\anevent^\#$ can have arbitrary $\aview^\abstscriptpshort$. Given $\com$, $\anevent^\#$, our non-deterministic poly-time procedure $\nondetalgo$ 
satisfies the following, the proof of which is in Sections \ref{app:datalog-invariants}. 
\begin{lemma}\label{Lemma:EncodingCorrect}
Given a parametrized system  $\com$ and a goal message $\anevent^\#$, 
Message Generation (MG) holds  iff there is some execution of $\nondetalgo$ that generates a query instance ($\gdatalogprog, \agatom$) such that $\gdatalogprog\vdash\agatom$.
The construction of $\gdatalogprog$ and $\agatom$
is in (non-deterministic) time polynomial in $\sizeOf{\com}$.
\end{lemma}
\vspace{-0.1cm}
The procedure $\nondetalgo$ generates one query instance $(\gdatalogprog, \agatom)$ per execution. We postpone the full description of $\nondetalgo$  and first give some intuition.
Since the parameterized system consists of $n$ loop-free $\thisthread$ threads, each can execute only linearly-many instructions in their size. The total number of instructions executed (and hence the total number of timestamps used) by the $\thisthread$ threads is polynomial in $|\com_\thisthread|$,  
 the combined size of $\thisthread$ programs (concretely the sum of sizes of individual $\com_{\thisthread}^{i}$ programs). 
$\nondetalgo$ guesses the $\thisthread$ threads part of the computation and generates a query instance $(\gdatalogprog,\agatom)$. 

$\gdatalogprog$ itself uses four main predicates. The {\bfseries{e}}nvironment message predicate $\messPred(\anadr, \avalue, \abstview)$ represents the availability of a $\thatthread$ message on variable $\anadr$ with value $\avalue$ and view $\aview^\abstscriptpshort$. 
The environment thread predicate $\statePred(\curr, \avaluation, \aview^\abstscriptpshort)$ encodes the $\thatthread$ thread configuration, where $\curr$ is the control-state, $\avaluation$ is the register valuation and $\aview^\abstscriptpshort$ is the thread view.    
We also have similar message and thread predicates for $\thisthread$ threads. The distinguished message predicate $\dmessPred(\anadr, \avalue, \abstview)$ represents the availability of a $\thisthread$ message.
Additionally, for each $\thisthread$ thread $i \in [n]$, we have a distinguished thread predicate $\dstatePred_i(\curr	, \avaluation, \aview^\abstscriptpshort)$ that encodes the configurations of $\thisthread[i]$. 

In the set of rules, we have the fact $\dmessPred(\anadr, \avalue_{\init}, \abstview_{\init})$ for each $\anadr \in \varset$ with $\avalue_\init$ the initial value and $\abstview_{\init}$ the initial view. 
We also have (i) facts $\statePred(\lambda_{\init}, \avaluation_{\init}, \abstview_{\init})$ and $\dstatePred[i](\lambda_{\init}, \avaluation_{\init}, \abstview_{\init})$ representing the initial states of both $\thatthread$ and $\thisthread$ threads, (ii) rules corresponding to the $\thatthread$ transitions and   
the guessed $\thisthread$ thread run fragments. Finally, the query atom $\agatom$, is a ground atom  $\in\{\messPred$, $\dmessPred\}$ and captures the goal message $\anevent^\#$ being generated. The instances generated in the non-deterministic branches of $\nondetalgo$ differ only 
due to the guessed $\thisthread$ run and the atom $\agatom$.

We now describe the full Datalog program, also proving Lemma \ref{Lemma:EncodingCorrect}. 

\subsubsection{Procedure $\nondetalgo$ for query instance generation}
\label{app:datalog}
We discuss the details of the procedure $\nondetalgo$ which generates the query instance $(\gdatalogprog, \agatom)$ non-deterministically. We use the following predicates in the constructed Datalog program. 

\begin{itemize}
	\item $\messPred(\anevent)$: the message generation predicate for $\thatthread$ threads, where $\anevent$ is a message;
	\item $\statePred(\curr, \avaluation, \aview^\abstscriptpshort)$ :
	 the thread state predicate for $\thatthread$ threads
	\item $\dmessPred (\anevent)$: the message generation predicate for $\thisthread$ threads, where $\anevent$ is a message;
	\item $\dstatePred[i](\curr, \avaluation, \aview^\abstscriptpshort)$: the thread state predicate, one for each $\thisthread$ thread
	\item $\freePred(\xvar,\atimestamp^+)$ : the timestamp availability predicate, which indicates that a timestamp $\atimestamp$ is not blocked by a CAS operation, per variable. 
\end{itemize}

The Datalog program generated has two parts, one does not depend on the non-deterministic choices made by $\nondetalgo$, while the other does. We describe the former part first, these rules for the Datalog Program are in Figure \ref{Figure:HornEncodingFixed}. The second set of rules, depending on the nondeterministic choice of $\nondetalgo$ is in Figure \ref{Figure:HornEncodingVariable}.

\newcommand{\aninst}{\texttt{i}}

\noindent {\textbf{The first set of rules in the Datalog program}} (Figure \ref{Figure:HornEncodingFixed}).  The facts, in green, provide the ground terms for the $\init$ messages as well as initial state of the $\thisthread$ and $\thatthread$ threads. The orange rules capture the thread local transitions of the $\thatthread$ threads. We deviate a bit from the standard notation for programs here, and instead view them as labelled transition systems. It is easy to see that the two notions are equivalent. The initial state labels are $\lambda^{\thatthread}_{\init}$ for the $\thatthread$ threads and $\lambda^i_\init$ for the $\thisthread$ threads. For a pair of labels, we write $\lambda_1 \xrightarrow{\aninst} \lambda_2$ to denote that $\lambda_2$ can be reached from $\lambda_1$ by executing $\aninst$. In the Datalog program, we have a rule for each such transition in the program. The thread-local transitions are in orange. Loads are in violet (first corresponding to loads from $\thatthread$ messages, the second for loading from $\thisthread$ messages). For loads, the rule requires a term with the message predicate (from which the thread is reading) in the body of the rule. Stores are in pink, the first rule corresponds to the new thread-local state after execution of the store. The second rule corresponds to the generation of a term for the new message (in the head). Though we use some higher order syntax for rules such as $\mathsf{assume}$, $\join$, and $<^{\thatthread}_{\anadr}$ we note that these can be easily translated to pure Datalog with small overhead given the polynomial size of the domain and the constant arity of the predicates. 

\tikzset{background rectangle/.style={fill=none
}}
\begin{figure}[h]
\centering
\small
\begin{subfigure}[h]{\textwidth}
\resizebox{\textwidth}{!}{
\begin{tikzpicture}[codeblock/.style={line width=0.5pt, inner xsep=0pt, inner ysep=5pt}  , show background rectangle]
\node[codeblock] (init) at (current bounding box.north west) {
$
\def\arraystretch{1.8}
\begin{array}{rl|l}
\hline
 &\text{rule} & \text{condition on program of $\thatthread$ threads, $\com_\thatthread$} \\
\hline
\hline
\rowcolor{orange!15}
\statePred(\lambda_2, \avaluation, \abstview) &~\mathop{:-}~ \statePred(\lambda_1, \avaluation, \abstview) & \text{if } \lambda_1 \xrightarrow{\myskip} \lambda_2 \\
\hline
\rowcolor{orange!15}
\statePred(\lambda_2, \avaluation, \abstview) &~\mathop{:-}~ \statePred(\lambda_1, \avaluation \text{ with } \semOf{\anexp}(\avaluation(\vecreg)) \neq 0, \abstview) & \text{if } \lambda_1 \xrightarrow{\assume{\anexpOf{\vecreg}}} \lambda_2 \\
\hline
\rowcolor{orange!15}
\statePred(\lambda_2, \avaluation[\areg \mapsto \anexpOf{\vecreg}], \abstview) &~\mathop{:-}~ \statePred(\lambda_1, \avaluation, \abstview) & \text{if } \lambda_1 \xrightarrow{\assign{\areg}{\anexpOf{\vecreg}}} \lambda_2 \\
\hline
\rowcolor{blue!10}
\statePred(\lambda_2, \avaluation[\areg \leftarrow \avalue], \abstview_1 \join^{\thatthread}_{\anadr} \abstview_2) &~\mathop{:-}~\statePred(\lambda_1, \avaluation, \abstview_1), \messPred(\anadr, \avalue, \abstview_2) & \text{if } \lambda_1 \xrightarrow{\load{\areg}{\anadr}} \lambda_2 \\
\hline
\rowcolor{blue!10}
\statePred(\lambda_2, \avaluation[\areg \leftarrow \avalue], \abstview_1 \join_\anadr \abstview_2) &~\mathop{:-}~\statePred(\lambda_1, \avaluation, \abstview_1), \dmessPred(\anadr, \avalue, \abstview_2), \abstview_1 <_\anadr \abstview_2 & \text{if } \lambda_1 \xrightarrow{\load{\areg}{\anadr}} \lambda_2 \\
\hline
\rowcolor{red!15}
\statePred(\lambda_2, \avaluation, \abstpview_2) &~\mathop{:-}~ \statePred(\lambda_1, \avaluation, \abstview_1), \freePred(\anadr, \abstview_2(\anadr)) & \text{if } \lambda_1 \xrightarrow{\store{\anadr}{\areg}} \lambda_2, \text{ with } \abstpview_1<^{\thatthread}_{\anadr}\abstpview_2 \\
\hline
\rowcolor{red!15}
\messPred(\anadr, \avaluation(\areg), \abstpview_2) &~\mathop{:-}~ \statePred(\lambda_1, \avaluation, \abstview_1), \freePred(\anadr, \abstview_2(\anadr)) & \text{if } \exists \lambda_2.~ \lambda_1 \xrightarrow{\store{\anadr}{\areg}} \lambda_2, \text{ with } \abstpview_1<^{\thatthread}_{\anadr}\abstpview_2 \\ \hline
\end{array}$
};
\end{tikzpicture}
}
\caption{The (fixed) set of rules in the Datalog program encoding the transition system of the $\thatthread$ threads.  Silent transitions (in orange); memory accesses: loads (in violet) and stores (in pink).}
\end{subfigure}
\begin{subfigure}[h]{\textwidth}
\centering
\begin{tikzpicture}[codeblock/.style={line width=0.5pt, inner xsep=0pt, inner ysep=5pt}  , show background rectangle]
\node[codeblock] (init) at (current bounding box.north west) {
$
\def\arraystretch{1.8}
\begin{array}{rl|l}
\hline
 \text{fact}\hspace{3cm} & & \text{comment} \\
 \hline\hline
 \rowcolor{green!10}
\dmessPred(\anadr, \avalue_{\init}, \abstview_{\init}) &~\mathop{:-}~ & \text{ for all variables } \anadr\\
\hline
\rowcolor{green!10}
\statePred(\lambda^\thatthread_{\init}, \avaluation_{\init}, \abstview_{\init}) &~\mathop{:-} & \lambda^\thatthread_{\init}\text{ is initial state of } \thatthread \text{ threads}\\
\hline
\rowcolor{green!10}
\dstatePred[i](\lambda^i_{\init}, \avaluation_{\init}, \abstview_{\init}) &~\mathop{:-} & \lambda^i_{\init}\text{ is initial state of } \thisthread[i] \text{ thread} \\ 
\hline
\end{array}$
};
\end{tikzpicture}
\caption{First set of facts in the Datalog program; these do not depend on the non-deterministic guess made by $\nondetalgo$ for the computation of $\thisthread$ threads. These facts encode the initial configurations of the threads and the initial messages.}
\end{subfigure}
\vspace{0.2cm}
\caption{First set of rules for the Datalog Program. This fixed rule set is independent of non-determinism of $\nondetalgo$.}
\label{Figure:HornEncodingFixed}
\end{figure}

These rules capture completely the $\thatthread$ thread component of the run. As we had mentioned earlier, the component of the query instance that differs due to non-determinism of $\nondetalgo$ is the $\thisthread$ part of the run. Essentially, $\nondetalgo$ guesses in polynomial time the executions of all the $\thisthread$ threads. This is possible since they are loop-free and hence execution lengths are linear in the size of their specifications. We now describe this second part of the Datalog query instance. 

\noindent\textbf{Second set of rules in the Datalog program}(Figure \ref{Figure:HornEncodingVariable}). 
We have a bound on the number of write timestamps that can be used by the $\thisthread$ threads - an easy bound is the combined number of instructions in $\thisthread$ threads, $|\com_\thisthread|$. We will refer to this bound as $T$. By the simplified semantics it suffices to consider the timestamps $\{0, 0^+, \cdots, T, T^+\}$. This follows since, the $\thisthread$ threads perform atmost $T$ writes. Hence we need only $T$ timestamps of the form $\mathbb{N}$. Additionally, we have only one timestamp of the form $\mathbb{N}^+$ between any two timestamps of the form $\mathbb{N}$. This shows that the view terms in the predicates of the  Datalog program can be guessed in polynomial space (since $T$ is polynomial in the input). 

Now for each $\thisthread$ thread $i$, the procedure $\nondetalgo$ non-deterministically guesses the computation $\arun_i$ for $\thisthread[i]$. That is, $\nondetalgo$ guesses the timestamps and the register valuations of $\thisthread_i$ at each configuration in this run, along with the messages $\thisthread_i$ loaded from. Post this, it converts $\arun_i$ to a set of rules which are then added to the earlier set from Figure \ref{Figure:HornEncodingFixed}. 

Consider the computation $\arun_i \equiv \lambda^i_{\init} \xrightarrow{\aninst_1} \lambda_1 \xrightarrow{\aninst_2} \lambda_2 \cdots \xrightarrow{\aninst_{|\arun_i|}} \lambda_{|\arun_i|}$ of length $|\arun_i|$. Let the views of $\thisthread$ thread $i$ at point $j$ in the run be given as $\abstview_j$. Additionally, if $\aninst_j$ is a load instruction, $\nondetalgo$ also guesses the message that was read by the $\thisthread$ thread $i$. Each instruction $\aninst_j$ in this computation is then converted into one amongst the rules in Figure \ref{app:rules-this} depending on the instruction $\aninst_j$ executed, represented in the figure as `condition'. Additionally to encode the $\mathbb{N}^+$ timestamps that have not been occupied by CAS operations (and hence are free to use by the $\thisthread$ threads), we have the rule in \ref{app:rules-free}. 

\begin{figure}[!h]
\centering
\small
\begin{subfigure}[h]{\textwidth}
\resizebox{\textwidth}{!}{
\begin{tikzpicture}[codeblock/.style={line width=0.5pt, inner xsep=0pt, inner ysep=5pt}, show background rectangle]
\node[codeblock] (init) at (current bounding box.north west) {
$\def\arraystretch{2}\begin{array}{rl|l}
\hline
  &\text{rule} & \text{condition on thread transition $\aninst_j$ of computation $\arun_i$ for thread $\thisthread_i$} \\
\hline\hline
\rowcolor{orange!15}
\dstatePred[i](\lambda_{j}, \avaluation, \abstview) &~\mathop{:-}~ \dstatePred[i](\lambda_{j-1}, \avaluation, \abstview) & \aninst_j = \myskip \\
\hline
\rowcolor{orange!15}
\dstatePred[i](\lambda_{j}, \avaluation, \abstview) &~\mathop{:-}~ \dstatePred[i](\lambda_{j-1}, \avaluation, \abstview) & \aninst_j = \assume{\anexpOf{\vecreg}} \land \semOf{\anexp}(\avaluation(\vecreg)) \neq 0 \\
\hline
\rowcolor{orange!15}
\dstatePred[i](\lambda_{j}, \avaluation[\areg \mapsto \assign{\areg}{\anexpOf{\vecreg}}], \abstview) &~\mathop{:-}~ \dstatePred[i](\lambda_{j-1}, \avaluation, \abstview) & \aninst_j = \assign{\areg}{\anexpOf{\vecreg}} \\
\hline
\rowcolor{blue!10}
\dstatePred[i](\lambda_{j}, \avaluation[\areg \leftarrow \avalue], \abstview_1 \join_{\anadr}^{\thatthread} \abstview_2) &~\mathop{:-}~\dstatePred[i](\lambda_{j-1}, \avaluation, \abstview_1), \messPred(\anadr, \avalue, \abstview_2) & 
 \aninst_j = \load{\areg}{\anadr} \land \text{ thread loads } \anevent = (\anadr, \avalue, \abstview_2) \land \abstview_2(\anadr) \in \mathbb{N}^+ \\
\hline
\rowcolor{blue!10}
\dstatePred[i](\lambda_{j}, \avaluation[\areg \leftarrow \avalue], \abstview_1 \join_\anadr \abstview_2) &~\mathop{:-}~\dstatePred[i](\lambda_{j-1}, \avaluation, \abstview_1), \dmessPred(\anadr, \avalue, \abstview_2) & \aninst_j = \load{\areg}{\anadr} \land \text{ thread loads } \anevent = (\anadr, \avalue, \abstview_2) \land \abstview_2(\anadr) \in \mathbb{N} \land \abstpview_1<_{\anadr}\abstpview_2 \\
\hline
\rowcolor{red!15}
\dstatePred[i](\lambda_{j}, \avaluation, \abstpview_2) &~\mathop{:-}~ \dstatePred[i](\lambda_{j-1}, \avaluation, \abstview_1) &  \\
\rowcolor{red!15}
\dmessPred(\anadr, \avaluation(\areg), \abstpview_2) &~\mathop{:-}~ \dstatePred[i](\lambda_{j-1}, \avaluation, \abstview_1) & \multirow{-2}{*}{$\aninst_j = \store{\anadr}{\areg} \land \text{ thread stores } \anevent = (\anadr, \avaluation(\areg), \abstpview_2) \land \abstpview_1<_{\anadr}\abstpview_2$} \\
\hline
\cellcolor{black!5} \dstatePred[i](\lambda_{j}, \avaluation, \abstpview_3) &~\cellcolor{black!5}\mathop{:-}~ \dstatePred[i](\lambda_{j-1}, \avaluation, \abstview_1), \messPred(\anadr, \avalue, \abstview_2) &  \cellcolor{black!5} \aninst_j = \cas(\anadr, \avalue, \areg) \land \abstpview_2(\anadr) \in \mathbb{N}^+, \abstpview = \abstpview_1 \join_{\anadr}^{\thatthread} \abstpview_2,\\
\cellcolor{black!5} \dmessPred(\anadr, \avaluation(\areg), \abstpview_3) &~\cellcolor{black!5}\mathop{:-}~ \dstatePred[i](\lambda_{j-1}, \avaluation, \abstview_1), \messPred(\anadr, \avalue, \abstview_2) & 
\cellcolor{black!5}\qquad\qquad \abstpview(\anadr) = \atimestamp^+, \abstpview_3 = \abstpview[\anadr \rightarrow \atimestamp+1] \\
\hline
\cellcolor{black!5} \dstatePred[i](\lambda_{j}, \avaluation, \abstpview_3) &~\cellcolor{black!5}\mathop{:-}~ \dstatePred[i](\lambda_{j-1}, \avaluation, \abstview_1), \dmessPred(\anadr, \avalue, \abstview_2) &  \cellcolor{black!5} \aninst_j = \cas(\anadr, \avalue, \areg) \land \abstpview_1(\anadr) \leq \atimestamp = \abstpview_2(\anadr) \\
\cellcolor{black!5} \dmessPred(\anadr, \avaluation(\areg), \abstpview_3) &~\cellcolor{black!5}\mathop{:-}~ \dstatePred[i](\lambda_{j-1}, \avaluation, \abstview_1), \dmessPred(\anadr, \avalue, \abstview_2) & 
\qquad\qquad\cellcolor{black!5} \abstpview = \abstpview_1 \join_\anadr \abstpview_2, \abstpview_3 = \abstpview[\anadr \rightarrow \atimestamp+1] \\\hline
\end{array}$
};
\end{tikzpicture}
}
\caption{These rules are chosen depending upon the nondeterministic choice made by $\nondetalgo$ of the computation $\arun_i$ of thread $i$. Each instruction $\aninst_j$ executing in $\arun_j$ is then mapped to one of the rules from above depending upon which condition (right column) is satisfied. Rules for silent $\aninst_j$ (in orange); memory accesses, loads (in violet) and stores (in pink), and CAS (gray). The second pink rule corresponds to message generation by the thread $i$ executing a store instruction. The first two CAS rules correspond to the case where the load is from a $\thatthread$ message and the last two correspond to the load from a $\thisthread$ message. The first rule for each case is the thread-local state change rule while the second rule generates the ground atom corresponding to the message generated by the CAS operation.}
\label{app:rules-this}
\end{subfigure}
\vspace{.2cm}
\begin{subfigure}[h]{\textwidth}
\resizebox{\textwidth}{!}{
\begin{tikzpicture}[codeblock/.style={line width=0.5pt, inner xsep=0pt, inner ysep=5pt}  , show background rectangle]
\node[codeblock] (init) at (current bounding box.north west) {
$\def\arraystretch{1.5}
\begin{array}{rl|l}
\hline
\hspace{3cm}\text{fact} &  \hspace{3cm}  & \hspace{3cm} \text{condition on availability of timestamps} \\ \hline\hline
\rowcolor{green!15}
\hspace{3cm}\freePred(\anadr, \atimestamp^+) &~\mathop{:-} \hspace{3cm}  &\hspace{3cm} \atimestamp \in\{0, \cdots T\} \land \text{ no } \thisthread \text{ performs a cas operation with timestamps } (\atimestamp, \atimestamp+1) \\ \hline
\end{array}$
};
\end{tikzpicture}
}
\caption{This fact corresponds to the availability of a $\mathbb{N}^+$ timestamp for stores by $\thatthread$ threads, which is known once all the $\thisthread$ computations have been guessed. These rules are not generated on a per-$\thisthread$ thread basis but rather once the computations $\arun_i$ for all $\thisthread$ threads have been non-determinsitically guessed. Referring to the simplified semantics, this rule captures $\atimestamp \not\in \mathsf{B}$, the fact that there is no CAS operation with timestamps $(\atimestamp, \atimestamp+1)$. Note that the $\freePred$ predicate 
plays a role in inferring the $\thatthread$ thread state and message predicates as seen in the last two rows in Figure \ref{Figure:HornEncodingFixed} : 
we can infer the $\thatthread$ thread state and message predicates with a view $\aview(\xvar)$ only when the respective timestamp 
is not blocked. This in turn is used in the first CAS operation (first gray row, Figure \ref{Figure:HornEncodingVariable}(a))
) when loads happen from a $\thatthread$ thread :  
the merged view $\abstpview(\xvar) = \abstpview_1(\xvar) \join_{\anadr}^{\thatthread} \abstpview_2(\xvar)$ is $\atimestamp^+$, and the new timestamp after CAS is $\atimestamp$+1. 
Note that this is possible since (i) if the timestamp of the $\thatthread$ thread for $\xvar$ from where we load was $\atimestamp^+$, then  there is no $\thisthread$ thread with a timestamp $\atimestamp$ for $\xvar$ and, (ii) if the timestamp of the $\thatthread$ thread from where we load was $\prec \atimestamp^+$, 
then the timestamp of the $\thisthread$ thread performing the CAS was $\atimestamp$ for $\xvar$. In  both cases, the timestamp after CAS will be $\atimestamp$+1.
 }
\label{app:rules-free}
\end{subfigure}
\vspace{0.2cm}
\caption{Second set of  rules for the Datalog Program. This rule set depends on the nondeterministic choice made by $\nondetalgo$ for the computations of the $\thatthread$ threads}
\label{Figure:HornEncodingVariable}
\end{figure}

Since we have the polynomial bound on $T$, it is easy to see that the rules above for the run $\arun_i$ executed by each $\thisthread$ thread $i$ can be generated in polynomial-time after nondeterministically guessing $\arun_i$. These (non-determinism dependent) rules along with the rules from Figure \ref{Figure:HornEncodingFixed} together form the complete Datalog program.

\subsubsection{Invariants for (Datalog Inference $\leftrightarrow$ Computations in the Simplified Semantics) and proof of Lemma 
\ref{Lemma:EncodingCorrect}}
\label{app:datalog-invariants}
Now we see how an inference process in the (complete) Datalog program corresponds to a computation in the simplified semantics. To do this, we give invariants which relate the inference of atoms in the Datalog program with the existence of events in the computation. These invariants together imply the equivalence between an inference sequence in the Datalog program and a computation of $\com$ under the the simplified semantics. Finally, if the goal message $\anevent^{\#}$ is reachable 
at the end of a computation $\rho$ of $\com$, then correspondingly, thanks to the invariants we obtain, we can also infer the ground term 
$\mathsf{g}$ being $\dmessPred(\anevent^{\#})$ or $\messPred(\anevent^{\#})$ in the Datalog program  depending on whether the goal message was generated by a $\thisthread$ thread 
or $\thatthread$ thread in the computation.

\begin{enumerate}
	\item 
\noindent\textbf{$\thatthread$ thread-local state} invariant
\begin{quote}
The ground atom $\statePred(\lambda, \avaluation, \abstview)$ can be inferred iff some $\thatthread$ thread can reach the $\anlcf = (\lambda, \avaluation, \abstview)$. 
\end{quote}
This says that some $\thatthread$ thread is able to reach the state from its transition system with label $\lambda$  such that the thread-local view and the register valuation at that time are $\abstview$ and $\avaluation$ respectively. We can prove that this holds by induction on the length of the run and by noting from the Datalog rules (Figure \ref{Figure:HornEncodingFixed}) that there is a transition $\lambda \xrightarrow{\aninst} \lambda'$ whenever there is a rule corresponding to that. Additionally, the load rules (in blue) require that the corresponding message atoms $(\messPred/\dmessPred)$ holds which as we will see below implies the possibility of the  generation of a message in the memory.

\item \noindent\textbf{$\thatthread$ thread message} invariant
\begin{quote}
The ground atom $\messPred(\anadr, \avalue, \abstview)$ can be inferred if the corresponding message $\anevent = (\anadr, \avalue, \abstview)$ can be generated in the simplified semantics by some $\thatthread$ thread. 
\end{quote}
Note that a ground atom of the form $\messPred(\anadr, \avalue, \abstview)$ can only be inferred using the last rule in Figure 
\ref{Figure:HornEncodingFixed}. 
The body of this rule contains the term $\statePred(\lambda, \avaluation, \abstview)$. This, if true, implies that some $\thatthread$ thread can reach the corresponding thread-local state by the first invariant (above). An $\thatthread$ thread in this thread-local state can generate the message $(\anadr, \avaluation(\areg), {\abstview}')$ since there is an outgoing transition from $\lambda$ with instruction $\anadr:=\areg$. Note that we have the check 
on the existence of the transition to ensure that the message can indeed be generated.  
This is required for the last rule to exist in the program. This implies that the message can in fact be generated.

\item \noindent\textbf{$\thisthread$ thread-local state}
For each $\thisthread$ thread $i$, we have the following invariant.
\begin{quote}
The ground atom $\dstatePred[i](\lambda, \avaluation, \abstview)$ can be inferred iff the $\thisthread$ thread $i$ can reach the $\anlcf = (\lambda, \avaluation, \abstview)$. 
\end{quote}
This is just the $\thisthread$ analog of the first invariant for $\thatthread$ threads. This can also be proved by induction on the length of the run (or the inference sequence). Also analogous to the invariant for the $\messPred$, we have an invariant for $\thisthread$ messages.

\item \noindent\textbf{$\thisthread$ thread message} invariant
\begin{quote}
The ground atom $\dmessPred(\anadr, \avalue, \abstview)$ can be inferred if the corresponding ($\thisthread$) message $\anevent = (\anadr, \avalue, \abstview)$ can be generated in the simplified semantics by some $\thisthread$ thread. 
\end{quote}

\item \noindent\textbf{$\freePred$ timestamp availability} invariant
\begin{quote}
If the fact $\freePred(\anadr, \atimestamp^+)$ is in the Datalog program then we have $\atimestamp \not \in \mathsf{B}$ throughout the computation $\arun$ of the simplified semantics.
\end{quote}

The base case for message predicates holds since the \emph{facts} $\dmessPred(\anadr, \avalue_\init, \abstview_\init)$ for all variables are given in the Datalog program. The base case for thread state predicates holds due to the fact  $\statePred(\lambda_\init, \avaluation_\init, \abstview_\init)$ which captures the initial state of the thread. The inductive steps can be formally proved by considering a computation $\arun$ under the simplified semantics and mapping each transition in $\arun$ to an inference step in the Datalog program. For the converse, we assume an inference sequence (a sequence of invocations of the rules) and, for each rule invoked to infer a new ground atom, we show that a corresponding transition can be taken by a thread in the simplified semantics so that the invariants are maintained. This in turn, is done by taking cases on the next instruction to be executed.
\end{enumerate}

The equivalence between transitions in a computation (hence a computation $\rho$) of the simplified semantics 
and the application of rules/facts in the Datalog program, leading to reachability 
of some message $\anevent$ in $\rho$ iff the corresponding ground term $\dmessPred(\anevent)$ 
or $\messPred(\anevent)$ is inferred in Datalog 
is sufficient to prove Lemma \ref{Lemma:EncodingCorrect}. In particular, the generation of $\anabstevent$ in some computation $\arun$ of $\com$ gives a sub-computation $\projectto{\arun}{\thisthread}$ performed by the $\thisthread$ threads. We consider the Datalog query instance $(\gdatalogprog, \agatom)$ generated where $\nondetalgo$ correctly guesses $\projectto{\arun}{\thisthread}$. By the message generation invariant, the ground atom $\dmessPred(\anabstevent)$ or 
$\messPred(\anabstevent)$ 
corresponding to $\anabstevent$ can be inferred, $\gdatalogprog\vdash\agatom$, giving the forward direction of the lemma.

For the reverse direction, we note that $\gdatalogprog \vdash \messPred(\anabstevent)$ or 
$\gdatalogprog \vdash \dmessPred(\anabstevent)$
immediately implies that the message can be generated in some computation of the system $\com$ (the $\thisthread$ computation is already determined in the guessed program $\gdatalogprog$).

\subsection{$\env$ Size}\label{Section:Environment}
Having described the encoding(s), the challenge now is to provide a polynomial bound on the cache size for the query instances generated by $\nondetalgo$. The $\env$ behaves like a memoized set of atoms which are used for the inference process. The reason  why a polynomial sized $\env$ suffices is that we can ``forget'' (remove from $\env$) previously inferred atoms when they are not being actively used. We use this crucially in the context of $\thatthread$ predicates, $\messPred, \statePred$. Technically this is possible since the arbitrary replication property of $\thatthread$ threads allows us to ``forget'' the state of the previously simulated $\thatthread$ thread and simulate a fresh copy instead.

Let $\quantity_0 = \sizeOf{\domain}\sizeOf{\varset} + |\com_\thisthread|$. We show that a $\dlogmem$ of size $\bigO{\quantity_0^2}$ is sufficient to infer $\agatom$. 
\begin{lemma}\label{Proposition:Bound}
For each $(\gdatalogprog, \agatom)$ generated by $\nondetalgo$, $\gdatalogprog\vdash\agatom$ if and only if $\gdatalogprog\vdash_k\agatom$ with $k\in\bigO{\quantity_0^2}$.
\end{lemma}
An inference sequence performed on $\gdatalogprog$ corresponds to a computation of the parameterized system $\com$ in the simplified semantics (Section \ref{app:datalog-invariants}). Hence, to see that the above size of $\dlogmem$ is sufficient
we analyze the structure of computations in the simplified semantics. 
The analysis will reveal a dependency relation among the messages generated. We will see that this gives enough information to guide the Datalog computation so as to use small sized $\dlogmem$.

Consider a computation $\abstrun$ ending in the configuration $\lastOf{\abstrun} = (\abstmem, \abstlcfmap)$. 
For every message $\anabstevent$ in $\abstmem$, we define $\genproc(\anabstevent)$ as the first thread which added $\anabstevent$ to the memory $\abstmem$.
(Recall that the simplified semantics admits repeated insertions for $\thatthread$ messages due to reuse of timestamps from $\nat^+$). 
We define $\dep(\anabstevent)$ as the set of messages which $\genproc(\anabstevent)$ reads from, before generating the first instance of $\anabstevent$. We define the notion of a dependency graph 
for a computation $\abstrun$. 
\begin{definition}
The \textit{dependency graph} of a computation $\abstrun$ with $\lastOf{\abstrun} = (\abstmem, \abstlcfmap)$ is the directed graph $\depgraph_{\abstrun}=(V, E)$ whose vertices $V = \abstmem$ are the messages in the final configuration and whose edges reflect the dependencies, $(\anabstevent_1, \anabstevent_2) \in E$ if $\anabstevent_1 \in \dep(\anabstevent_2)$.
\end{definition}
As $\dep(-)$ is based on the linear order of the computation, 
the dependency graph is acyclic. The acyclicity of dependency graphs follows immediately from the definition of $\dep{}$. If there is a cycle, then all the threads involved in the cycle would be dependent on each other for the first generation of the respective message, thus causing a deadlock.
We denote the sets of sink and source vertices of $\depgraph$ by $\sinkv(\depgraph)$ resp. $\sourcev(\depgraph)$. A path in $G$ 
is also called a \emph{dependency sequence}. A path or dependency sequence  $m_1 \rightarrow m_2 \rightarrow m_3 \rightarrow \dots m_{n-1} \rightarrow m_n$ 
thus says that $m_1$ was read by some thread which generated  $m_2$, $m_2$ in turn was read by a thread which generated $m_3$ and so on till 
the thread which generated $m_n$ read $m_{n-1}$. Given such a sequence, 
we say $m_i$ is an ancestor of $m_j$ if $i < j$.  
The height of a vertex $v$ is the length of a longest path from a source vertex to~$v$. 
The maximal height over all vertices is $\heightv(\depgraph)$.
See Figure \ref{fig:depgraph} for an example. 
\lstdefinestyle{examplec}{
	backgroundcolor = \color{white},
  breaklines=true,
  postbreak=\mbox{\textcolor{red}{$\hookrightarrow$}\space},
  xleftmargin=\parindent,
  language=C,
  tabsize=2,
  showstringspaces=false,
  basicstyle=\small\ttfamily,
  keywordstyle=\bfseries\color{green!40!black},
  morekeywords={load, store},
  morekeywords=[2]{memory_order_acquire, memory_order_release},
  keywordstyle=[2]\bfseries\color{orange!40!black},
  commentstyle=\itshape\color{purple!40!black},
  identifierstyle=\color{blue},
  stringstyle=\color{orange}
}
\begin{figure}[!h]
\begin{minipage}{.6\textwidth}
\scriptsize
\raggedleft
\newcommand{\tempview}[2]{\overline{#1#2}}

\tikzstyle{init message}=[fill={rgb,255: red,211; green,211; blue,211}, draw=white, shape=rectangle, tikzit draw=white, tikzit fill={rgb,255: red,211; green,211; blue,211}, tikzit shape=rectangle]
\tikzstyle{contributor1}=[fill={rgb,255: red,255; green,220; blue,180}, draw=white, shape=rectangle, tikzit draw=white, tikzit fill={rgb,255: red,255; green,220; blue,180}, tikzit shape=rectangle]
\tikzstyle{contributor2}=[fill={rgb,255: red,230; green,170; blue,232}, draw=white, shape=rectangle, tikzit fill={rgb,255: red,230; green,170; blue,232}, tikzit draw=white, tikzit shape=rectangle]

\tikzstyle{dependency}=[->]
\begin{tikzpicture}
	\begin{pgfonlayer}{nodelayer}
		\node [style=init message] (4) at (-2, 0) {$(\mathsf{y},0,\tempview{0}{0})$};
		\node [style=init message] (5) at (0.5, 0) {$(\mathsf{x},0,\tempview{0}{0})$};
		\node [style=contributor2] (6) at (0.5, 1) {$(\mathsf{y},1,\tempview{0}{0^+})$};
		\node [style=contributor1] (7) at (-1.2, 1) {$(\mathsf{x},1,\tempview{0^+}{0^+})$};
		\node [style=contributor1] (9) at (-2, 2) {$(\mathsf{y},2,\tempview{0^+}{0^+})$};
		
		\node [style=init message] (10) at (2.25, 0) {$(\mathsf{y},0,\tempview{0}{0})$};
		\node [style=init message] (11) at (5, 0) {$(\mathsf{x},0,\tempview{0}{0})$};
		\node [style=contributor2] (12) at (4, 1) {$(\mathsf{y},1,\tempview{0}{0^+})$};
		\node [style=contributor1] (13) at (2.25, 1) {$(\mathsf{x},1,\tempview{0^+}{0^+})$};
		\node [style=contributor2] (14) at (5, 2) {$(\mathsf{y},2,\tempview{0^+}{0^+})$};
	\end{pgfonlayer}
	\begin{pgfonlayer}{edgelayer}
		\draw [style=dependency] (5) to (6);
		\draw [style=dependency] (6) to (7);
		\draw [style=dependency] (4) to (7);
		\draw [style=dependency, bend left, looseness=1.25] (4) to (9);
		\draw [style=dependency] (6) to (9);
		\draw [style=dependency] (11) to (12);
		\draw [style=dependency] (10) to (13);
		\draw [style=dependency] (13) to (14);
		\draw [style=dependency, bend right, looseness=1.25] (11) to (14);
		\draw [style=dependency] (12) to (13);
	\end{pgfonlayer}
\end{tikzpicture}
\end{minipage} \hspace{1em}
\begin{minipage}{.3\textwidth}
\raggedright
\begin{tabular}{c}
\texttt{x} = 0 = \texttt{y} \\
\begin{tabular}{
c||c}
{\color{orange} $T_1$} & {\color{violet} $T_2$} 
\tabularnewline
\begin{lstlisting}[style=examplec]
y.load(0)
y.load(1)
x.store(1)
y.store(2)
\end{lstlisting}
&
\begin{lstlisting}[style=examplec]
x.load(0)
y.store(1)
x.load(1)
y.store(2)
\end{lstlisting}
\end{tabular}
\end{tabular}
\end{minipage}
\caption{Two possible dependency graphs for the code snippet. 
$T_1, T_2$ are both $\thatthread$ threads. The color of each message $\anabstevent$ signifies $\genproc(\anabstevent)$ ($T_1$ orange, $T_2$ violet, $\init$ gray). We denote the view as a vector $\overline{t_x t_y}$. Since we only consider the thread adding a message for the first time $\genproc(\mathsf{y}, 2, \overline{0^+0^+})$ can be either $T_1$ (left graph) or $T_2$ (right graph).
}
\label{fig:depgraph}
\end{figure}

\noindent {\textbf{Compact Computations}}. Unfortunately, dependency graphs may contain exponentially many vertices (due to the views), and given the \pspace-hardness in Section \ref{sec:pspace-h} 
there is no way to reduce this to polynomial size. 
Yet, there are two parameters that we can reduce, the `fan-in' of each vertex $v$ (number of messages read by $\genproc(v)$ before generating $v$), and 
and the `height' of the dependency graph (longest dependency sequence). 
A computation~$\abstrun$ is \emph{compact} if its dependency graph $\depgraph_{\abstrun}$ satisfies the following two bounds.
(1) Every message $v$ depends on a small number of other messages, $\sizeOf{\dep(v)} \leq \quantity_0$.
(2) The dependency sequences are polynomially long, that is, $\heightv(\depgraph_{\abstrun}) \leq \quantity_0$. 
The following lemma  says that compact computations are sufficient: 

\begin{lemma}\label{lem:polyb}
Any message that can be generated in the simplified semantics, can be generated by a compact computation. 
\end{lemma}

\begin{figure}[!ht]
\begin{subfigure}{.5\linewidth}
\centering
\includegraphics[scale=0.3]{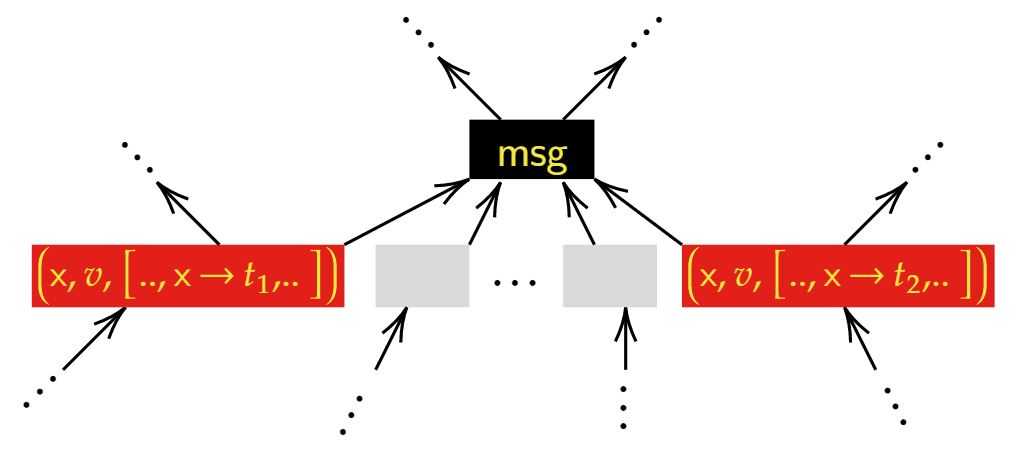}
\end{subfigure}%
\begin{subfigure}{.5\linewidth}
\centering
\includegraphics[scale=0.3]{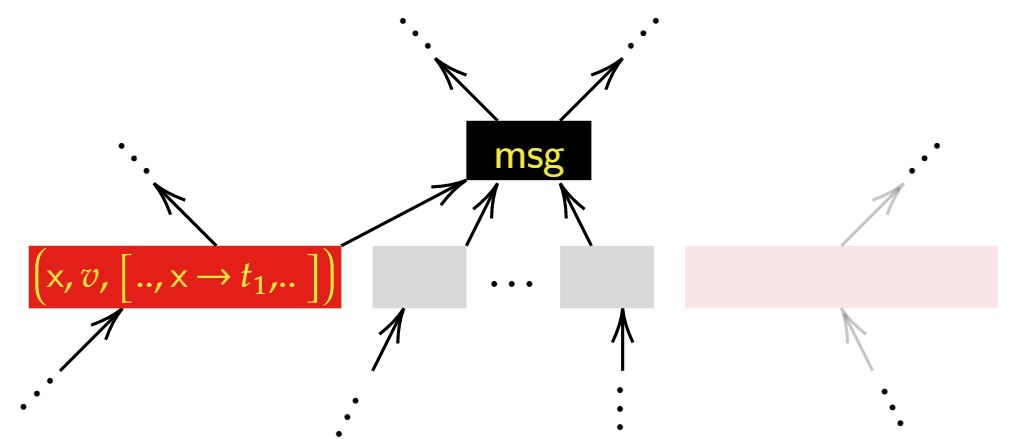}
\end{subfigure}\\[1ex]
\caption{Fan-in reduction: the dependency graph on the left can be converted to the one on the right by eliminating the redundant dependency of thread $\genproc(\anevent)$ on $(\xvar, \avalue, [\dots, \xvar \rightarrow \atimestamp_2, \dots])$ when $\atimestamp_2 > \atimestamp_1$.}
\label{fig:small-fanin}
\end{figure}

\begin{proof}
We prove both parts (fan-in and height) of this lemma by showing that if there exists a computation whose dependency graph violates the bound for fan-in (similarly height), then there must exist a computation whose dependency graph  has a lower fan-in (height) with the rest of the graph (fan-ins of other vertices) unchanged. We first show this for fan-in. We will assume that the programs $\thisthreadprogsimple$ executed by $\thisthread$ threads have been specified as a transition system (note that we can interconvert between the while-language and transition system representation with only polynomial blowup). Then $|\thisthreadprogsimple|$ is an upper bound on the total number of transitions in all $\thisthread$ threads together.

\begin{enumerate}
	\item \textbf{Fan-in}. Suppose to the contrary, we had $|\dep(v)| > 2|\domain||\varset| + |\thisthreadprogsimple|$ for some message $v$. Consider the thread $p = \genproc(v)$ which generated the message represented by vertex $v$ for the first time. There are only $|\domain||\varset|$ distinct (variable, value) pairs, $|\domain||\varset|$ many $\init$ messages and only $|\thisthreadprogsimple|$ many $\thisthread$ messages ($|\thisthreadprogsimple|$ is an upper bound on the number of transitions the $\thisthread$ thread can take). Hence by a pigeonhole argument, $p$ must have read two $\thatthread$ messages with same (variable, value) pair but distinct abstract views. Let these messages be $m_1 = (\anadr, \avalue, \abstview_1)$ and $m_2 = (\anadr, \avalue, \abstview_2)$ where the abstract views are unequal. 
	 
Without loss of generality assume that  $p=\genproc(v)$  read $m_1$ first, and $m_2$ later (in order)  
before it generated $v$.  

It can be seen that any time $p$ read $m_2$, it could have read $m_1$ instead. This follows since timestamp comparisons are irrelevant when reading from $\thatthread$ messages.
The thread-local view obtained on replacing a read of $m_2$ with that of $m_1$ will only decrease or remain the same.
 From the simplified semantics,  after reading $m_1$ once, the thread view $\abstview$ satisfies $\abstview \sqsupseteq \abstview_1$ (per-variable). Hence reading from $m_1$ again leads to the thread view being ${\abstview}' >_{\anadr}^{\thatthread} \abstview$.  On the other hand, after reading $m_2$, the view will be  ${\abstview}' \join_{\anadr}^{\thatthread} \abstview_2$ which is clearly higher than ${\abstview}'$. Indeed, instead of reading from $m_2$, 
 the loading thread  can read from $m_1$, resulting in a lower view for $\xvar$ (compared to reading from $m_2$). 
 \smallskip 
 
 Let $\arun'$ denote the subcomputation starting from the position right after reading from the $\thatthread$ message $m_2$. We can see that if we replace 
 this read operation by reading from $m_1$, we can continue on $\arun'$ as before. 
 \begin{itemize}
 	\item Indeed, all 
 store operations on $\rho'$ are independent of this load from $m_1$ (or $m_2$). 
 \item Consider a load operation 
 along $\rho'$. A load on a variable $\yvar \neq \xvar$ is not affected clearly. Consider now a load on $\xvar$ performed by loading some message $m_3$. 
 Assume the load is performed  by $p$. 
   The view 
 of $\xvar$ along $\rho'$ for thread $p$ was coming from $m_2$ which was a least that given by $m_1$; indeed if loading from $m_3$ was possible 
 in $\rho'$ when the view on $\xvar$ was at least $\aview_2(\xvar)$, it definitely is possible now with a lower view on $\xvar$. 
\item  Lastly, consider a CAS operation on variable $\xvar$, along $\rho'$. Assume 
 the load was made from $m_2$. If $\abstview(\xvar)=\atimestamp_2^+$ in $m_2$, then the CAS operation 
  will add a new message $m_3$ on $\xvar$ with $\aview_3(\xvar)=\atimestamp_2+1$. However, note that 
  the same thread can still perform  CAS by reading from $m_1$, with $\abstview(\xvar)=\atimestamp_1^+$ in $m_1$ (with 
  $\atimestamp_1^+ < \atimestamp_2^+$) by adding 
 a new message $m_3$ on $\xvar$ with $\aview_3(\xvar)=\atimestamp_1+1$.
 
 \end{itemize}
 Hence reading from $m_1$ instead of $m_2$ does not affect the sub computation $\rho'$. Thus we can eliminate all reads of $m_2$ to decrease $|\dep(v)|$.  Thus, $|\dep(v)| \leq  2|\domain||\varset| + |\thisthreadprogsimple|$ for each vertex $v$. 
\smallskip 

\begin{figure}[!ht]
\begin{subfigure}{.5\linewidth}
\centering
\includegraphics[scale=0.3]{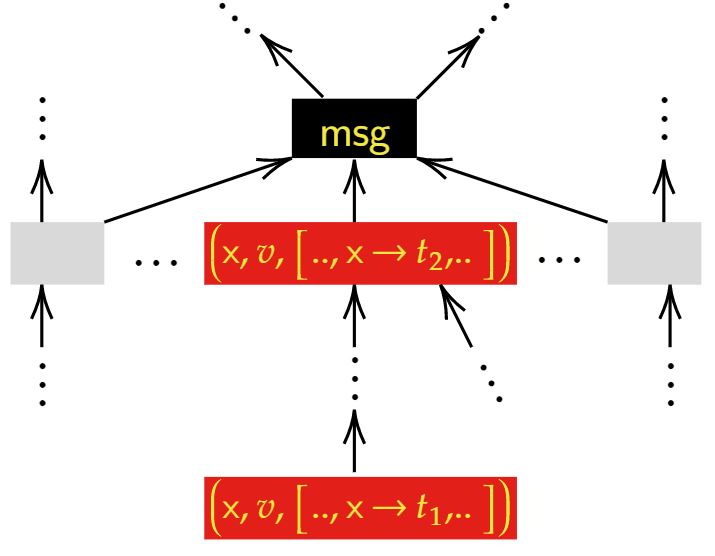}
\end{subfigure}%
\begin{subfigure}{.5\linewidth}
\centering
\includegraphics[scale=0.3]{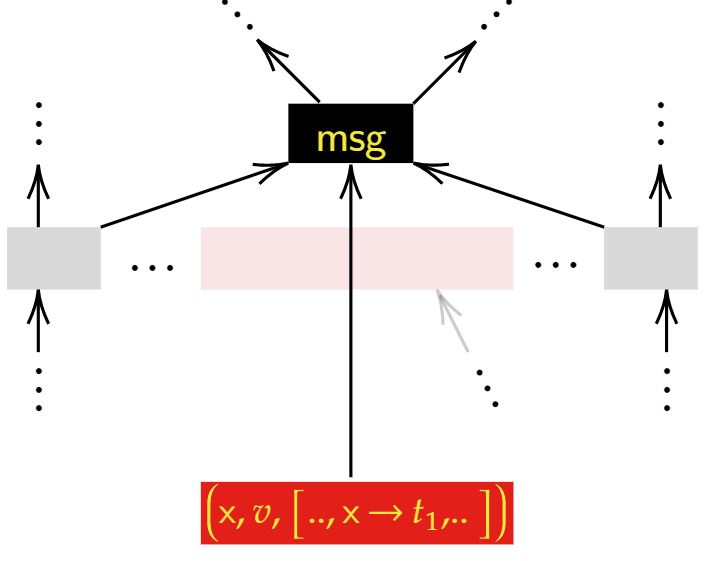}
\end{subfigure}\\[1ex]
\caption{Height compression: the dependency graph on the left can be converted to the one on the right by allowing $\genproc(\anevent)$ to directly read from $(\xvar, \avalue, [\dots, \xvar \rightarrow \atimestamp_1, \dots])$. We note that $\atimestamp_1 \leq \atimestamp_2$ making the new dependency graph, a graph of some valid computation. This prospectively reducing the height of the graph.}
\label{fig:small-height}
\end{figure}

  \item 
\textbf{Height}.  Let there be a dependency sequence  of length greater than $2|\domain||\varset| + |\thisthreadprogsimple|$. There are only $|\domain||\varset|$ (variable, value) pairs, $|\domain||\varset|$ many $\init$ messages and atmost $|\thisthreadprogsimple|$ many $\thisthread$ messages. Hence by a pigeonhole argument, for a dependency sequence longer than $2|\domain||\varset| + |\thisthreadprogsimple|$ there exists a (variable, value) pair $(\anadr, \avalue)$ such that there are two $\thatthread$ messages $m_1 = (\anadr, \avalue, \abstview_1)$ and $n_1 = (\anadr, \avalue, \abstview_2)$ along it. Without loss of generality, let $n_1$ be an ancestor of $m_1$. So, $n_1$ has been read before generating $m_1$. 
Then we must have $\abstview_1 \sqsupseteq \abstview_2$ by the RA Semantics (since the thread generating $m_1$ indirectly accumulates the view of $n_1$). Then the thread reading from (depending-upon) $m_1$ could have directly read from $n_1$ instead (note that since $m_1$ itself depends on $n_1$, by the time $m_1$ has been generated, $n_1$ must have been as well). By reading from $n_1$ its view may only decrease or remain the same thus not affecting the run (as justified above). Thus we can eventually reduce the dependency sequences so that all have length at most $2|\domain||\varset| + |\thisthreadprogsimple|$.
\end{enumerate}
This gives us the result.
	
\end{proof}

In $\dlogmem$ Datalog, the inference of an atom $\agatom$ from the program $\gdatalogprog$ involves a sequence of applications of the Add (to $\dlogmem$) and Drop (from $\dlogmem$) rules that ends with $\agatom$ being inferred. Such a sequence for $\gdatalogprog \vdash \agatom$ corresponds to a run $\abstrun$ under the simplified RA semantics. We show that this follows by the structure of the query instance $(\gdatalogprog, \agatom)$. The run $\abstrun$ can be compacted to ${\abstrun}'$ by Lemma \ref{lem:polyb}. From the dependency graph of ${\abstrun}'$ we can read off an inference strategy that keeps the 
 $\dlogmem$ size polynomial in $|\varset|,|\domain|$ and $|\com_\thisthread|$.
The following lemma formalizes this argument and so proves Proposition~\ref{Proposition:Bound}. We now proceed by showing Lemma \ref{lem:datalog-inference}. This lemma along with Lemma \ref{lem:polyb}  together gives  Proposition \ref{Proposition:Bound} and will lead to the coveted \pspace-bound. Since the term $2|\domain||\varset| + |\thisthreadprogsimple|$ will occur repeatedly, we denote it by the the quantity $\quantity_0$. From here on, $\quantity_0=2|\domain||\varset| + |\thisthreadprogsimple|$. 
\begin{lemma}[Datalog Inference Strategy]\label{lem:datalog-inference}
Let $\nondetalgo$ generate the query instance $(\gdatalogprog, \agatom)$. The inference for $\gdatalogprog \vdash \agatom$ implies the existence of an execution $\abstrun$ under the simplified semantics, which can be compacted to ${\abstrun}'$. The computation ${\abstrun}'$ can be mapped back to a new inference sequence such that $\gdatalogprog \vdash_k \agatom$ for $k \in \bigO{\quantity_0^2}$. 
\end{lemma}
\begin{proof}
This lemma has two parts: (1) it states that computations in the simplified semantics and inference sequences in the $\dlogmem$ Datalog program are related and (2) it says that compact computations can be mapped to an inference sequence with a small $\env$ size.

Let $(\gdatalogprog, \agatom)$ be generated by the procedure $\nondetalgo$ with $\gdatalogprog \vdash \agatom$. We need to show that $\agatom$ can also be inferred from $\gdatalogprog$ with a small $\env$. Recall that when generating the Datalog program $\gdatalogprog$, the procedure $\nondetalgo$  
guesses the computations of the $\thisthread$ processes. Consider some inference sequence for $\gdatalogprog \vdash \agatom$. For each application of an inference rule in the sequence, we can find a corresponding transition of a thread in the simplified semantics. This follows from the invariants in section \ref{app:datalog-invariants}. Hence we can convert the sequence of inferences to a run $\abstrun$. This run in turn can be compacted by the arguments in Lemma \ref{lem:polyb}, to get a smaller run ${\abstrun}'$. Now we need to see how this compact run implies the existence of an inference sequence with smaller sized $\env$. To do this we consider the dependency graph of ${\abstrun}'$.

We proceed by induction on the height of messages in the dependency graph. 
We strengthen the statement and show that for every message $\anabstevent$ at a height given by $\heightv(\anabstevent)=h$, we have 
$\gdatalogprog \vdash_k \messPred(\anabstevent)$ ($\gdatalogprog \vdash_k \dmessPred(\anabstevent)$)
for $k = h \times \quantity_0$.
The lemma follows by the definition of compactness, which guarantees $h\leq \heightv(\depgraph_{\abstrun}) \leq \quantity_0$.

The base case is trivial, since all messages in $\sinkv(\depgraph_{\abstrun})$ are facts in the Datalog program $\gdatalogprog$. 
We now show the inductive case for a message $v \in \depgraph_{\abstrun}$ at height $h+1$. 
The messages $v'$ in $\dep(v)$ have height at most $h$. 
The inductive hypothesis thus yields $\gdatalogprog \vdash_{h\quantity_0} v'$. 
We infer these messages one at a time, store them in the $\env$, and discard all atoms in the $\env$ used for the inference of the $v'$. Hence at each step in the inference sequence, the $\env$ contains a subset of $\dep(v)$ which has already been inferred, and, additionally some atoms which are currently being used for the inference of the next member of $\dep(v)$. The former is bounded by $\quantity_0$ by the compactness (Lemma \ref{lem:polyb}, $|\dep(v)| < \quantity_0$)  while the latter is bounded by $h \quantity_0$ by the induction hypothesis. Together the size of the $\env$ never exceeds $(h+1)\quantity_0$. Thus by reusing the space in the $\env$ to infer members of $\dep(v)$, we only require an additional space of $\quantity_0$. At the end of this process,  the size of the  $\env$ equals $|\dep(v)|$ and the space consumption of the dependencies is at most 

\begin{equation*}\
\underbrace{(\quantity_0 - 1)}_{\text{bound on }|\dep(v)|} + \underbrace{h\quantity_0}_{\text{inductive hypothesis for next atom at height }h} = (h+1)(\quantity_0) - 1
\end{equation*}	

Now we want to infer the message corresponding to $v$, having inferred and inserted into $\env$ atoms corresponding to messages from $\dep(v)$. This inference of $v$ from the messages in $\dep(v)$ requires us to simulate the run of $\genproc(v)$ using the rules of the Datalog program (by mapping each transition executed by $\genproc(v)$ to its corresponding rule from the Datalog program). We note that at all points in the simulation it suffices to store exactly one extra atom either of $\statePred$ or of $\dstatePred$ (depending upon the type of $\genproc(v)$) corresponding to the local state of $\genproc(v)$. The additional atom can be accommodated along with $\dep(v)$ since $|\dep(v)| + 1 < (h+1)\quantity_0$ (since $|\dep(v)| < \quantity_0$). 

Hence a $\dlogmem$ of size at most $2(h+1)(|\domain||\varset|+|\thisthreadprogsimple|)$ is sufficient, and by induction the lemma follows.
\end{proof}

Lemma \ref{lem:polyb} along with the compact inference sequences, Lemma \ref{lem:datalog-inference}, together show that for all the query instances generated by $\nondetalgo$, inference is possible if it is possible with a small $\env$. This shows Lemma \ref{Proposition:Bound} giving us \pspace-membership.

\section{Safety Verification with Leader}\label{sec:nexp-c}

In this section our goal is to support compositional verification methods prominent 
in program logics and thread-modular reasoning style algorithmic verification. Such approaches focus on a single thread and study its interaction with others.

We extend the system from section \ref{sec:consec-lfree} by adding a single distinguished `ego' thread, which we refer to as the \textit{leader}, denoted by the symbol $\leader$. 
Amongst the $n$ $\thisthread$ threads only the $\leader$ can execute loops, while the others, like section \ref{sec:consec-lfree} are required to be loop-free.
 
The environment once again consists of arbitrarily many identical $\thatthread$ threads that are required to be cas-free. We can represent this as
$\thatthread(\uncassy)\parallel \thisthread_1 \parallel  \thisthread_2(\unloopy) \parallel  \dots \parallel  \thisthread_{n}(\unloopy)$
which we refer to as the leader setting.

Note that the simplified semantics presented in Section \ref{Section:Simplification} applies here. 
This allows us to leverage Theorem \ref{Theorem:Semantics} by which we can operate on the simplified semantics instead. The main challenge of this section then is to go from the simplified semantics in the presence of a leader to an \nexp~verification technique, by means of a small model argument.

\subsection{Dependency Analysis}
As discussed before, the safety verification problem amounts to solving the message generation problem (MG) (section \ref{sec:dataenc}). Let the goal message be denoted $\anevent^\#$. 

We demonstrate that the simplified semantics helps solving the problem. 

Our main finding is that message generation has short witness computations (assuming the domain is finite).  The proof of Theorem \ref{Theorem:ShortWitness} is in Section \ref{app:nexp}.

\begin{theorem}\label{Theorem:ShortWitness}
In the leader setting, a message can be generated in the simplified semantics if and only if it can be generated by a computation of length at most exponential in the input specification, $|\com_\thisthread|\cdot |\com_\thatthread| \cdot |\setreg| \cdot |\domain| \cdot |\varset|$.
\end{theorem}

\begin{corollary}\label{Corollary:NEXP}
In the leader setting, the message generation problem for RA is in \nexp.
\end{corollary}
We establish the result in two steps. 
First we show that every computation in the simplified semantics has a ``backbone'', which is made up solely by some threads called \emph{essential threads} (Lemma \ref{lem:essproc}).  Then we show how to 
truncate this backbone to obtain a short computation (Section \ref{sec:shortwitness}).

\smallskip

\noindent {\textbf{Analyzing Dependencies in the Dependency Graph}}. 
The following study of dependencies generalizes the one in Section~\ref{Section:Environment}. 
In a computation of the simplified semantics, messages from the $\thisthread$ threads have unique timestamps whereas messages from $\thatthread$ threads may have identical timestamps.
We recall $\genproc(\anevent)$, the \textit{thread} which first generated  message $\anevent$,  and the \emph{dependency set} of a message~$\anevent$, denoted by $\dep(\anevent)$ as defined earlier in Section~\ref{Section:Environment}.

We define $\dep(\anevent) = \emptyset$ for initial messages.
We write $\dep^*(\anevent)$ for the reflexive and transitive closure of $\dep$, the smallest set containing $\anevent$ and such that for all $\anevent' \in \dep^*(\anevent)$ we have $\dep(\anevent')\subseteq\dep^*(\anevent)$. 

Similar to Lemma~\ref{lem:polyb}, we now show that we can focus on computations where any write event directly depends on a small number of other events,  and where dependency sequences  are short. The main difference with Section \ref{Section:Environment} is that since the leader has loops, we cannot apriori bound executions w.r.t. $|\com_\thisthread|$. Keeping this in mind, we provide an alternative notion for compact computations. 

\noindent{\textbf{Compact Computations}}. We call a computation $\arun$ \emph{compact} if for every $\thatthread$ message $\anevent \in \dep^*(\anevent^\#)$ in the computation  
(1) $\sizeOf{\dep(\anevent) \cap \eventsOf{\projectto{\arun}{\thatthread}} } \leq \sizeOf{\domain}|\varset|$ and 
(2) for every $\anevent'\neq \anevent$ from $\dep^*(\anevent) \cap \eventsOf{\projectto{\arun}{\thatthread}}$ either the variable or the value is different from $\anevent$.
The first point addresses the situation where an $\thatthread$ thread reads two messages with the same variable and value but different views:
it says that the thread could have chosen to read one of the messages twice. 
The second point says there is no need to generate two $\thatthread$ messages with the same variable and value along a dependency sequence. 
A thread reading the second message could equally well read the first message, the $\tplus{\atimestamp}$ timestamp for $\thatthread$ messages would make it available forever. 

\begin{lemma}
In the leader setting, if the message $\anevent^\#$ can be generated in the simplified semantics, then it can be generated by a compact computation.
\label{lem:polydep-leader}
\end{lemma}

In a compact computation, both fan-in (size of $\dep$ set) and depth (along a dependency sequence) of $\thatthread$ messages is $\mathcal{O}(|\domain||\varset|)$ since there are only as many distinct (variable, value) pairs. Hence $\mathcal{O}((|\domain||\varset|)^{|\domain||\varset|})$ many $\thatthread$ messages are sufficient to generate $\anevent^\#$.
Our goal is to derive a similar bound on $\thisthread$ messages. 
First, we consider the $\thisthread$ messages read by $\thatthread$ threads, i.e. the $\thisthread$-$\thatthread$ reads-from dependencies.
The $\thisthread$-$\thisthread$ dependencies will be handled later. 

\smallskip
\noindent\textbf{Essential Messages and Threads}. 
Given a computation $\arun$ in the simplified semantics, the \emph{essential messages} for generating message $\anevent$, 
denoted by $\essmessOf{\anevent}$, is the smallest set that includes $\anevent$ and is closed as follows.
\begin{enumerate}
	\item (1) $\forall$ messages $\anevent'\in \essmessOf{\anevent}\cap \eventsOf{\projectto{\arun}{\thatthread}}$ we have $\dep(\anevent') \subseteq \essmessOf{\anevent}$. 
	\item  $\forall$  $\anevent'\in\essmessOf{\anevent}\cap \eventsOf{\projectto{\arun}{\thisthread}}$ we have $\dep(\anevent')\cap  \eventsOf{\projectto{\arun}{\thatthread}}\subseteq \essmessOf{\anevent}$. 
\end{enumerate}
 Note the asymmetry, for the $\thatthread$ threads we track all dependencies, for the $\thisthread$ threads we only track the dependencies from $\thatthread$.

For a computation $\rho$, the threads generating essential messages of $\anevent^{\#}$ for the first time and the set of $\thisthread$ threads are \emph{essential threads}; $\ethreadOf{\arun}\!=\{\genproc(m) {\mid} m {\in} {\essmessOf{\anevent^\#}
\}} \cup \thisthread$.

We claim that projecting $\arun$ to \emph{essential threads} yields a valid computation in the simplified semantics.
Essential messages thus form the backbone of the computation mentioned above. We now give the proof of Lemma \ref{lem:essproc} and Corollary \ref{cor:expsuff}.  

\begin{lemma}
\label{lem:essproc}
If $\arun$ is a computation in the simplified semantics, so is $\projectto{\arun}{\ethreadOf{\arun}}$.
\end{lemma}
\begin{proof}
To prove this theorem it suffices to show that there is no thread in $\ethreadOf{\arun}$ that reads from some thread $\athread'\not\in\ethreadOf{\arun}$. Then we simply can project away the threads not in $\ethreadOf{\arun}$ and all the reads-from dependencies will still be respected. 

	This follows trivially from the definition of $\essmessOf{~}$. Indeed we have that for an essential $\thatthread$ thread $\athread$ the messages (and hence threads) that $\athread$ reads from are also essential. All $\thisthread$ threads are essential by definition. Additionally, for any $\thisthread$ thread, we add all its $\thatthread$ dependencies to the essential set. The set $\ethreadOf{\arun}$ is then closed under reads-from dependencies and hence the computation $\projectto{\arun}{\essmessOf\arun}$ is valid under RA.
\end{proof}

Now we discuss bounding of essential messages.
Essential $\thatthread$ messages and (and essential $\thatthread$ threads) are atmost exponential, bounded by $\quantity_1 =(|\domain||\varset|)^{|\domain||\varset|}$ using the earlier compactness argument.
We show that the number of essential $\thisthread$ messages is bounded as well.
Firstly, each $\thatthread$ thread has a state space (control-state, registers) bounded by $\quantity_2 =|\com_\thatthread||\setreg|^{|\domain|}$. Given the earlier bound on total number of essential $\thatthread$ messages (and hence those by a single thread), an $\thatthread$ thread run of length greater than $\mathcal{O}(\quantity_1\quantity_2)$ implies that there will exist a sub-run in which (1) no essential message was generated and (2) the thread revisited the same local state twice. We can truncate this sub-sequence since the absence of essential messages implies that external reads-from dependencies are not affected. Hence the computation for a single $\thatthread$ is $\quantity_1\quantity_2$-bounded. Given the $\quantity_1$-bound on $\thatthread$ threads, the total number of $\thisthread$ messages consumed by the $\thatthread$ threads can be atmost $\quantity_1^2\quantity_2$. This implies sufficiency with exponentially many essential $\thisthread$ messages.

\begin{corollary}\label{cor:expsuff}
Let the goal message $\anevent^\#$ be generated in a computation of system $\com$.
Then for some compact computation, $|\essmessOf{\anevent^\#}|$ is at most exponential in $|\com|$.  
\end{corollary}
\begin{proof}
 Recall the notation $\genproc(\anevent)$ which refers to a thread which generated 
 the message $\anevent$ for the first time. In the following, if $t=\genproc(\anevent)$, we also refer to $t$ as the ``first writer'' 
 of $\anevent$. 

	First we observe that $\essmessOf{\anevent} \subseteq \dep(\anevent)$. Hence in particular, we have $\essmessOf{\anevent^\#} \cap \mathsf{Msgs}(\projectto{\arun}{\thatthread}) \subseteq \dep(\anevent^\#) \cap \mathsf{Msgs}(\projectto{\arun}{\thatthread})$. 
	This is shown to be at most exponential ($\mathcal{O}(\quantity_1)$) by Lemma \ref{lem:polydep-leader}, since both the height and the $\thatthread$ fan-in of the dependency graph restricted to $\thatthread$ is polynomial. Given that each essential message is generated for the first time by a unique essential thread, the number of essential $\thatthread$ threads is also bounded by $\mathcal{O}(\quantity_1)$.

	Now, consider the fragment $\arun'$ of the computation between two consecutive first-writes (first points of generation) of two essential $\thatthread$ messages. Now if any $\thatthread$ thread performs more than $\mathcal{O}(\quantity_2) = \mathcal{O}(|\com_\thatthread||\setreg|^{|\domain|})$ many transitions within $\arun'$ it would imply that there are two configurations $\anlcf_1, \anlcf_2$ within $\arun'$ at which the local-states of the thread (modulo view) are identical - this follows since $|\com_\thatthread|$ is the program size and $|\setreg|^{|\domain|}$ is the number of distinct register valuations. Additionally note that the view at $\anlcf_1$ cannot be greater than that at $\anlcf_2$ (monotonicity of views in RA). Hence we can simply truncate the sub-computation between $\anlcf_1$ and $\anlcf_2$ while keeping the computation still valid under RA (the thread with lower view can still perform all its remaining transitions). In this truncation no essential messages will be lost and hence the reads-from dependencies will be respected.

	To explain further, suppose to the contrary that some thread $\athread$  which is the first writer of  an essential message  executed  more than $\mathcal{O}(\quantity_1\quantity_2)$ number of transitions $\anlcf_0 \anlcf_{1} \cdots \anlcf_{l}$. 
	Since the total number of essential messages is only $\mathcal{O}(\quantity_1)$, there must exist a subsequence $\sigma$ such that no essential $\thatthread$ messages 
	were generated (for the first time) 	in $\sigma$. Additionally, since the state-space of each thread is $\mathcal{O}(\quantity_2)$, by a pigeon-hole argument, it follows that two local configurations $\anlcf_i,\anlcf_j$ of $\athread$ in $\sigma$ are equal. We can simply truncate the fragment of the run between these configurations since no essential messages have been generated for the first time.

	Then it suffices for each first writer $\thatthread$ thread to take at most $\mathcal{O}(\quantity_1\quantity_2)$ many transitions and consequently read at most exponentially many $\thisthread$ messages. Recall that the $\thisthread$ messages that are read by first writers of essential $\thatthread$ messages are essential themselves. Since the number of essential $\thatthread$ threads which are first writers itself is bounded by $\mathcal{O}(\quantity_1)$, the number of essential $\thisthread$ messages is bounded by $\mathcal{O}(\quantity_1^2\quantity_2)$, which is exponential in the input. Since $\essmessOf{\anevent^\#}$ is a union of essential $\thisthread$ and $\thatthread$ messages we get the exponential bound on essential messages.
\end{proof}

Combined with Lemma \ref{lem:essproc}, the corollary says it is sufficent to focus on computations with atmost exponentially many essential threads and essential messages.
We now want to bound the computation of the $\thisthread$ threads.

\subsection{Short Witnesses}
\label{sec:shortwitness}
The computation truncation idea as applied to $\thatthread$ threads earlier does not apply
to the leader. 
Recall the asymmetry in the definition of essential dependencies;
we did not include the $\thisthread$-$\thisthread$ load dependencies.
The dependencies come in two forms: (1) those involving (either as message writer or as reader) some non-leader $\thisthread$ thread and (2) $\leader$-$\leader$ dependencies. The former are poly-sized owing to the loop-free nature of the non-leader $\thisthread$ threads. Hence, we focus on $\leader$-$\leader$ dependencies. For a memory $\amem^{\abstscriptpshort}$, let 
$\projectto{\amem^{\abstscriptpshort}}{\leader}$ be the set 
of $\leader$ messages in it. Assuming $\aview^{\abstscriptpshort}$ is the view of the $\leader$, let $\canread(\aview^{\abstscriptpshort}, \amem^{\abstscriptpshort})$ denote the $(\anadr, \avalue)$ pairs
in messages of $\projectto{\amem^{\abstscriptpshort}}{\leader}$ which can be read by $\leader$.
\begin{definition}
$\canread(\aview^{\abstscriptpshort}, \amem^{\abstscriptpshort})=\set{(\anadr, \avalue)\bnf (\anadr, \avalue, \aview_1^{\abstscriptpshort})\in \projectto{\amem^{\abstscriptpshort}}{\leader},~ \aview^{\abstscriptpshort}(\anadr) = \aview_1^{\abstscriptpshort}(\anadr)}$.
 \end{definition}

We note that a pair $(\anadr, \avalue)$ is in $\canread$ when this pair is the last store by the $\leader$ on $\anadr$ following which $\abstview(\anadr)$ has not changed. Observe that there can be at most $|\varset|^{|\domain|}$ many distinct $\canread$ functions. 
Consider a sub-computation of the leader between two generations of essential messages.
We call configurations $\acf^{\abstscriptpshort}_1$ and $\acf^{\abstscriptpshort}_2$ \emph{$\leader$-equivalent} if 
(1) the local configurations of the leader coincide except for the views $\aview_1^{\abstscriptpshort}$ resp. $\aview_2^{\abstscriptpshort}$ and
(2) the memories $\amem^{\abstscriptpshort}_1$ and $\amem^{\abstscriptpshort}_2$ satisfy\\
\begin{align*}
	\canread(\aview^{\abstscriptpshort}_1, \amem^{\abstscriptpshort}_1) = \canread(\aview^{\abstscriptpshort}_2, \amem^{\abstscriptpshort}_2) \\ 
\end{align*}

Then the computation of the leader between $\acf^{\abstscriptpshort}_1$ and $\acf^{\abstscriptpshort}_2$ can be projected away while retaining a computation in the simplified semantics.
Since there are only $\mathcal{O}(|\com_\leader|(|\setreg||\varset|)^{|\domain|})$ many distinct configurations that are not $\leader$-equivalent, after projecting away the redundant part, the leader will have an at most exponentially long computation between generation of two consecutive essential messages.
Given the exponential bound on all essential messages, we see that post projection, the leader computation is reduced to exponential size.
Combined with the argument for the $\thatthread$ and non-leader $\thisthread$ threads, gives Theorem~\ref{Theorem:ShortWitness}. 
Note that the resulting non-deterministic algorithm does not run in polynomial space as there may be exponentially many essential $\leader$ messages which need to be generated concurrently with the $\thatthread$ threads.

\subsection{Theorem \ref{Theorem:ShortWitness} : \nexp-membership of safety verification in the leader case}
\label{app:nexp}
We now move on to Theorem \ref{Theorem:ShortWitness}. It suffices to show that we only need to consider computations of exponential length in order to verify safety properties of a parameterized system under the simplified semantics in the leader case. For this, we show exponential bounds on the $\thatthread$ and $\thisthread$  components of the computation.

We have already seen that for the essential $\thatthread$ threads, $\mathcal{O}(\quantity_1^2\quantity_2)$ is an upper bound on the number of transitions they need to make. Additionally this bound also applies to the number of essential $\thisthread$ messages. Note that the non-leader $\thisthread$ threads are loop-free and hence their number of transitions is polynomial in $|\com_\thisthread|$. Hence we now focus on computations of the leader. We denote $\quantity_3 = |\com_\leader|(|\setreg||\varset|)^{|\domain|}$ which is a bound on the number of distinct (non equivalent) leader configurations  and use it below in the proof.

For the $\leader$, we need to maintain more states (as compared to the $\thatthread$ threads) to ensure that the truncated run is valid. This is so as we also want to capture $\leader$-$\leader$ dependencies as well. The $\canread$ function does precisely this - at each point in the run it tracks the set of $\leader$ messages that can be read by the $\leader$ itself. 

Assume once again that there is a (super-exponential) leader computation with length greater than $\mathcal{O}(\quantity_1^2\quantity_2\quantity_3)$. Then since $\mathcal{O}(\quantity_1^2\quantity_2)$ is a bound on the number of total $\thisthread$ essential messages (and in particular essential $\leader$ messages), there must exist a sub-computation of the $\leader$ of length greater than $\mathcal{O}(\quantity_3)$ that is free of essential message generation. Let this sub-computation be $\anlcf_1 \anlcf_2 \cdots \anlcf_l$. Assume the memory states along this sub-computation to be $\amem_1^{\abstscriptpshort} 
\amem_2^{\abstscriptpshort} \dots \amem_l^{\abstscriptpshort}$. 

We augment each configuration $\anlcf_i$ with the respective memory state $\amem_i^{\abstscriptpshort}$ obtaining an \emph{augmented configuration} as explained below. 
Consider the configurations obtained by augmenting $\anlcf = (\com, \avaluation, \abstpview)$ to the set $\canread(\abstpview, \amem^{\abstscriptpshort})$.
That is, given $\anlcf_i = (\com_i, \avaluation_i, \abstpview_i)$, on augmentation with $\canread(\abstpview_i, \amem_i^{\abstscriptpshort})$ we obtain 
 the augmented state as $\langle \com_i, \avaluation_i, \canread(\abstpview_i, \mathsf{m}_i) \rangle$. Now, $\canread$ can take atmost $|\varset|^{|\domain|}$ many values, while the leader local-state (modulo view) has only $|\com_\leader||\setreg|^{|\domain|}$ values. This implies, (by a pigeon-hole argument), the existence of a pair $i,j$ such that $\langle \anlcf_i, \canread(\abstpview_i, \mathsf{m}_i) \rangle$ and $\langle \anlcf_j, \canread(\abstpview_j, \mathsf{m}_j) \rangle$ are equivalent. 

Now, the view of the $\leader$ thread is monotonic. This implies that if for $i\neq j$ we have $\langle \com_i, \avaluation_i, \canread(\abstpview_i, \mathsf{m}_i) \rangle = \langle \com_j, \avaluation_j, \canread(\abstpview_j, \mathsf{m}_j) \rangle$ then the sub-computation between $i$ and $j$ may be truncated. Thus the run $\anlcf_1 \cdots \anlcf_{i} \anlcf_{j+1} \cdots \anlcf_l$ is also a valid run of the thread. Moreover it does not affect other threads since once again no essential messages are lost.

Hence for any super-exponential (order greater than $\mathcal{O}(\quantity_1^2\quantity_2\quantity_3)$) leader computaion, there exists a shorter computation which also preserves reachability. Thus for safety verification it suffices to consider runs of atmost exponential length, immediately giving an \nexp~ upper bound.

\section{Limits of Semantic Simplification I: \pspace-hardness of $\thatthread(\uncassy,\unloopy)$}
\label{sec:pspace-h}

\newcommand{\booleanvars}{\textit{Vars}}

We show that the applications of semantic simplification to the loop-free and leader settings are tight, and further simplification is not possible. 

Having shown that safety verification of $\thatthread(\uncassy) \parallel \thisthread_1(\unloopy) \parallel  \cdots \parallel  \thisthread_n(\unloopy)$ is in \pspace, we give a matching lower bound. For the lower bound, it suffices to consider the variant with no $\thisthread$ threads and loop-free $\thatthread$ threads, $\thatthread(\uncassy,\unloopy)$. 
In fact, this result captures the inherent complexity in Parameterized RA, termed as $\mathsf{PureRA}$, i.e. RA in its simplest form. 
The simplicity of $\mathsf{PureRA}$ comes from 
(1) disallowing registers, and 
(2) stores can only write value 1 and the memory is initialized with 0 values. 
We obtain \pspace-hardness even with this reduced form, which is surprising, given that in its full form it is in \pspace. 
Notice that the \pspace-hardness with registers is trivial, since \pspace~can be encoded in valuations of registers.

\subsection{Pure RA} In this section, we elaborate on the \pspace-hardness of checking safety properties of parameterized systems under RA in the absence of $\thisthread$ threads (and loop-free, cas-free $\thatthread$ threads), which we can denote as $\thatthread(\uncassy, \unloopy)$. 
In fact, we investigate the inherent complexity in RA, by removing all extra frills like registers, as well as arbitrary data domains. So what we have is, 
$\mathsf{Pure}$ RA, which is basically, RA in its simplest form. The simplicity of $\mathsf{Pure}$ RA comes from the fact that 
we do not use registers, and the only writes that are allowed are 
that of writing value 1 to any shared variable, where we assume that the memory was initialized to 0 so that we have a data domain of $\{0,1\}$. The remarkable thing about this result is that we obtain \pspace-hardness, which is surprising, given that in its full form it is in \pspace~by Section \ref{sec:consec-lfree}.
Notice that the \pspace-hardness with registers is trivial, since computations can be encoded in register operations themselves. 
\tikzset{background rectangle/.style={fill=none
}}
\begin{figure*}[h]
\centering
\small
\begin{tikzpicture}[codeblock/.style={line width=0.5pt, inner xsep=0pt, inner ysep=5pt}  , show background rectangle]
\node[codeblock] (init) at (current bounding box.north west) {
$
\def\arraystretch{1.5}
\begin{array}{rl}
\rowcolor{black!15}
  \pspacelibrary &= \funcAG \choice \funcSATC \choice \funcFEC{0} \choice  \cdots \choice \funcFEC{n-1} \choice \funcGC 
\\
\rowcolor{black!5}
&\quad ~ \mathsf{choose}(u) = (\assign{t_{u}}{0}) \choice (\assign{f_{u}}{0})\\
\rowcolor{black!5}
\funcAG &= \mathsf{choose}(u_0); \mathsf{choose}(e_1); \mathsf{choose}(u_1); \cdots; \mathsf{choose}(u_n); (\assign{s}{1})\\
\rowcolor{black!15}
\funcSATC &= \assume{(s = 1)}; \mathsf{check}(\Phi); \\
\rowcolor{black!15}
&~\quad((\assume{(t_{u_n} = 0)}; \assign{a_{n,1}}{1}; ) \choice (\assume{(f_{u_n} = 0)}; \assign{a_{n,0}}{1}))\\
\rowcolor{black!5}
\funcFEC{i} &=~ \assume{(a_{i+1,0} = 1)}; \assume{(a_{i+1,1} = 1)}; (\assume{(f_{e_{i+1}} = 0)} \choice \assume{(t_{e_{i+1}} = 0)}); \\
\rowcolor{black!5}
 &~\quad((\assume{(t_{u_i} = 0)}; \assign{a_{i,1}}{1}) \choice (\assume{(f_{u_i} = 0)}; \assign{a_{i,0}}{1}))\\
\rowcolor{black!15}
\funcGC &= \assume{(a_{0,0} = 1)}; \assume{(a_{0,1} = 1)}; \assert{\texttt{false}}\\
\end{array}$
};
\end{tikzpicture}
\vspace{-0.2cm}
\caption{The parametrized system used in the reduction}
\label{fig:pspace}
\end{figure*}

\subsubsection{A QBF Encoding} To show the \pspace-hardness of checking safety properties of parameterized systems of the class $\thatthread(\uncassy,\unloopy)$, we establish a reduction from the canonical \pspace-complete problem, $\qbfsat$. The $\qbfsat$ problem is described as follows.
 Given a quantified boolean formula 
  $\Psi=\forall u_0 \exists e_1 \forall u_1 \exists e_2 \cdots \exists e_n \forall u_n ~\Phi(u_0, e_1,\cdots u_n)$, over variables $Vars(\Psi)=\{u_0, \dots, u_n, e_1, \dots, e_n\}$, decide if $\Psi$ is  true. 
$\Psi$ has $n+1$  universally quantified variables and $n$ existentially quantified variables. 
To establish the reduction, we construct an instance of the parametrized reachability problem for RA (in fact $\mathsf{Pure}$ RA) consisting of the parametrized system $\com$, such that $\com$ is unsafe if and only if the $\qbfsat$ instance is true. We assume that the $\qbfsat$ instance $\Psi$ is as given above and now detail the construction.

The program $\com$ executed by the $\thatthread$ threads (given in Figure \ref{fig:pspace}) consists of functions (sub-programs), one of which may be executed non-deterministically:
\begin{equation*}
  \pspacelibrary = \funcAG \choice \funcSATC \choice \funcFEC{0} \choice  \cdots \choice \funcFEC{n-1} \choice \funcGC
\end{equation*}
 
\subsection{Infrastructure}
\textbf{Gadgets used.} The task of checking the satisfiability of $\Psi$ is distributed over the $\thatthread$ threads executing these functions. Each function has a particular role, which we term as gadgets and now describe.
\setlist[itemize]{leftmargin=1cm}
\begin{itemize}
  \item[$\funcAG$] The \emph{Assignment Guesser} guesses a possible satisfying assignment for $Vars(\Psi)$.
  \item[$\funcSATC$] The \emph{SATisfiability Checker} checks satisfiability of $\Phi$ w.r.t. an assignment guessed by $\funcAG$. 
  \item[$\funcFEC{i}$] The $\forall \exists$ (\textit{ForallExists}) \emph{Checker} at level $i$, $0 \leq i \leq n-1$ ($\funcFEC{i}$) verifies that the $(i+1)$th quantifier alternation $\forall u_{i} \exists e_{i+1}$ is respected by the guessed assignments. This proceeds in levels, where the check function at level $i+1$, $\funcFEC{i+1}$ `triggers' the check function at level $i$, $\funcFEC{i}$, till we have verified that all assignments satisfying $\Phi$ constitute the truth of $\Psi$.
  \item[$\funcGC$] The \textit{Assertion Checker} reaches the $\assert{\texttt{false}}$ instruction when all the previous functions act as intended, implying that the formula was true.  
\end{itemize}
Due to the parameterization, an arbitrary number of threads may execute the different functions at the same time. However, there is no interference 
between threads, and there is a natural order between the roles: 
$\funcSATC$ requires $\funcAG$ to function as intended, and $\funcFEC{i}$ requires the functions $\funcAG$, $\funcSATC$ and $\funcFEC{j}$, $n-1 \geq j > i$.

\textbf{Shared Variables.}
We use the following set of shared variables in $\com$: For each $x \in Vars(\Psi)$, we have boolean shared variables $t_x$ and $f_x$ in $\com$. These variables represent true and false assignments to $x$ using the respective boolean variables in a way that is explained below. All the shared variables used are boolean, and  the initial value of all variables is 0. 
 We also have a special (boolean) variable $s$.

\textbf{Encoding variable assignments of $\Psi$: the essence of the construction.}
Recall that the messages in the memory are of the form $(\xvar, \avalue, \aview)$ where $\xvar$ is a shared variable, $\avalue \in \{0,1\}$, and $\aview$ is a view. 
To begin, the views of all variables are assigned time stamp 0.   
An assignment to the variables in $\Psi$ can be read off from the $\aview$ of a message $(s, 1, \aview)$ in the memory state. For $v \in Vars(\Psi)$, if $\aview(t_v) = 0$, then $v$ is considered to have been assigned true, while if $\aview(f_v) = 0$, then $v$ is assigned false. Our construction, explained below, ensures that exactly one of the shared variables $t_v, f_v$ will have time stamp 0 in the view of the message $(s, 1, \aview)$. The zero/non-zero timestamps of variables $t_x$ and $f_x$ in the view of $(s, 1, \aview)$ can be used to check satisfiability of $\Phi$ since only a thread with a zero timestamp can read the initial message on the corresponding variable.

\textbf{Checking a single clause.} As an example, consider the $i^{th}$ clause $e_1 \lor \lnot u_3 \lor u_5$. The satisfiability check is implemented in a code fragment as follows. $\mathsf{check}(i) = (\assume{t_{e_1} = 0}) \choice (\assume{f_{u_3} = 0}) \choice (\assume{t_{u_5} = 0})$ and 
$\mathsf{check}(\Phi) = \mathsf{check}(1); \mathsf{check}(2); \cdots; \mathsf{check}(l)$.
Finally, we have the boolean variables $a_{i, 0}$ and $a_{i, 1}$ for $i \in \{0, \cdots n\}$: these are $2(n+1)$ `universality enforcing' variables that ensure that all possible assignments to the universal variables in $Vars(\Psi)$ have been checked. 

\subsection{The Construction} 
\label{app:constr-pspace}
First we describe the various gadgets. 
\subsubsection{The Gadgets}
We now detail the gadgets (functions) mentioned in Figure \ref{fig:pspace}. 

\textbf{Assignment Guesser}: $\funcAG$: The job of the Assignment Guesser is to guess a possible assignment for the variables. This is done by writing 1 to exactly one of the variables $t_x, f_x$ for all $x \in Vars(\Psi)$. Each such write is required to have a timestamp greater than 0 by the RA semantics, and the view $\aview$ of the writing thread is updated similarly. After making the assignment to all 
variables in $Vars(\Psi)$ as described, the writing thread adds the message $(s,1, \aview)$ to the memory. 

Consequently, the view $\aview$ of the writing thread (and hence the message) satisfies 
\begin{equation*}
    \forall x \in Vars(\Phi), ~~~ \aview(t_x) = 0 \oplus \aview(f_x) = 0
\end{equation*}
We interpret this as: the assignment chosen for $x \in Vars(\Phi)$ is true if $\aview(t_x)=0$ and is false if $\aview(f_x)=0$. The chosen assignment is thus encoded in $\aview$ and hence can be incorporated by threads loading 1 from $s$ using the message $(s,1,\aview)$, (see $\funcSATC$). This follows since load operations of the RA semantics cause the thread-local view to be updated by the view in the message loaded.

\textbf{SAT Checker}: $\funcSATC$: The SAT Checker reads from one of the messages of the form $(s, 1, \aview)$ generated by $\funcAG$. Using the code explained in Figure \ref{fig:pspace}, it must check that the assignment obtained using the $\aview$ satisfies $\Phi$. The crucial observation is that $\assume{(t_x= 0)}$ ($\assume{(f_x=0)}$) being successful is synonymous with the timestamp of $t_x$ ($f_x$) in $\aview$ being 0. This holds since $\assume{(v= 0)}$ requires the ability to read the initial message on $v$ which in turn requires the thread-local view on $v$ to be 0. Timestamp of $t_x$ ($f_x$) in $\aview$ itself being 0 is equivalent to $x$ being assigned the value $\texttt{true}$ ($\texttt{false}$) by $\funcAG$.

Finally it checks that either $t_{u_n}$ or $f_{u_n}$ had timestamp 0 in $\aview$, and writes 1 to $a_{n,1}$ or $a_{n,0}$ correspondingly 
 in Figure \ref{fig:pspace}.  
For insight, we note prematurely that we will enforce both these writes to 
$a_{n,1}$ and $a_{n,0}$ as a way of ensuring the universality for the variable $u_n$.
The main task is to verify the `goodness' of the assignments satisfying $\Phi$. One of the things to verify is that, we have satisfying assignments for both values true/false of the universal variables $u_i$. 

If the $\assume{(t_{u_n}=0)}$ evaluates to true 
in $\funcSATC$  then in the view of the message $(s,1,\aview)$ obtained at the end of $\funcAG$, $\aview(t_{u_n}) = 0$.
 We now need a $\funcAG$ function (executed by some thread) to make an assignment such that in the view of the message $(s,1,\aview)$, we have $\aview(f_{u_n}) = 0$, and the formula $\Phi$ is satisfiable again. The next step is to check if these assignments which differ in $u_n$  are sound with respect to the $\forall u_{n-1} \exists e_{n} $ part
of $\Psi$ : that is, the assignment to $e_n$ is independent to that of $u_n$. 
  This procedure has to be iterated with respect to all of $u_0,u_1, \dots, u_{n-1}$ by (1) first ensuring that $\Phi$ is satisfiable for both assignments to $u_i$, $0 \leq i \leq n-1$ and (2) verifying that such assignments are sound with respect to the quantifier alternation in $\Psi$ for $1 \leq i \leq n-1$, (that is the choice of assignment to $e_i$ is independent of all variables in $\{u_i, e_{i+1}, \cdots, u_n\}$). 

\textbf{ForallExists Checker}: $\funcFEC{\_}$: The $n$ $\forall\exists$ \emph{Checker}s 
$\funcFEC{0}, \dots, \funcFEC{n-1}$ 
take over at this point, consuming the writes made earlier. In general, for each $i\in \{0, \cdots, n-1\}$, we have  $\forall\exists$ \emph{Checker} function  of $n$ kinds, $\funcFEC{0}, \dots, \funcFEC{n-1}$ that operate at levels $0, \dots, n-1$. $\funcFEC{i}$ operates at level $i$ by reading 1 from $a_{i+1,0}, a_{i+1,1}$ variables, and making writes to $a_{i,0}, a_{i,1}$ variables for $0 \leq i \leq n-1$. 

\textit{Universality Check} :  $\funcFEC{i}$ first verifies that all possible valuations to the universally quantified variable $u_{i+1}$ made $\Phi$ satisfiable : the two 
statements $\assume{(a_{i+1,0} = 1)}; \assume{(a_{i+1,1} = 1)}$
verify this by reading 1 from $a_{i+1,0}$ and $a_{i+1,1}$ (note how all higher $\funcFEC{j}$, $j>i$ level functions enforce this by generating a dependency tree such as the one in Figure \ref{fig:dep}).  

\textit{Existentiality Check} : Next, $\funcFEC{i}$ checks that the satisfying assignments of $\Phi$ seen so far agree on the existentially quantified variable $e_{i+1}$ : the
statements $(\assume{(f_{e_{i+1}} = 0)} \choice \assume{(t_{e_{i+1}} = 0)})$  check this. 
Assume that we have satisfying assignments of $\Phi$ which do not agree on  $e_{i+1}$. Then we have messages $(a_{i+1,0}, 1, \aview_1)$ and $(a_{i+1,1}, 1, \aview_2)$ such that $\aview_1(t_{e_{i+1}}) > 0$ ($e_{i+1}$ assigned $\texttt{false}$) but $\aview_2(t_{e_{i+1}}) = 0$ ($e_{i+1}$ assigned $\texttt{true}$). 
Now when $\funcFEC{i}$ reads from these messages, its view $\aview$, will have both, $\aview(t_{e_{i+1}}) > 0$ and $\aview(f_{e_{i+1}}) > 0$. This will disallow $\funcFEC{i}$ from executing $(\assume{(f_{e_{i+1}} = 0)} \choice \assume{(t_{e_{i+1}} = 0)})$  since the messages in the memory where $t_{e_{i+1}}$ and $f_{e_{i+1}}$ have value 0 (and time stamp 0) cannot be read. This enforces that the choice of the existentially quantified variable $e_{i+1}$ is independent of the choice of the assignments made to the variables in $\{u_{i+1}, e_{i+2}, \cdots, u_n\}$, and hence the proper semantics of quantifier alternation is maintained. 
  
\textit{Propagation} Finally, the $\funcFEC{i}$ function `propagates' assignments to the next level, that is, to $\funcFEC{i-1}$ after a last verification. Let $A_{i+1,j}$ contain all assignments satisfying $\Phi$ which agree on $e_{i+1}$, and 
where $u_i$ is assigned value $j \in \{0,1\}$. Such assignments are propagated to the next level by a $\funcFEC{i}$ function which writes 1 to $a_{i,j}$. $\funcFEC{i-1}$ is accessible only when $A_{i+1,0}$ and $A_{i+1,1}$ are both propagated.

 \begin{figure*}[h]
 \centering
\tikzset{
  treenode/.style = {align=right, inner sep=1.5pt, text centered,
    font=\sffamily},
   arn_r/.style = {treenode, rectangle, draw=none,  
    },
  }
   \scalebox{.7}{ 
\begin{tikzpicture}[-,>=stealth',level 1/.style={sibling distance = 9cm}, level 2/.style={sibling distance = 4.5cm}, level 3/.style={sibling distance = 2cm},
  level distance = 1.5cm] 
\node [arn_r,fill=none] {\LARGE{\color{red}${\mathsf{assert}}$}}
    child{node [arn_r,fill=blue!20] {\Large{${\mathsf{FE}[0]}$}}
                    child{ node [arn_r, fill=green!20] {\Large{${\mathsf{FE}[1]}$}} 
              child{ node [arn_r] {\Large{
              $\begin{array}{ll} {\mathsf{SATC}}  \end{array} $} } edge from parent node[pos=0.5, left]
                         {\Large{$a_{2,1}{=}1$}}} 
        child{ node [arn_r] {\Large{
              $\begin{array}{ll} {\mathsf{SATC}}  \end{array}$}} edge from parent node[pos=0.5,right]      {\Large{$a_{2,0}{=}1$}}}
                edge from parent node[pos=0.5, left]
                         {\Large{$a_{1,1}{=}1$}}  
            }
                 child{ node [arn_r,fill=green!20] {\Large{${\mathsf{FE}[1]}$}}
              child{ node [arn_r] {\Large{ $\begin{array}{ll} {\mathsf{SATC}} \end{array}$}} edge from parent node[pos=0.5, left] {\Large{$a_{2,1}{=}1$}} }
              child{ node [arn_r] {\Large{$\begin{array}{ll} {\mathsf{SATC}} \end{array}$}} edge from parent node[pos=0.5, right] {\Large{$a_{2,0}{=}1$}}}
  edge from parent node[pos=0.5,right]
                         {\Large{$a_{1,0}{=}1$}}}
            edge from parent node[pos=0.5, sloped, above]
                         {\Large{$a_{0,0}{=}1$}} 
                             }
                                 child{ node [arn_r,fill=blue!20] {\Large{${\mathsf{FE}[0]}$}}
            child{ node [arn_r,fill=orange!20] {\Large{${\mathsf{FE}[1]}$}} 
              child{ node [arn_r] {\Large{$\begin{array}{ll} {\mathsf{SATC}} \end{array}$}} edge from parent node[pos=0.5, left] {\Large{$a_{2,1}{=}1$}} }
              child{ node [arn_r] {\Large{$\begin{array}{ll}{\mathsf{SATC}} \end{array}$}} edge from parent node[pos=0.5,right]{\Large{$a_{2,0}{=}1$}}
              }
              edge from parent node[pos=0.5,left]
                         {\Large{$a_{1,1}{=}1$}}
            } 
             child{ node [arn_r, fill=orange!20] {\Large{${\mathsf{FE}[1]}$}}
              child{ node [arn_r] {\Large{$\begin{array}{ll} {\mathsf{SATC}}\end{array}$}} edge from parent node[pos=0.5, left] {\Large{$a_{2,1}{=}1$}}
              }
              child{ node [arn_r] {\Large{$\begin{array}{ll}{\mathsf{SATC}}\end{array}$}}
              edge from parent node[pos=0.5,right]{\Large{$a_{2,0}{=}1$}}
              }
              edge from parent node[pos=0.5,right]
                         {\Large{$a_{1,0}{=}1$}}
            } edge from parent node[pos=0.5, sloped, above]
                         {\Large{$a_{0,1}{=}1$}}
    }
; 
\end{tikzpicture}}
\caption{The dependency tree for the case of 
$\forall u_0 \exists e_1 \forall u_1 \exists e_2 \forall u_2 \Phi$. 
The same color of sibling nodes $\funcFEC{i}$ represents that 
the value of $e_{i+1}$ is same at both of these. 
} 
\label{fig:dep}
 \end{figure*}
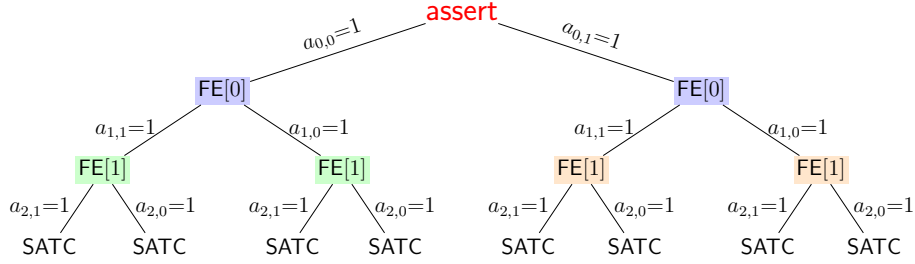

\textbf{Assert Checker}: $\funcGC$: After the $n$ $\forall\exists$ Checkers finish, the Assertion Checker reads 1 from the variables $a_{0,0}$ and $a_{0,1}$ and reaches the 
assertion $\assert{\texttt{false}}$. This is possible only if all the earlier functions act as intended, which in turn is only possible if the QBF evaluates to true. 

\subsubsection{Roles played by the threads}
The non-deterministic branching between the choices of the gadgets above means that each $\thatthread$ thread executes exactly one of the gadgets. However together they check $\Psi$ in a distributed fashion as one thread passing on a part of its state to the next one by the load-stores for the $a_{\_, 0/1}$ variables as mentioned above. Hence a computation that reaches the assertion requires each thread to play a part in this tableau. We now describe this.

First a set of $2^n$ threads run the $\funcAG$ gadgets and they guess one assignment each such that all possible assignments for the universally quantified variables are covered and such that the existentially quantified variables are chosen such that the semantics of quantifier alternation is respected. Essentially this means the the $2^n$ assignments guessed would be a sufficient witness to the truth of $\Psi$.

Now, $2^n$ threads execute $\funcSATC$ and check that each of the assignments guessed (one thread checks one assignment) satisfies $\Phi$. They produce a `proof' that this check is complete by writing to variables $a_{n, 0/1}$. This also checks the innermost universality is respected. 
At level $n-1$, $2^{n-1}$ threads execute $\funcFEC{n-1}$. Each $\funcFEC{n-1}$
reads 1 from both $a_{n,0}$ and $a_{n,1}$ and reads 0 from exactly one of $t_{e_n}$ or $f_{e_n}$. Depending on the view read 
from the level below, they either write 1 to $a_{n-1,0}$ or to $a_{n-1,1}$. (Prematurely this corresponds to 
the assignments $A_{n,0}$ and $A_{n,1}$ in the proof below.) In essence these threads check that the last quantifier alternation ($\forall u_{n-1} \exists e_{n}$) is respected.
$2^{n-2}$ threads then execute $\funcFEC{n-2}$ at level $n-2$, reading 1 from both $a_{n-1,1}$ and $a_{n-1,0}$, and reading
0 from exactly one of  $t_{e_{n-1}}$ or $f_{e_{n-1}}$. These threads then write 1 
to either $a_{n-2,0}$ or to $a_{n-2,1}$, (representing assignments $A_{n-1,0}$ and $A_{n-1,1}$ in the proof below). These threads check that the second last quantifier alternation $\forall u_{n-1} \exists e_{n-1}$ is respected. 
This continues till two threads execute $\funcFEC{0}$, and  
writes 1 to $a_{0,1}$ or $a_{0,0}$. 
These two writes are read by a thread executing $\funcGC$. 
The views of these threads are all stitched together by the stores and loads they perform on the variables $s$ (for guessing assignments) and $a_{\_,0/1}$ for checking proper alternation.
Figure \ref{fig:dep} illustrates 
how the view (in which the assignments are embedded as described earlier) propagate through these threads for the case of the QBF $\forall u_0 \exists e_1 \forall u_1 \exists e_2 \forall u_2 \Phi$. The nodes represent individual threads executing the corresponding gadget and the edges represent the variable which a child writes to pass on its view to its parent.

\begin{lemma}
$\Psi$ is true iff the $\assert{\texttt{false}}$ statement 
is reachable in $\pspacelibrary$.
\label{lem:pspace-hard}
\end{lemma}
 This gives us the main theorem 
\begin{theorem}
The verification of safety properties for parametrized systems of the class $\thatthread(\uncassy, \unloopy)$ under RA 
is \pspace-hard.   
\label{thm:pspace-c}
\end{theorem}

\subsection{Correctness of the Construction: Proof of Lemma \ref{lem:pspace-hard}}
\label{app:pspace-hard}
We prove that reaching $\assert{\texttt{false}}$  is possible in the parameterized system $\pspacelibrary$ iff the QBF 
$\Psi$ is satisfiable. First we fix some notations. 
Given the QBF $\Psi=\forall u_0 \exists e_1 \dots \exists e_{n}\forall u_n \Phi(u_0,e_1, \dots, u_n)$, we define for $0 \leq i \leq n$, the 
level $i$ QBF corresponding to $\Psi$ as follows.
\begin{enumerate}
	\item For $0 \leq i \leq n-1$, the level $i$ QBF, denoted 
	$\Psi_i$ is defined as 
	$$\Psi_i \equiv \forall u_i \exists e_{i+1} \forall u_{i+1} \exists e_{i+2} \dots \forall u_n \textcolor{blue}{\exists e_1 \exists e_2 \dots \exists e_i \exists u_0 \exists u_1 \dots \exists u_{i-1}} \Phi(u_0, e_1, \dots, u_n)$$ 
	\item For $i=n$, the level $n$ QBF, denoted $\Psi_n$ is defined as 
	$$\Psi_n \equiv \forall u_n \textcolor{blue}{\exists e_1 \dots \exists e_n \exists u_0 \exists u_1 \dots \exists u_{n-1}}
	\Phi(u_0, e_1, \dots, u_n)$$	
	\end{enumerate}
Note that $\Psi_0$ is the same as $\Psi$. 
To prove Lemma \ref{lem:pspace-hard}, we prove the following helper lemmas. For ease of arguments, we add 
some labels in our gadgets, and reproduce them below. 

\begin{figure}[h]
\centering
\setlength{\belowcaptionskip}{-7pt}
\setlength{\abovecaptionskip}{-2pt}
\begin{empheq}[box={\mycode[rounded corners]}]{align*}
\mathsf{choose}(u) &= (\assign{t_{u}}{0}) \choice (\assign{f_{u}}{0})\\
\funcAG &= \mathsf{choose}(u_0); \mathsf{choose}(e_1); \mathsf{choose}(u_1); \cdots; \mathsf{choose}(u_n); (\assign{s}{1})
\end{empheq}
\caption{Implementation of the Assignment Guesser $\funcAG$ gadget}
\label{fig:ass-appendix}
\end{figure}

\begin{lemma}
$\Psi_n$ is true iff we reach the label $\lambda_1$ of the $\funcSAT$ gadget (ref. Figure \ref{fig:sat-appendix}) 
in some thread, and the label $\lambda_2$ of the $\funcSAT$ gadget in some thread.
\label{lem:base} 	
\end{lemma}

\begin{lemma}
For $0 \leq i \leq n-1$, $\Psi_i$ is satisfiable iff we reach the label $\lambda_3$ in the $\funcFEC{i}$ gadget (ref. Figure \ref{fig:aec-appendix}) in some thread, and the label $\lambda_4$ 
in the $\funcFEC{i}$ gadget  
in some thread.\label{lem:induct}	
\end{lemma}

\begin{lemma}
$\assert{\texttt{false}}$ is reachable iff 
we reach 
the label $\lambda_3$ in  the $\funcFEC{0}$ gadget  
in some thread, and the label $\lambda_4$ 
in the $\funcFEC{0}$ gadget  
in some thread. 
\label{lem:goal}	
\end{lemma}

In the following, 
we write $\Phi$ for $\Phi(u_0, e_1, \dots, e_n, u_n)$ since the 
free variables of $\Phi$ are clear.

\begin{figure}[h]
\centering
\setlength{\belowcaptionskip}{-7pt}
\begin{empheq}[box={\mycode[rounded corners]}]{align*}
\funcSATC = &\assume{(s = 1)}; \mathsf{check}(\Phi); \lambda_0: \mathsf{skip}; \\
&[(\assume{(t_{u_n} = 0)};  \assign{a_{n,1}}{1}; \lambda_1 : \mathsf{skip};)\\ 
& \choice\\
&(\assume{(f_{u_n} = 0)};  \assign{a_{n,0}}{1}; \lambda_2 : \mathsf{skip};)]
\end{empheq}
\caption{Implementation of the SAT Checker $\funcSATC$ gadget with labels $\lambda_0, \lambda_1, \lambda_2$}
\label{fig:sat-appendix}
\end{figure}

\subsubsection*{Proof of Lemma \ref{lem:base}}
 
Assume $\Psi_n$ is satisfiable.   
Then there are satisfying assignments 
$\alpha_1$ and $\alpha_2$ s.t. $\alpha_1(u_n)=0$, 
$\alpha_2(u_n)=1$, such that $\alpha_1, \alpha_2 \models \Phi$. These assignments $\alpha_1, \alpha_2$ can be guessed by  
$\funcAG$ gadgets 
in two threads, resulting in adding messages 
$(s,1,\view_1)$ and $(s,1,\view_2)$ to the memory, such that 
$\view_1(f_{u_n})=0$ and $\view_2(t_{u_n})=0$. Correspondingly, there are  $\funcSATC$ gadgets which read from these views, (they read 1 from $s$), 
and check for the satisfiability of $\Phi$ using the $\view_1, \view_2$ 
values of $t_x, f_x$ for $x \in Vars(\Psi)$. Since both are satisfying assignments, the label $\lambda_0$  is reachable in both 
$\funcSATC$ gadgets. One of them will reach the label $\lambda_1$  reading 
$t_{u_n}=0$ (using $\view_2$) and the other 
will reach the label $\lambda_2$ reading $f_{u_n}=0$ (using $\view_1$). 

Conversely, assume that the label $\lambda_1$  of $\funcSATC$ is reachable in one thread, while the label $\lambda_2$  of $\funcSATC$  is reachable in another thread. Then we know that in one thread, we have read a message 
$(s,1,\view_1)$, checked for the satisfiability of $\Phi$ using $\view_1$, and also verified that $\view_1(t_{u_n})=0$, while 
in another thread,    we have read a message 
$(s,1,\view_2)$, checked for the satisfiability of $\Phi$ using $\view_2$, and also verified that $\view_2(f_{u_n})=0$. Thus, we have 2 satisfying assignments to $\Phi$, one where 
$u_n$ has been assigned to 0, and the other, where $u_n$ has been assigned 1. Hence $\Psi_n$ is satisfiable. 
\qed 

\begin{definition}
Let $\view$ be a view. We say that an assignment $\alpha : Vars(\Psi) \rightarrow \{0,1\}$ 
is embedded in $\view$ iff for all $x \in Vars(\Psi)$, $\view(t_x)=0 \Leftrightarrow \alpha(x)=1$ and  
	$\view(f_x)=0 \Leftrightarrow \alpha(x)=0$. The term ``embedded'' is used since the view 
	also has (program) variables outside of $t_x$ and $f_x$.
\end{definition}

For $0 \leq i \leq n$, let $\pA{i}$ and $\yA{i}$ respectively represent the set of assignments which are \emph{embedded} 
in the views reaching the labels $\lambda_3, \lambda_4$ of the $\funcFEC{i}$ gadget. Thus, we know that 
$$\pA{n}=\{\alpha \models \Phi \mid \alpha(u_n)=1\}, ~\text{and}~ 
\yA{n}=\{\alpha \models \Phi \mid \alpha(u_n)=0\}$$ 

\begin{figure}[ht]
\centering
\setlength{\belowcaptionskip}{-7pt}
\setlength{\abovecaptionskip}{3pt}
\begin{empheq}[box={\mycode[rounded corners]}]{align*}
\funcFEC{i} =~ &[\assume{(a_{i+1,0} = 1)}; \assume{(a_{i+1,1} = 1)}]; 
\kappa_1: \mathsf{skip};\\
&
[\assume{(f_{e_{i+1}} = 0)} \choice \assume{(t_{e_{i+1}} = 0)}]; \kappa_2: \mathsf{skip}; \\
& [(\assume{(t_{u_i} = 0)}; \assign{a_{i,1}}{1}; \lambda_3: \mathsf{skip};) \choice (\assume{(f_{u_i} = 0)}; \assign{a_{i,0}}{1};\lambda_4: \mathsf{skip};)]
\end{empheq}
\caption{$\forall\exists$ Checker at level $i$, $\funcFEC{i}$, with labels $\kappa_1, \kappa_2, \lambda_3, \lambda_4$.  We have $n$ such gadgets, one for each level $0 \leq i \leq n-1$.}
\label{fig:aec-appendix}
\end{figure}

\begin{lemma}
For $0 \leq i \leq n-1$, define sets of assignments $$\pA{i,0}=\{\alpha \in \pA{i+1} \uplus \yA{i+1} \mid \alpha(u_i)=1, \alpha(e_{i+1})=0\}$$
	$$\pA{i,1}=\{\alpha \in \pA{i+1} \uplus \yA{i+1} \mid \alpha(u_i)=1, \alpha(e_{i+1})=1\}$$
	$$\yA{i,0}=\{\alpha \in \pA{i+1} \uplus \yA{i+1} \mid \alpha(u_i)=0, \alpha(e_{i+1})=0\}$$
	$$\yA{i,1}=\{\alpha \in \pA{i+1} \uplus \yA{i+1} \mid \alpha(u_i)=0, \alpha(e_{i+1})=1\}$$
	where $\uplus$ denotes disjoint union. 	Then $\pA{i}$ is equal to one of the sets $\pA{i,1}$ or $\pA{i,0}$. Similarly, 
	$\yA{i}$ is equal to one of the sets $\yA{i,1}$ or $\yA{i,0}$.
 \label{obs}
 \end{lemma}
\noindent{Proof of Lemma \ref{obs}}. We already know the definitions 
of $\pA{n}$ and $\yA{n}$. Consider the case of $\pA{n-1}$ and $\yA{n-1}$. By construction, to reach label $\lambda_3$  
of $\funcFEC{n-1}$,
\begin{itemize}
\item[(a)] we need to have reached labels $\lambda_3, \lambda_4$ 
of $\funcFEC{n}$. The view (say $\view_{A}$) on reaching the label $\kappa_1$ 
in $\funcFEC{n-1}$ has embedded assignments from $\pA{n} \uplus \yA{n}$.   
\item[(b)] To reach the label $\kappa_2$ of $\funcFEC{n-1}$, we need either $f_{e_n}$ to have 
time stamp 0 or $t_{e_n}$ to have time stamp 0
in $\view_{A}$. If we had $\view_{A}(t_{e_n})>0$ and
$\view_{A}(f_{e_n})>0$, then the label  $\kappa_2$  is not reachable. That is, the assignments embedded in 
$\view_{A}$ agree 
on the assignment of $e_n$. 
\item[(c)] To reach the label $\lambda_3$ in $\funcFEC{n-1}$, 
the assignments embedded in $\view_{A}$ agree 
on the assignment of $u_{n-1}$, such that $u_{n-1}$ is assigned 1. 
Thus, $\pA{n-1}$ is obtained from $\pA{n} \uplus \yA{n}$ by keeping those assignments which agree 
on $e_n$ and where $u_{n-1}$ is true. 
    	\end{itemize}

Similarly, to reach label $\lambda_4$ in $\funcFEC{n-1}$, 
\begin{itemize}
\item[(a)] we need to have reached the labels $\lambda_3, \lambda_4$  
of $\funcFEC{n}$.  The view (say $\view_B$) on reaching the label $\kappa_1$ 
in $\funcFEC{n-1}$ has embedded assignments from $\pA{n} \uplus \yA{n}$.   
\item[(b)] To reach label $\kappa_2$, we need either $f_{e_n}$ to have 
time stamp 0 or $t_{e_n}$ to have time stamp 0
in $\view_B$. If we had $\view_B(t_{e_n})>0$ and
$\view_B(f_{e_n})>0$, then the label $\kappa_2$ is not reachable. That is, the assignments embedded in $\view_B$ agree 
on the assignment of $e_n$. 
\item[(c)] To reach the label $\lambda_4$ in $\funcFEC{n-1}$, 
the assignments embedded in $\view_B$ agree 
on the assignment of $u_{n-1}$, such that $u_{n-1}$ is assigned 0. 
Thus, $\yA{n-1}$ is obtained from $\pA{n} \uplus \yA{n}$ by keeping those assignments which agree on $e_n$ and where $u_{n-1}$ is false. 
    	\end{itemize}

The proof easily follows for any $\pA{i}, \yA{i}$, using the definitions of $\pA{i+1}, \yA{i+1}$ as above. \qed

\subsubsection*{Proof of Lemma \ref{lem:induct}}
We give an inductive proof for this, using Lemma \ref{lem:base} as the base case. As the inductive step, assume that $\Psi_{i+1}$ is satisfiable iff we reach the label $\lambda_3$  of the 
$\funcFEC{i+1}$
gadget in some thread, and the label $\lambda_4$  of  
the $\funcFEC{i+1}$ gadget in some thread. 

\medskip 

Assume $\Psi_i$ is satisfiable. We can write $\Psi_{i}$ as 
$\forall u_{i}\exists e_{i+1} \Psi_{i+1}$.
 We show that there is a thread   
  which reaches label $\lambda_3$  
of the $\funcFEC{i}$
gadget  with a view that has $\pA{i}$
\emph{embedded} in it, 
 and there is a thread which 
reaches the label $\lambda_4$  of the $\funcFEC{i}$ gadget 
with a view that has $\yA{i}$ \emph{embedded} in it.
\medskip

By inductive hypothesis, 
since $\Psi_{i+1}$ is satisfiable, 
there is a  thread which reaches the label $\lambda_3$  
of the $\funcFEC{i+1}$ gadget with a view $\view_A$ that has $\pA{i+1}$ embedded in it, 
 and there is a thread which 
reaches label $\lambda_4$ of the $\funcFEC{i+1}$ gadget with a view $\view_B$ 
that has $\yA{i+1}$ embedded in it. 
Note that $a_{i+1,1}, a_{i+1,0}$ have been written 1 by these threads respectively, 
such that $\view_A(a_{i+1,1})>0$ and 
$\view_B(a_{i+1,0})>0$. 
Thanks to this, there is a thread which 
can take on the role of the $\funcFEC{i}$ gadget now. 
This thread begins with a view $\view_C$ which is the merge of $\view_B$ and 
$\view_A$. The label $\kappa_1$ of  this $\funcFEC{i}$ gadget is reachable 
by reading 1 from both $a_{i+1,1}$ and $a_{i+1,0}$, and we want   
 $\view_C(t_{e_{i+1}})=0$ or $\view_C(f_{e_{i+1}})=0$. As seen in item(b) in the proof 
 of observation \ref{obs}, this is possible only if 
 $\view_B(t_{e_{i+1}})=0$ 
 and $\view_A(t_{e_{i+1}})=0$, or 
 $\view_B(f_{e_{i+1}})=0$ 
 and $\view_A(f_{e_{i+1}})=0$.

 By assumption, since $\Psi_i$ is satisfiable, there exists assignments 
   from $\pA{i+1}$ and $\yA{i+1}$ which agree on $e_{i+1}$ and $u_i$. In particular, 
   the satisfiability of  $\Psi_i=\forall u_i \exists e_{i+1} \Psi_{i+1}$ says that 
   we have a set of assignments $S \subseteq \pA{i+1} \uplus \yA{i+1}$ which satisfy $\Psi_i$, 
   such that for all $\alpha \in S$, $\alpha(u_i)=1$ and $\alpha(e_{i+1})$ is some fixed 
   value.  Similarly, the satisfiability of $\Psi_i$ also gives us a 
   set of assignments $S' \subseteq \pA{i+1} \uplus \yA{i+1}$ 
   such that for all $\alpha \in S'$, $\alpha(u_i)=0$ and $\alpha(e_{i+1})$ is some fixed 
   value. It is easy to see that $S=\pA{i}$, while $S'=\yA{i}$.  Thus, 
   the satisfiability of $\Psi_i$ implies the feasibility 
   of the assignments $\pA{i}$ and $\yA{i}$. This in turn, gives us the following. 
   
  Thus, starting with a view $\view_C$ which has embedded assignments 
  $\pA{i+1} \uplus \yA{i+1}$,  it is possible for a thread to 
\begin{enumerate}
\item read 1 from $a_{i+1,1}$ and $a_{i+1, 0}$ (these are present 
in the $\view_C$), 
\item Check that either $t_{e_{i+1}}$ has time stamp 0 in $\view_C$ 
or $f_{e_{i+1}}$ has time stamp 0 in $\view_C$ (this is possible 
since the embedded assignments agree on $e_{i+1}$),
\item Check that $t_{u_i}$ has time stamp 0 in $\view_C$ (this is possible 
since the embedded assignments are such that $u_i$ is assigned 1)   	
\end{enumerate}
This ensures that the thread reaches the label $\lambda_3$ of $\funcFEC{i}$
with a view having $\pA{i}$ embedded in it (notice that the last two checks filter out $\pA{i}$ from  $\pA{i+1} \uplus \yA{i+1}$).

In a similar manner, starting with a view $\view_C$ which has embedded assignments 
  $\pA{i+1} \uplus \yA{i+1}$,  it is possible for a thread to 
\begin{enumerate}
\item read 1 from $a_{i+1,1}$ and $a_{i+1, 0}$ (these are present 
in the $\view_C$), 
\item Check that either $t_{e_{i+1}}$ has time stamp 0 in $\view_C$ 
or $f_{e_{i+1}}$ has time stamp 0 in $\view_C$ (this is possible 
since the embedded assignments agree on $e_{i+1}$),
\item Check that $f_{u_i}$ has time stamp 0 in $\view_C$ (this is possible since the embedded assignments are such that $u_i$ is assigned 0)   	
\end{enumerate}
This ensures that the thread reaches the label $\lambda_4$  of $\funcFEC{i}$
with a view having $\yA{i}$ embedded in it (notice that the last two checks filter out $\yA{i}$ from  $\pA{i+1} \uplus \yA{i+1}$).

\medskip 
Conversely, assume that we have two threads which have reached respectively, labels $\lambda_3, \lambda_4$ 
of $\funcFEC{i}$ gadget having views in which $\yA{i}$ and 
$\pA{i}$ are embedded. 
 We show that $\Psi_i$ 
is satisfiable. 

By the definition of $\pA{i}$, we know that we have assignments from 
$\pA{i+1} \uplus \yA{i+1}$ which agree on $e_{i+1}$, and 
which set $u_i$ to 1. The fact that we reached the label $\lambda_3$ 
of $\funcFEC{i}$ gadget with a view having  $\pA{i}$ embedded in it shows that 
these assignments are feasible. Similarly, 
reaching the label $\lambda_4$  of $\funcFEC{i}$ with 
a view having $\yA{i}$ embedded in it 
shows that we have assignments from 
$\pA{i+1} \uplus \yA{i+1}$ which agree on $e_{i+1}$, and 
which set $u_i$ to 0. The existence of these two assignments 
proves the satisfiability of $\Psi_i$.

\subsubsection*{Proof of Lemma \ref{lem:goal}}
Assume that we reach $\assert{\texttt{false}}$. Then we have read 1 from $a_{0,0}$ and 
$a_{0,1}$. These are set to 1 only when the labels $\lambda_4, \lambda_3$  
of $\funcFEC{0}$ have been visited. The converse 
is exactly similar : indeed if we reach the labels $\lambda_4, \lambda_3$ 
of $\funcFEC{0}$, we have written 1 to 
 $a_{0,0}$ and $a_{0,1}$. This enables 
 the reads of 1 from $a_{0,0}$ and 
$a_{0,1}$ leading to $\assert{\texttt{false}}$.

\begin{theorem}
Parameterized safety verification for $\thatthread(\uncassy, \unloopy)$ is \pspace-hard. 
\label{thm:pspace-c}
\end{theorem}

\section{Limits of Semantic Simplification II: \nexp-hardness of $\thatthread(\uncassy,\unloopy) \parallel \thisthread_1(\uncassy)$} 
In this section we show an \nexp~lower bound on the safety verification problem in the presence of a single leader $\thisthread$ thread, $\thatthread(\uncassy, \unloopy)||\thisthread(\uncassy)$.
 The lower bound is obtained with a 
 fragment of RA which does not use registers and surprisingly in which $\thisthread$ also does not perform any compare-and-swap operations. As in the case 
 of the \pspace-hardness, we work with a fixed set of shared memory locations $\setvars$ (also called shared variables) from a finite  data domain $\setvals$. We show the hardness via a reduction from the succinct version of 3CNF-SAT, denoted $\succsat$. 
 Following the main part of the paper, we refer to the distinguished $\thisthread$ thread as the `leader' and individual threads from $\thatthread$ as `contributors'.

\subsection{$\succsat$: succinct satisfiability}

The complexity of succinct representations was studied in the pioneering work \cite{DBLP:journals/iandc/GalperinW83} for graph problems. Typically, the complexity of a problem is measured as a function of some quantity $V$, with the assumption that the input size is polynomial in $V$. If the underlying problem concerns graphs, then $V$ is the number of vertices in the graph, while if the underlying problem concerns boolean formulae, then  
$V$ is the size of the formula. \cite{DBLP:journals/iandc/GalperinW83} investigated the complexity  of graph problems, when the input has an \emph{exponentially  succinct} representation, that is, the input size is polylog in $|V|$, where $V$ is the 
number of vertices of the graph, and showed that succinct representations rendered trivial graph problems NP-complete, while \cite{DBLP:journals/iandc/PapadimitriouY86} showed that graph properties which are NP-complete under the usual representation became \nexp-complete under succinct representations.
$\succsat$ is essentially a exponentially succinct encoding of a 3CNF-SAT problem instance. Let $\phi(x_0, \cdots, x_{2^n-1})$ be a 3CNF formula with $2^n$ variables and $2^n$ clauses. Assume an $n$ bit binary address for each clause. 

A succinct encoding of $\phi$ is a circuit $D(y_1, \cdots, y_n)$ (with size polynomial in $n$), which, on an $n$ bit input  $y_1\cdots y_n$ interpreted as a binary address for clause $c$, generates $3n + 3$ bits, specifying the indices of the 3 variables from $x_1, \cdots, x_{2^n}$ occurring in clause $c$ and their signs (1 bit each). 
Thus, the circuit $D$ provides a complete description of $\phi(x_1, \cdots, x_{2^n})$ when evaluated with all $n$-bit inputs. Define $\succsat$ as the following \nexp-complete \cite{DBLP:journals/iandc/PapadimitriouY86} problem. 

\begin{center}
\it
Given a succinct description $D$ of $\phi$, check whether $\phi$ is satisfiable. 
\end{center}

Adopting the notation above, we assume that we have been given $n$, the formula $\phi$ with $2^n$ boolean variables $\bvars = \{x_0, \cdots, x_{2^n-1}\}$, and the succinct representation $D$ with input variables $\{y_1, \cdots, y_n\}$. Denote the variables in clause $c$ as $\va(c), \vb(c), \vc(c)$ and their signs as $\siga(c), \sigb(c), \sigc(c)$. We denote the $n$-bit address $\Bar{c}$ of a clause $c$ as a (boolean) word $\Bar{c}\in \{0,1\}^n$ and commonly use the variable $\alpha$ to refer to clause addresses. We denote the variable addresses also as $n$-bit (boolean) words and commonly use the the variable $\beta$ to represent them. We construct an instance of the parametrized reachability problem consisting of a $\leader$ leader thread running program $\comlead$ and the $\thatthread$ contributor threads running program $\comcont$. We show that this system is `unsafe' (an $\assert{\texttt{false}}$ is reachable) if and only if the $\succsat$ instance is satisfiable. 

\tikzset{background rectangle/.style={fill=none
}}
\begin{figure}[h]
\centering
\small
\begin{tikzpicture}[codeblock/.style={line width=0.5pt, inner xsep=0pt, inner ysep=5pt}  , show background rectangle]
\node[codeblock] (init) at (current bounding box.north west) {
$
\def\arraystretch{1.5}
\begin{array}{rl}
\rowcolor{black!15}
\com_\thatthread &= \funcCLENC \choice \funcSAT \choice \funcFC{0} \choice  \cdots \choice \funcFC{n-1} \choice \funcGC \\
\rowcolor{black!5}
& \mathsf{choose}(u) = (\assign{t_{u}}{1}) \choice (\assign{f_{u}}{1}) \\
\rowcolor{black!5}
\funcCLENC &= \mathsf{choose}(u_0); \mathsf{choose}(u_1); \cdots; \mathsf{choose}(u_{n-1}); \assign{s}{1}\\
\rowcolor{black!15}
\funcSAT &= ~(\assume{s = 1}); ~\com_{\mathsf{CV}}; ~\com_{\mathsf{Check}}; \\
\rowcolor{black!15}
		   &~((\assume{(t_{u_{n-1}} = 0)}; ~\assign{a_{n-1,1}}{1}) ~\choice ~(\assume{(f_{u_{n-1}} = 0)}; ~\assign{a_{n-1,0}}{1}))\\
\rowcolor{black!5}
\funcFC{i} &=~ \assume{a_{i+1,0} = 1}; ~\assume{a_{i+1,1} = 1}; \\
\rowcolor{black!5}
	 &~((\assume{t_{u_i} = 0}; ~\assign{a_{i,1}}{1}) ~\choice ~(\assume{f_{u_i} = 0}; ~\assign{a_{i,0}}{1}))\\
\rowcolor{black!15 }
	\funcGC &= \assume{(a_{0,0} = 1)}; ~\assume{(a_{0,1} = 1)}; ~\assert{\texttt{false}}
\end{array}$
};
\end{tikzpicture}
\vspace{-0.2cm}
\caption{The contributor program $\com_\thatthread$ used in the reduction. The sub-routines 
$\com_{\mathsf{CV}}$ and  $\com_{\mathsf{Check}}$ are 
described later.
}
\label{fig:nexphapp}
\end{figure}

\subsection{Key features}

The leader running program $\com_\leader$ guesses an assignment to the boolean variables in $\phi$. The contributors running the program $\com_\thatthread$ will be tasked with checking that the assignment guessed by the leader does in fact satisfy the formula $\phi$. They do this in a distributed fashion, where one clause from $\phi$ is verified by one contributor. Then similar to the \pspace-hardness proof, the program $\comcont$ forces the contributors to combine checks for individual clauses as a dependency tree. This is so that the root of the tree is able to reach an assertion failure only if all threads could successfully check their clauses under the leader's guessed assignment. However since all the contributors run the same program, the trick is to enforce that all clauses will be checked.

\textbf{Gadgets.} $\comcont$ consists of a set of gadgets (modelled as `functions' in the program), only one of which may be non-deterministically executed by the contributors, while $\com_\leader$ is the program executed by the leader.

\begin{equation*}
	\comcont = \funcCLENC \choice \funcSAT \choice \funcFC{0} \choice  \cdots \choice \funcFC{n-1} \choice \funcGC 
\end{equation*}
Recall that in the \pspace-hardness proof too, there were similar gadgets which were executed by the $\thatthread$ threads. The gadgets in $\comcont$ do execute various tasks as follows. 
\setlist[itemize]{leftmargin=1cm}
\begin{itemize}
	\item[$\funcCLENC$] guesses an $n$ bit address $\Bar{c}$ of a clause $c$ in $\phi$,
	\item[$\funcSAT$] (1) acquires a clause address $\Bar{c}$ generated by $\funcCLENC$, (2) uses the circuit $D$ to obtain the indices of variables $\va{c}$, $\vb{c}$, $\vc{c}$ in clause $c$, along with its sign (this is done by the sub-routine $\com_\mathsf{CV}$), (3) accesses the assignment made to the variables by the leader (sub-routine $\com_{\mathsf{Check}}$) and (4) the assignment is such that $c$ is satisfied.
	\item[$\funcFC{i}$] ($0 \leq i \leq n-2$) together ensure that the satisfiability of all the clauses in $\phi$ has been checked. This is done by instantiations of $\funcSAT$, in levels (similar to the proof of \pspace-hardness). At the $i$th level, $\funcFC{i}$  checks the $\forall$ universality of the $i^{th}$ address bits of clause $c$.
	\item[$\funcGC$] finally reaches $\assert{\texttt{false}}$, if all the previous functions execute faithfully, implying that the $\succsat$ instance is satisfiable.
\end{itemize}

The non-deterministic branching implies that each $\thatthread$ thread will only be able to execute one of these gadgets. The check for satisfiability of $\phi$ is distributed between the $\thatthread$ threads much like the \pspace-hardness construction. For this distributed check, threads are allocated roles depending upon the function (gadget) they execute. Additionally, the distinguished leader thread is tasked with guessing the assignment. We now describe this.

\textbf{Role of the leader.} We have one leader thread which guesses a satisfying assignment for the boolean variables $\bvars$ as a string of writes made to a special program variable $g$. The writes made to $g$ have $n+2$ values $\valt, \valf, 1, \dots, n$ in a specific order. Let the initial values of all variables in the system be $\init \notin \{\valt, \valf,1,\dots, n\}$. To illustrate a concrete example, consider the case where $n = 3$. Let the guessed assignment for $\bvars$ be $w = \texttt{ftftttff} \in \{\texttt{t,f}\}^{2^3}$, where $\texttt{t}$ denotes true and $\texttt{f}$ false. Then the writes made by the leader are as below, where $\valt$ and $\valf$ are macros for data domain values (other than $\{\init, 1, \cdots, n\}$) representing true and false respectively.
    \begin{equation*}
            \red{\valf} ~ \blu{1} ~ \red{\valt} ~ \blu{2} ~ \red{\valf} ~ \blu{1} ~ \red{\valt} ~ \blu{3} ~ \red{\valt} ~ \blu{1} ~ \red{\valt} ~ \blu{2} ~ \red{\valf} ~ \blu{1} ~ \red{\valf}
	\end{equation*}
	The leader alternates writing a guessed assignment for $x_0, \dots, x_7$ (in $\red{\text{red}}$) with writing a value from $\{1,\dots,n\}$ (in $\blu{\text{blue}}$). The values in $\{1, \dots, n\}$ (here $\{1,2,3\}$) in blue are written in a deterministic pattern as $\blu{1}~ \blu{2} ~ \blu{1} ~ \blu{3} ~\blu{1}~ \blu{2} ~ \blu{1}$, which we call a `binary search pattern' with 3 values, denoted $\bsp(3)$ for short. $\bsp(n)$ is a unique word of length $2^n -1$ over $\{1,\dots,n\}$, defined inductively as follows.
	\begin{align*}
		\bsp(1) &= 1 \\
		\bsp(n) &= \bsp(n-1) \cdot n \cdot \bsp(n-1) \quad \text{ for } n \geq 2
	\end{align*} 
	
	The assignments for $x_0, \dots, x_{2^n-1}$ are interspersed alternately with symbols in $\bsp(n)$ by the leader while writing to $g$. Formally, let $\mathcal{S}(n, w)=\bsp(n) \shuffle w$ represent the perfect shuffle (alternation) of $\bsp(n)$ with the guessed assignment $w\in\{\valt,\valf\}^{2^n}$. The leader writes the word $\mathcal{S}(n, w)$ to $g$. 
	From the example above, $\mathcal{S}(3, \texttt{ftftttff}) = \red{\valf} ~ \blu{1} ~ \red{\valt} ~ \blu{2} ~ \red{\valf} ~ \blu{1} ~ \red{\valt} ~ \blu{3} ~ \red{\valt} ~ \blu{1} ~ \red{\valt} ~ \blu{2} ~ \red{\valf} ~ \blu{1} ~ \red{\valf}$. 
	We show that the shuffle sequence which need to generated is easily implementable by the leader with a polysized program.

	\begin{lemma}
		There exists a program \emph{$\comlead$}, which nondeterministically choses $w \in \{\valt, \valf\}^{2^n}$ and generates the write sequence $\mathcal{S}(n, w)$ on a shared memory location $g$, with the size of the program growing polynomially with $n$.
	\end{lemma}

\textbf{How contributors access variable assignments, intuitively.}
 
	Each contributor wants to check a single clause for which it needs to access the 3 variables and their signs occurring in that clause. Since it pertains to the $\bsp$, we first understand this task and discuss the others (selecting clause, acquiring variable address and sign, etc.) later. For now we assume that the contributor has a variable $x$ with address $\beta$ and sign $\sigma$ wants to access the assignment made to variable $x$ by the leader. 
	
	For boolean variable $x$, the contributor uses the $\bsp(n)$ pattern to locate the assignment made to $x$, by reading a subword of $\mathcal{S}(n, w)$. From program variable $g$, the contributor reads $n+1$ values $\{\valf, 1, \dots, n\}$ or $\{\valt,1, \dots, n\}$ without repetitions, depending upon the sign of $x$ in the clause ($\valf$ if sign is negative, $\valt$ if positive). In the running example, if contributor wants to access $x_2$ from $\red{\valf} ~ \blu{1} ~ \red{\valt} ~ \blu{2} ~ \red{\valf} ~ \blu{1} ~ \red{\valt} ~ \blu{3} ~ \red{\valt} ~ \blu{1} ~ \red{\valt} ~ \blu{2} ~ \red{\valf} ~ \blu{1} ~ \red{\valf}$, it reads the sequence $\blu{2} ~ \red{\valf} ~ \blu{1} ~ \blu{3}$. Likewise, the  value of $x_6$ is obtained  by reading  $\blu{3}~\blu{2}~ \red{\valf}~\blu{1}$, while  for $x_0$, the contributor must read $\red{\valf}~\blu{1}~\blu{2}~\blu{3}$. We note that for each $x\in \bvars$, there is a unique `access pattern', which forces the thread to acquire the assignment of exactly $x$ and not any other variable. In this search, it is guided by the $\bsp$, which acts as an indexing mechanism. It helps the contributor narrow down unambiguously, to the part which contains the value of $x$. 

\subsubsection{Formal description of contributors}

The contributors check that each clause in $\phi$ has been satisfied in a distributed fashion. Each contributor executes one of the functions in $\comcont$. They do this as follows.

    \textbf{Clause Encoding}: $\funcCLENC$: A thread executing $\funcCLENC$ selects, nondeterministically, a clause address $\alpha\in \{0,1\}^n$. This is done by writing 1 to either $t_{u_i}$ or $f_{u_i}$ for all $0 \leq i \leq n-1$. Finally, 1 is written into a special variable $s$. The function $\funcCLENC$ in Figure \ref{fig:nexphapp} describes this. The view of the message $(s, 1, \aview)$ encodes the address $\alpha$ of a clause satisfying
\begin{align*}
	(\aview(t_{u_i}) > 0 \iff \alpha[i] = 0) \text{ and } (\aview(f_{u_i}) > 0 \iff \alpha[i] = 1) \text{ for } 0 \leq i \leq n-1
\end{align*}
	Recall that this is the same encoding technique as used for the \pspace-hardness proof. Each bit is encoded in the view of a message. Overall $2^n$ threads will execute $\funcCLENC$ to cover all the clauses in the formula.

    \textbf{Satisfaction checking (for one clause)}: $\funcSAT$: A thread executing $\funcSAT$ acquires the address $\Bar{c}$ of a clause $c$ through the view $\aview$ by reading the message $(s,1,\aview)$ generated by $\funcCLENC$. This thread has to check the satisfiability of the clause with address $\Bar{c}$. For this, it needs to know the 3 boolean variables $\va(c)$, $\vb(c)$, $\vc(c)$ appearing in $c$. Recall that we have been given, as part of the problem, the circuit $D$ which takes an $n$ bit address $\alpha$ corresponding to some clause as input, and outputs the $3n+3$ bits corresponding to the 3 variables appearing in the clause, along with their signs. We use $D$, and the encoding of a clause address $\Bar{c}$ stored in $\aview$, to compute $D(\alpha)$. We have a polysized sub-routine $\com_{\mathsf{CV}}$ (CV for circuit value) that can compute the circuit value of $D$.
        
	\textbf{Circuit Value}: $\com_{\mathsf{CV}}$: The $\com_{\mathsf{CV}}$ sub-routine takes the address $\alpha$ (of a clause $c$) and converts it into the index of one of the variables in $c$. Thus in essence, $\com_{\mathsf{CV}}$ evaluates the circuit-value problem $D(\alpha)$ by simulating the (polynomially-many) gates in $D$. The idea is to keep two boolean program variables for each node in $D$, and propagate 
	the evaluation of nodes in an obvious way (for instance, if we have $\wedge$ gate with input gates $g_1, g_2$ evaluating  
	  to 0, and 1 respectively, then $t_{\wedge}$ will be written 1). 
	We now briefly explain how circuit value can be evaluated, by taking an example of a single gate.

	\newcommand{\encAddr}{\mathsf{encAddr}}
	For each node $p$ in $D$ we use two boolean program variables, $t_p$ and $f_p$. We say that a view $\aview$ encodes the value at node $p$ if the following holds. We write $\mathsf{encAddr}(\aview)$ to denote the value values for boolean variables encoded in $\aview$.
	\begin{equation}
		(\aview(t_p) > 0 \iff p = 0) \text{ and } (\aview(f_p) > 0 \iff p = 1)
		\label{eqn:viewenc}
	\end{equation}
\begin{wrapfigure}{r}{0.4\linewidth}
\centering
\begin{tikzpicture}[label distance=2mm]
    \node[nand gate US, draw, scale=2] at (0,0) (nand) {};

    \node (i1) at ($(nand) + (-2, 0.5)$) {$t_{i_1}, f_{i_1}$};
    \node (i2) at ($(nand) + (-2, -0.5)$) {$t_{i_2}, f_{i_2}$};

    \draw (nand.output) -- node[above]{$t_{o}, f_{o}$} ($(nand) + (2, 0)$);
	\draw (i1) |- (nand.input 1);
	\draw (i2) |- (nand.input 2);
\end{tikzpicture}
\end{wrapfigure}

	Now assume a thread has a view $\aview_1$ when it wants to evaluate a logic (NAND) gate $G$, with output node $o$ and input nodes $i_1$ and $i_2$. We assume $\aview_1$ must encode the values of $i_1$ and $i_2$ (the thread has evaluated inputs of $G$) and the thread must not have evaluated $G$ before (we have that $\aview_1(t_o) = \aview_1(f_o) = 0$). Assuming that these conditions hold, the thread executes instructions such that the new view $\aview_2$ of the thread (1) differs from $\aview_1$ only on $t_o$ and $f_o$ and (2) $\aview_2$ correctly encodes value of $o$. The function in Figure \ref{fig:nand} evaluates $G$.
\begin{figure}[h]
\centering
\setlength{\belowcaptionskip}{-7pt}
\setlength{\abovecaptionskip}{3pt}
\begin{empheq}[box={\mycode[sharp corners]}]{align*}
\com_{\mathsf{NAND}} = &~(((\assume{f_{i_1} = \init})\choice(\assume{f_{i_2} = \init})); \assign{f_o}{1}) \choice \\
	&~((\assume{t_{i_1} = \init}); (\assume{t_{i_2} = \init}); (\assign{t_o}{1}))
\end{empheq}
\caption{$\com_{\mathsf{NAND}}$ - encoding evaluation for a NAND Gate in views of threads}
\label{fig:nand}
\end{figure}
	The main observation is that a thread can read $\init$ from a variable only if its view on that variable is 0 (since there is only one $\init$ message with timestamp 0). Claim (1) holds trivially since only timestamps of only $t_o$ or $f_o$ may be augmented (reading from $\init$ will not change timestamp). By a little observation we see that the thread can write to $f_o$ if one of $f_{i_\{1,2\}}$ have timestamp 0, implying that one of the inputs to the gate is 0 by the assumption. Then it checks out that the new view on $f_o$ is greater than 0, thus claim (2) holds. The case for $t_o$ may be checked similarly. Since $D$ has polynomially many gates, any thread can evaluate them in topological order, and hence eventually will end up with the evaluation of $D(\alpha)$. Also note that since the thread relied on its internal view, the same set of program variables $\{t_p,~ f_p | p \in G\}$ may be used by all threads (hence crucially avoiding the number of variables to vary with the thread count).

\begin{lemma}
	There exists a sub-routine $\com_{\mathsf{CV}}$ that starting with the view $\aview$ from $(s,1,\aview)$, evaluates the circuit value $D(\alpha)$, where $\alpha$ is the clause address encoded in the variables $t_{u_i}$ and $f_{u_i}$ in $\mathsf{encAddr}(\aview)$. Also, $\com_{\mathsf{CV}}$ is polysized in $n$.
\end{lemma}
	
	Once $D(\alpha)$ has been computed, the thread can nondeterministically choose one of the three variables appearing in clause $c$, say $x \in \{\va(c), \vb(c), \vc(c)\}$. For simplicity we include this as a part of the routine $\com_{\mathsf{CV}}$ itself. The address $\beta$ of the variable $x$, the contributor accesses the assignment made by the leader to $x$ and checks if it satisfies  clause $c$. This is done by the routine $\com_{\mathsf{Check}}$. 

    \textbf{Clause Check}: $\com_{\mathsf{Check}}$: Having acquired the address $\beta = \beta_{n-1}, \dots, \beta_0$ and sign $\sigma$ of variable $x$, by executing $\com_{\mathsf{CV}}$, the thread checks that variable $x$ satisfies clause $c$. To faithfully access the assignment to $x$ from the variable $g$, the $\bsp$ 
      guides the thread. The `access pattern' for $x$ denoted by $\mathsf{AP}(n)$ ($\beta, \sigma$ implicit) which is recursively defined as 
\begin{align*}
	\text{for } 0 < i \leq n ~ \mathsf{AP}(i) &= \begin{cases} 
		i \cdot \mathsf{AP}(i-1) &\mbox{ if } \beta_{i-1}=1 \\ 
		\mathsf{AP}(i-1) \cdot i &\mbox{ if } \beta_{i-1}=0 
	\end{cases}  \\
	\text{for checking satisfiability } \mathsf{AP}(0) &= \begin{cases} 
		\red{\valf} &\mbox{ if } \sigma = 0 \\ 
		\red{\valt} &\mbox{ if } \sigma = 1 
	\end{cases}
\end{align*}

	For example if $x_6$, with negative sign ($\sigma = 0$), was to be accessed, then the access pattern would be $\mathsf{AP}(3)=\blu{3} \cdot \mathsf{AP}(2)=\blu{3} \cdot \blu{2} \cdot \mathsf{AP}(1)=\blu{3} \cdot \blu{2}  \cdot \mathsf{AP}(0) \cdot \blu{1}$ = $\blu{3} \cdot \blu{2}  \cdot \red{\valf} \cdot \blu{1}$
    ; likewise, for $x_4$, with negative sign ($\sigma = 0$), the access pattern would be  $\mathsf{AP}(3)= \blu{3} \cdot \mathsf{AP}(2)=
    \blu{3} \cdot \mathsf{AP}(1) \cdot \blu{2}=\blu{3} \cdot \mathsf{AP}(0) \cdot \blu{1} \cdot \blu{2}$ = $\blu{3} \cdot \red{\valf} \cdot \blu{1} \cdot \blu{2}$.

    Going back to our example, if $\red{\valf} ~ \blu{1} ~ \red{\valt} ~ \blu{2} ~ \red{\valf} ~ \blu{1} ~ \red{\valt} ~ \blu{3} ~ \red{\valt} ~ \blu{1} ~ \red{\valt} ~ \blu{2} ~ \red{\valf} ~ \blu{1} ~ \red{\valf}$ was written to $g$ by the leader, it is easy to see that the reads with access pattern for $x_6$ ($\blu{3} \cdot \blu{2}  \cdot \red{\valf} \cdot \blu{1}$) would be successful, since $x_6$ had been assigned to \texttt{false} by the leader while that for $x_4$ ( $\blu{3} \cdot \red{\valf} \cdot \blu{1} \cdot \blu{2}$) would fail since it was assigned \texttt{true} while the contributor wished to read $\red{\valf}$.
	$\mathsf{AP}(0)$ is defined to ensure satisfiability of the clause, $\mathsf{AP}(0)=\red{\valf}$ iff $f_{sign}=0$ (the sign of the variable in the clause is negative)
    and $\mathsf{AP}(0)=\red{\valt}$ iff $t_{sign}=0$ (the sign of the variable in the clause is positive).

    The above recursive formulation gives us a poly-sized sub-routine which reads values matching the $\mathsf{AP}$ sequence. We thus have the following lemma.
          
\begin{lemma}
	There exists a sub-routine $\com_{\mathsf{Check}}$, which, starting with a view $\aview$ encoding 
	(in $t_{d_0},f_{d_0}, \dots, t_{d_{n-1}},f_{d_{n-1}}$ and $t_{sign},f_{sign}$)
	the address and sign of boolean variable $x$ in clause $c$, terminates only if  $c$ is satisfied under the assignment to $x$ made by the leader.
\end{lemma}

    
%
%
    Until now, a thread which reads a clause from $\funcCLENC$ has checked its satisfiability with respect to the assignment guessed by the leader, using the $\funcSAT$ module. However, to ensure satisfiability of $\phi$, this check must be done for all $2^n$ clauses.
    This is done in levels $0 \leq i \leq n-2$ using $\funcFC{i}$, exactly as in the \pspace-hardness proof. Finally, we reach $\assert{\texttt{false}}$ 
    reading 1 from both $a_{0,0}, a_{0,1}$. However, in this case we do not have to check for alternation, but only for the universality in the assignments.
   
	\textbf{Forall Checker}: $\funcFC{i}$: The $\funcFC{n-2}$ gadget checks the `universality' with respect to the second-last bit of the clause address,
	$\funcFC{n-3}$ gadget does this check with respect to the third-last bit, and so on, till $\funcFC{0}$ does this check for the first bit, ensuring that all clauses have been covered. 
		
	$2^n$ threads execute $\com_{\funcSAT}$, and write 1 to $a_{n-1,1}$ and $a_{n-1,0}$, depending on the last address bit of the clause it checks. Next, $2^{n-1}$ threads execute $\funcFC{n-2}$. A thread executing $\funcFC{n-2}$ reads 1 from both $a_{n-1,0}$ and $a_{n-1,1}$ representing two clauses whose last bits differ; this thread checks that the second last bits in these two clauses agree: it writes 1 to $a_{n-2,0}$ (if the second last bit is 0) or to $a_{n-2,1}$ (if the second last bit is 1). When $2^{n-1}$ threads finish executing $\funcFC{n-2}$, we have covered the second last bits across all clauses. This continues with $2^{n-2}$ threads executing $\funcFC{n-3}$. A thread executing $\funcFC{n-3}$ reads 1 from both $a_{n-2,0}$ and $a_{n-2,1}$ representing two clauses whose second last bits differ and checks that the third last bits in these two clauses agree. Finally, we have 2 threads executing $\funcFC{0}$,  certifying the universality of the first address bit, writing 1 to $a_{0,0}$ and 
	$a_{0,1}$.

	\textbf{Assertion Checker}: $\funcGC$: The assertion checker gadget in Figure \ref{fig:nexphapp}, reads 1 from $a_{0,0}, a_{0,1}$. If this is successful, then we reach the $\assert{\texttt{false}}$.

\subsubsection{Compare and contrast with \pspace-hardness proof} As is evident there are many things common between the two proofs. We now recapitulate the similarities and differences. 
\begin{itemize}
	\item In the \pspace-h we wanted to check for truth of the QBF, hence guessing an assignment was not necessary. Here the leader is tasked with guessing an assignment to the boolean variables.
	\item In \pspace-h we want to check for quantifier alternation in the boolean variables in $\Psi$. Here we want to also check for universality of addresses, i.e. the fact that all clauses have been checked. This makes the $\funcFC{\_}$ gadget a bit simpler than its $\funcFEC{\_}$ counterpart.
	\item In the \pspace-h, the CNF formula, $\phi$ was given in a simple form and hence all threads executing $\funcSAT$ checked the formula. Additionally, given the exponential size of the formula, the task is distributed between (exponentially) many threads. Here CNF formula also was in an encoded form, hence we had to devise the circuit value machinery to extract it from the succinct representation $D$.
\end{itemize}

\subsection{Correctness of the construction}

The proof of this lemma is very close to that of Lemma \ref{lem:pspace-hard}.  
Some of the terminology we use in this proof is borrowed from the proof 
of Lemma \ref{lem:pspace-hard}. As in the case 
of section \ref{app:pspace-hard}, we add some labels 
in the function descriptions for ease of argument in the proof. 
We describe the notations and key sub lemmas required for the proof.

\paragraph*{Notation and Interpretation of Boolean Variables involved in the construction}
\begin{itemize}
    \item We denote by $\alpha_U$, an assignment on the (boolean) variables $\{u_0, u_1, \cdots u_{n-1}\}$, interpreted as the (n-bit) address of a clause. Here $u_{n-1}$ is the most significant bit (MSB) and $u_0$ as the least significant bit (LSB). We view the assignment so generated as $\overline{\alpha_U} \in \{0,1\}^n$ as an $n$-bit vector. $\alpha_U(u_i)$ gives the 
    assignment to $u_i$. 
    \item We denote by $\alpha_D$, an assignment on (boolean) variables $\{d_0, d_1, \cdots d_{n-1}, d_{sign}\}$, interpreted as the (n-bit) index of a variable in $Vars(\Psi)$ and one sign bit. Here $d_{n-1}$ is the MSB and $d_0$ is the LSB. We view the assignment as $\overline{\alpha_D} \in \{0,1\}^{n+1}$ as an $(n+1)$-bit vector.
    \item For an assignment $\overline{\alpha} \in \mathbb{B}^n$, $D_1(\overline{\alpha_U})$,  (similarly $D_2(\overline{\alpha_U})$ and $D_3(\overline{\alpha_U})$) are the $n+1$ bits signifying $\va(\overline{\alpha_U}), \siga(\overline{\alpha_U})$ ($\vb(\overline{\alpha_U}), \sigb(\overline{\alpha_U})$  and $\vc(\overline{\alpha_U}), \sigc(\overline{\alpha_U})$) respectively.
\end{itemize}

\begin{figure}[h]
\centering
\setlength{\belowcaptionskip}{-7pt}
\setlength{\abovecaptionskip}{3pt}
\begin{empheq}[box={\mycode[sharp corners]}]{align*}
\funcSAT = &~(\assume{s = 1}); ~\com_{\mathsf{CV}}; \lambda_1: \mathsf{skip}; ~\com_{\mathsf{Check}}; \\
		   &~((\assume{(t_{u_{n-1}} = 0)}; ~\assign{a_{n-1,1}}{1}) ~\choice ~(\assume{(f_{u_{n-1}} = 0)}; ~\assign{a_{n-1,0}}{1}))
\end{empheq}

\caption{$\funcSAT$ - acquiring a clause $c_i$ and checking satisfiability of that clause, with the label $\lambda_1$}
\label{fig:d-sim-app}
\end{figure}
    
\begin{figure}[ht]
\centering
\setlength{\belowcaptionskip}{-7pt}
\setlength{\abovecaptionskip}{3pt}
\begin{empheq}[box={\mycode[sharp corners]}]{align*}
\funcFC{i} =~ &\assume{a_{i+1,0} = 1}; ~\assume{a_{i+1,1} = 1}; \\
	& ~((\assume{t_{u_i} = 0}; ~\assign{a_{i,1}}{1}; \lambda_3: \mathsf{skip}) ~\choice ~(\assume{f_{u_i} = 0}; ~\assign{a_{i,0}}{1}); \lambda_4: \mathsf{skip})
\end{empheq}
\caption{$\funcFC{i}$ at level $i$ with the labels $\lambda_3, \lambda_4$.  We have $n-1$ such gadgets, one for each level $0 \leq i \leq n-2$}
\label{fig:check-i-app}
\end{figure}

\subsubsection{Acquiring Variable Index and Sign}
\label{app:b1}
We observe that each thread executing a $\funcCLENC$ function 
  makes a (single) write to $s$, with the message $(s, 1, \view)$ where $\view$ has \textit{embedded} in it an assignment $\alpha_U$. We write $\alpha \diamond \view$ to denote that the assignment $\alpha$ is embedded in $\view$. Now a thread $p$ 
  executing a $\funcSAT$ function acquires the assignment $\alpha_U$, and computes (non-deterministically) one amongst $D_1(\overline{\alpha_U})$, $D_2(\overline{\alpha_U})$, $D_3(\overline{\alpha_U})$ reaching the label $\lambda_1$.  The correctness invariant involved is formalized in the following lemma.

\begin{lemma}
Let a thread $p$ executing the $\funcSAT$ function  read a message $(s, 1, \view)$ with $\alpha\diamond \view$. When the $p$ reaches 
the label $\lambda_1$ computing $D_i$ ($i \in \{1,2,3\}$) with $\alpha_D = D_i(\overline{\alpha_U})$. Let the view of the thread be $\view'$. Then we have $\alpha_D\diamond\view'$.
\label{lem:ccv}
\end{lemma}

\subsubsection{Checking Satisfiability of Clause}
\label{app:b2}
Continuing from above, let $p$ compute $D_i$ in $\com_{\mathsf{CV}}$ 
reaching $\lambda_1$. 
 Then by lemma \ref{lem:ccv}, we have $\alpha_D = D_i(\overline{\alpha})$ embedded in the view of the thread. Now, using $\com_{\mathsf{Check}}$, 
  $p$ checks that the clause with index $\overline{\alpha_U}$ is satisfied by the 
  $n+1$ bits  $\overline{\alpha_D}$ representing a variable and the sign of the variable in the clause.   
   Finally the thread makes writes to one of the program variables $a_{n-1,0}$ or $a_{n-1,1}$. We have the following lemma that shows correctness of this operation.
\begin{lemma}
A thread $p$ can make the write $(a_{n-1,0}, 1)$ (similarly $(a_{n-1,1},1)$) if and only if, clause $\overline{\alpha_U}$ is satisfied and if $\alpha_U(u_{n-1}) = 1$ (similarly $\alpha_U(u_{n-1}) = 0$).
\end{lemma}

\subsubsection{Checking all Clauses}

In section \ref{app:b1} and section \ref{app:b2} we have discussed how the system can check satisfiability of a single clause. Now, we need to check that each clause satisfied. We do this via additional modules to the PSPACE construction.

Towards this goal, define a \emph{level  predicate} $\issat(u_{n-1}, u_{n-2}, \cdots, u_0)$ denoting that the clause $\overline{\alpha_U} = u_{n-1}\cdots u_1 u_0$ is satisfiable.  Now very similarly to section \ref{app:pspace-hard} we define the following formulae:

For $0 \leq i \leq n-1$,
\begin{equation*}
    \Upsilon_i \equiv \forall u_i \forall u_{i+1} \dots\forall u_{n-1} \textcolor{blue}{\exists u_0 \dots \exists u_{i-1}} \issat(u_{n-1}, u_{n-2}, \cdots, u_0)
\end{equation*}

And we claim the following lemma,

\begin{lemma}
For $0\leq i \leq n-2$, $\Upsilon_i$ is true $\iff$ the labels 
$\lambda_3, \lambda_4$
 in the gadget $\funcFC{i}$
 can be reached.
\label{lem:suc}
\end{lemma}
The proof of Lemma \ref{lem:suc} follows exactly in similar lines to that of Lemma \ref{lem:base} and Lemma \ref{lem:induct}. 
Finally, note that $\Upsilon_0$ is equivalent to the $\succsat$ instance being satisfiable. We have the following final lemma to show correctness of the entire construction.

\begin{corollary}
We can reach the $\assert{\texttt{false}}$ assertion in the $\funcGC$ gadget 
 $\iff \Upsilon_0$ is true.
\end{corollary}
This gives us the main theorem 
\begin{theorem}
The verification of safety properties for parametrized systems of the class $\thatthread(\uncassy, \unloopy) \parallel \thisthread(\uncassy)$ under RA 
is \nexp-hard.   
\label{thm:nexp-c}
\end{theorem}

\section{Conclusion}
Atomic CAS operations are indispensible for most practical implementations of distributed protocols, yet, they hinder verification efforts. 
Undecidability of safety verification in the non-parameterized setting \cite{AAAK19} and even in the loop-free parameterized setting $\thatthread(\unloopy)$, are a testament to this.

We tried to reconcile the two by studying the effects of allowing restricted access to CAS operations in parameterized systems. Systems which prevent the $\thatthread$ threads from performing CAS operations and allow only either (1) loop-free $\thisthread$ programs or (2) loop-free $\thisthread$ programs along with a single (`ego')
program with loops lead to accessible complexity bounds. The simplified semantics based on a timestamp abstraction provides the infrastructure for these results. The \pspace-hardness gives an insight into the core complexity of RA ($\mathsf{PureRA}$) that stems from the consistency mechanisms of view-joining and timestamp comparisons. 

We conclude with some interesting avenues for future work.  A problem arising from this work is the decidability of CAS-free parameterized systems, $\thatthread(\uncassy)||\thisthread_1(\uncassy) \parallel \cdots \parallel \thisthread_n(\uncassy)$ which seems to be as elusive as its non-parameterized twin $\thisthread_1(\uncassy) \parallel \cdots \parallel \thisthread_n(\uncassy)$. We believe that ideas in this paper can be adapted to causally consistent shared memory models \cite{LahavBoker20} as well as transactional programs 
\cite{DBLP:conf/concur/BeillahiBE19} in the parameterized setting. On, the practical side, the Datalog encoding suggests development of a tool, considering that Horn-clause solvers are state-of-the-art in program verification.

\bibliography{references}

\end{document}